\documentclass[sort&compress,AMA,STIX1COL]{WileyNJD-v2}

\usepackage{cancel}

\articletype{Research Article}%

\received{26 March 2022}
\revised{- June 2022}
\accepted{- June 2022}

\begin{document}

\title{Exponentially Stable Adaptive Optimal Control of Uncertain LTI Systems\protect}

\author[1]{Anton Glushchenko*}

\author[1]{Konstantin Lastochkin}

\authormark{Glushchenko A. \textsc{et al}}

\address[1]{\orgdiv{Ya.Z. Tsypkin Laboratory of Adaptive and Robust Systems}, \orgname{V.A. Trapeznikov Institute of Control Sciences of Russian Academy of Sciences}, \orgaddress{\state{Moscow}, \country{Russia}}}

\corres{*Anton Glushchenko, Russia, Moscow, Profsoyuznaya street, h.65, 117997, ICS RAS. \email{aiglush@ipu.ru}}

\abstract[Summary]{A novel method of an adaptive linear quadratic (LQ) regulation of uncertain continuous linear time-invariant systems is proposed. Such an approach is based on the direct self-tuning regulators design framework and the exponentially stable adaptive control technique developed earlier by the authors. Unlike the known solutions, a procedure is proposed to obtain a non-overparametrized regression equation (RE) with respect to the unknown controller parameters from an initial RE of the LQ-based reference tracking control system. On the basis of such result, an adaptive law is proposed, which under mild regressor finite excitation condition provides monotonous convergence of the LQ-controller parameters to an adjustable set of their true values, which bound is defined only by the machine precision. Using the Lyapunov-based analysis, it is proved that the mentioned law guarantees the exponential stability of the closed-loop adaptive optimal control system. The simulation examples are provided to validate the theoretical contributions.}


\jnlcitation{\cname{%
\author{Glushchenko A.}, and
\author{K. Lastochkin}} (\cyear{2023}), 
\ctitle{Exponentially Stable Adaptive Optimal Control of Uncertain LTI Systems}, \cjournal{International Journal of Adaptive Control and Signal Processing}, \cvol{2023;00:1--26}.}

\maketitle

\section{Introduction}\label{sec1}

The model reference adaptive control (MRAC) is an effective framework to solve a great number of applied control problems with parameter uncertainty \cite{Sastry11, Lavretsky13,Ioannou12,Aranovskiy21, Nguyen18}. However, the main assumption of the classical MRAC paradigm is that the Erzberger conditions \cite{Erzberger67}, i.e. the {\it matching conditions} between the reference and plant models, are satisfied. But in the many practical scenarios such requirement is not always met \cite{Nguyen18}. This fact narrows the applicability domain of the conventional MRAC approaches. Moreover, both classical and modern methods of the adaptive control provide stability of the reference model tracking error, but generally do not guarantee the required time-domain or integral performance indexes \cite{Lavretsky13}.

On the other hand, the methods of optimal control theory \cite{Lewis95,Kwakernaak72}: (1) do not require to satisfy strict matching conditions, (2) allow one to synthesize the optimal control law with respect to the chosen cost function, (3) provide certain guarantees on the control quality in the robust sense. The classical solution of the optimal control problem is a state feedback controller minimizing the linear-quadratic cost function, which depends on the states and control signals, over the infinite time interval \cite{Lewis95,Kwakernaak72}. Considering the linear systems, the parameters of such a controller are calculated via solution of the well-known matrix algebraic Riccati equation (ARE). The solution of ARE requires {\it a priori} knowledge of the plant parameters, and it is usually found using offline analytical methods \cite{Kwakernaak72,Kalman60,Laub79,Vaughan69} or numerical procedures \cite{Kwakernaak72,Kleinman68,Arnold84}.

Considering the above-mentioned advantages of the optimal control theory methods, their application to the plants with parameter uncertainty \cite{Polyak05} to overcome the drawbacks of the conventional MRAC is an extremely actual problem. This is confirmed by a significant number of studies \cite{Sutton18,Vrabie09,Jiang12,Vamvoudakis10,Bhasin13,Modares13,Vamvoudakis15,Modares14,Lee14,Palanisamy14,Jha18a,Jha18b,Jha19a,Jha19b,Kiumarsi17,Jha17,Jha18c,Jha19c} aimed at the synthesis of an adaptive linear-quadratic (ALQ) regulators via solution of ARE when the plants parameters are unknown.

In general, known mathematically sound methods of ALQ synthesis can be divided into two broad groups:
\begin{enumerate}
    \item[1)] ALQ based {\it reinforcement learning} (RL);
    \item[2)] ALQ based {\it self-tuning regulators} (STR).
\end{enumerate}

Without pretending to provide an exhaustive review, the main studies, which belong to the mentioned groups, are considered below.

\textbf{ALQ-RL}. Approaches from this group exploit the {\it policy iteration} (PI) methodology consisting of the policy evolution ({\it critic update}) step, in which the value of the cost function is calculated using the controller current parameters, and the policy improvement ({\it actor update}) step, in which the controller parameters are adjusted towards the direction of the cost function decrease \cite{Sutton18}. This methodology has become applicable to the ALQ control systems synthesis on the ground of the disseminating study of Kleinman \cite{Kleinman68}, in which a computationally effective fast convergent discrete algorithm of online solution of the Riccati equation for plants with completely known dynamics is proposed. Later, based on the Kleinman algorithm and PI, the ALQ controllers for plants with partially \cite{Vrabie09} and completely \cite{Jiang12} unknown dynamics were proposed. The disadvantages of these solutions are that: (1) their stability analysis is rather complex \cite{Liberzon03} as the algorithm to adjust the control law parameters is discrete, while the plant is continuous, (2) some initial stabilizing controller is required to be known, (3) the ARE solution and controller parameters estimates converge to their true values only if the regressor persistent excitation (PE) condition is met. To overcome the first and second drawbacks, continuous time PI-based schemes are proposed, which do not require an initial stabilizing policy \cite{Vamvoudakis10,Bhasin13,Modares13}, but the PE condition is still a must for convergence. In \cite{Vamvoudakis15,Modares14,Lee14,Palanisamy14}, various approaches to relax the PE requirement, which are based on the integral sliding-window filtering and methods of intelligent data storage on the input-output characteristics of the plant, are proposed for the PI framework. The computational complexity of some of the above-mentioned solutions turns out to be rather high, especially when they are applied to control large scale systems. Other approaches \cite{Jha18a,Jha18b,Jha19a,Jha19b}, which, according to their authors statements, have less computational complexity, require an initial stabilizing policy and could be reduced to the discrete Kleinman algorithm. However, in \cite{Jha18a,Jha18b,Jha19a,Jha19b}, in contrast to \cite{Vamvoudakis15,Modares14,Lee14,Palanisamy14}, the stability of the closed-loop system for the whole set of switching states is proved on the basis of the switching systems stability theory \cite{Liberzon03}. In general, a significant drawback of ALQ-RL schemes is that the quality of the plant state transients is often unacceptable in the course of the iterative discrete adjustment of the control law parameters. A more detailed review of ALQ-RL methods can be found in \cite{Kiumarsi17}.

\textbf{ALQ-STR}. Approaches from this group are directly based on the methods of identification theory and adaptive control \cite{Sastry11,Lavretsky13,Ioannou12,Aranovskiy21,Nguyen18} and assume either a direct continuous adjustment of the optimal controller or require identification of the plant model parameters and the elements of the Riccati equation solution. The main problem of ALQ synthesis using STR methods is the necessity to obtain: (1) the regression equation (RE) with respect to the measurable signals and the unknown control law parameters and/or ARE solution, (2) the adaptive law, which guarantees the system stability. In \cite{Ioannou12} it is proposed to identify the parameters of the plant model and, using such estimates and the {\it certainty equivalence principle}, to solve the ARE continuously. The disadvantage of this approach is the complexity of the analytical solution of the ARE at each time instant, as well as that the convergence of the LQ-controller parameters adjustment process is ensured only when the restrictive PE condition is satisfied. In order to obtain a simpler solution, which could be implemented on low-cost embedded hardware, several methods are proposed in \cite{Jha17,Jha18c,Jha19c} to reduce the ARE to a measurable RE, in which the unknown parameters are the ARE solution and/or the matrix of the optimal controller coefficients. The RE obtained by such parameterization allows one to apply both classical \cite{Sastry11} and modern \cite{Chowdhary13} algorithms to identify its parameters directly. In \cite{Jha17,Jha18c}, concurrent-learning-based estimation laws \cite{Chowdhary13} are introduced on the basis of the obtained RE, which guarantee exponential convergence of the identified parameters to a bounded set of their true values if the regressor initial excitation condition (IE) is met. In \cite{Jha19c}, a gradient descent method is used to derive a law that ensures a local exponential convergence of the unknown parameters estimates to zero when the PE condition is satisfied. However, firstly, the obtained regressions \cite{Jha17,Jha18c,Jha19c} are substantially overparameterized, which causes well-known drawbacks \cite{Sastry11,Aranovskiy21}, and secondly, the authors do not provide a strict formal discussion on the conditions under which the regressor of the obtained RE satisfies the IE or PE requirements (in fact, the assumption that these conditions are satisfied is not sufficiently justified in these studies). Moreover, in \cite{Jha18c}, the excitation of the regressor, in principle, exists due to the fact that the quadratic term of ARE is considered as a disturbance (see (12)), while in \cite{Jha17} it is provided as the plant parameters estimates are continuously substituted into the RE (see (25)). All these facts demonstrate the artificial nature of the regressor excitation due to the perturbations, which are introduced into RE, at the expense of the obtained estimation accuracy. Therefore, despite the satisfactory results of the numerical experiments presented in \cite{Jha17,Jha18c,Jha19c}, the feasibility of the discussed algorithms from the point of view of the regressor excitation, as well as the general capability to adjust their accuracy are open questions. In \cite{Jha18c},  Remark 2, it is noted that the artificial excitation of the regressor can be obtained by addition of noise or test signals to the control action. This will certainly allow one to meet the IE or PE requirements almost everywhere, but in practice such operation is often infeasible and even dangerous.

The present study relates to the ALQ-STR group, and, compared to the considered solutions \cite{Ioannou12,Jha17,Jha18c,Jha19c}, its main contributions can be formulated as follows.
\begin{enumerate}
    \item[\textbf{C1}] A regression equation without overparameterization is obtained, which constitutes of the signals measured from the plant and the unknown matrix of the optimal control law parameters. {Such equation allows one to identify controller parameters directly without consecutive substitutions of intermediate signals estimates from one RE to another.}
    \item[\textbf{C2}] The finite excitation (FE) of the initial regressor (the vector of signals, which are measured from the plant) guarantees that the regressor of the parameterization C1 is also FE.
    \item[\textbf{C3}] The adaptive law is proposed, which, if the FE condition is satisfied for the initial regressor in C2, ensures monotonic convergence of the elements of the optimal control law matrix to a set of their true values, which bound is defined only by the computational power of the hardware-in-use.
    \item[\textbf{C4}] The global exponential stability of the plant with the developed adaptive self-tuning regulator C3 is ensured.
\end{enumerate}

The above contributions constitute the novelty of the proposed approach and fully overcome the shortcomings of the known solutions, which are noted in the review and relate to both ALQ-RL and ALQ-STR.

The main result of the study, which allows us to obtain C1-C4, is based on the direct self-tuning regulators concept \cite{Ioannou12} and the method of the exponentially stable adaptive control design \cite{Glushchenko22a}. According to them, in order to synthesize a state feedback adaptive control system, a RE with a scalar regressor with respect to the unknown ideal parameters of the control law is parameterized. The key aspect in this case is to transfer the excitation of the initial vector of measurable signals (the plant states, control action, reference etc.) to the scalar regressor of the resulting RE. Then an adaptive law is derived, which guarantees global exponential stability of the closed-loop control system. In contrast to classical adaptive control, this approach does not need a measurable reference model, but it requires analytical equations to calculate the ideal parameters of the control law. The concept of STR-based adaptive control with guarantee of global exponential stability under FE condition has been used for the first time by the authors to solve classical problems of the state and output feedback adaptive control for SISO \cite{Glushchenko22a} and MIMO \cite{Glushchenko22b} systems. It was later applied to derive an adaptive pole-placement control system for SISO plants \cite{Glushchenko22c}. In this study, the above-mentioned concept is applied to develop an adaptive suboptimal control system under FE condition, {and, like in all previous above-mentioned cases, the major challenge of such extension is to propose parametrization to derive RE mentioned in C1}.

The remainder of this paper is organized as follows. A rigorous problem statement is presented in Section II. The main result is reported in Section III, and Section IV validates the proposed approach in numerical simulations.

\textbf{Notation.} $\left|.\right|$ is the absolute value, $\left|\left|.\right|\right|$ is the Euclidean norm of a vector, $\lambda_{min}\left(.\right)$ and $\lambda_{max}\left(.\right)$ are the matrix minimum and maximum eigenvalues respectively, ${I_{n \times n}}$ is an identity $n \times n$ matrix, ${0_{n \times n}}$ is a zero $n \times n$ matrix, ${1_{n \times n}}$ is a $n \times n$ matrix, which is filled with ones, $det\{.\}$ stands for a matrix determinant, $adj\{.\}$ represents an adjoint matrix, ${L_\infty }$ is the space of all essentially bounded functions, $o\left(.\right)$ means “is ultimately smaller than”. The following definition from \cite{Sastry11} is used throughout the paper. {$O(.)$ stands for the 'Big $O$' notation to estimate complexity of algorithms.}

\begin{definition}\label{definition1}
The regressor $\varphi \left( t \right)$ is finitely exciting $\left( {\varphi \left( t \right) \in {\rm{FE}}} \right)$ over $\left[ {t_r^ + {\rm{;}}\;{t_e}} \right]$ if there exist ${t_e} \ge t_r^ +  > 0$ and $\alpha  > 0$ such that the following holds:
\begin{equation}\label{eq1}
    \int\limits_{t_r^ + }^{{t_e}} {\varphi \left( \tau  \right){\varphi ^{\rm{T}}}\left( \tau  \right)d} \tau  \ge \alpha {I_{n \times n}},
\end{equation}
where $\alpha$ is the excitation level.
\end{definition}

The corollary of the Kalman-Yakubovich-Popov lemma is introduced (see \cite{Tao03}, Theorem 9.4).

\begin{corollary}\label{corollary1}
For any matrix $D > 0$, a Hurwitz matrix $A \in {\mathbb{R}^{n \times n}}$, a scalar $\mu  > 0$, a matrix $B \in {\mathbb{R}^{n \times m}}$ such that a pair $\left(A, B\right)$ is controllable there exists a matrix $M = {M^{\rm{T}}} > 0$ and matrices $N \in {\mathbb{R}^{n \times m}},\;K \in {\mathbb{R}^{m \times m}}$ such that:
\begin{equation}\label{eq2}
\begin{array}{c}
{A^{\rm{T}}}M + MA =  - N{N^{\rm{T}}} - \mu M{\rm{,}}\;MB = NK{\rm{,}\;}\\
{K^{\rm{T}}}K = D + {D^{\rm{T}}}.
\end{array}
\end{equation}
\end{corollary}

\section{Problem statement}\label{sec2}

The optimal control problem for a class of linear time-invariant plants with parameter uncertainty is considered:
\begin{equation}\label{eq3}
    \forall t \ge t_r^ + {\rm{,}\;}\left\{ \begin{array}{l}
{{\dot x}_p}\left( t \right) = {A_p}{x_p}\left( t \right) + {B_p}u\left( t \right){\rm{,}\;}{x_p}\left( {t_r^ + } \right) = {x_{p0}}{\rm{,}}\\
z\left( t \right) = C_p^{\rm{T}}{x_p}\left( t \right){\rm{,}}
\end{array} \right.
\end{equation}
where ${x_p}\left( t \right) \in {\mathbb{R}^{{n_p}}}$ is a measurable plant state vector with an unknown initial conditions ${x_{p0}}$, $u\left( t \right) \in {\mathbb{R}^m}$ is the control vector, $z\left( t \right) \in {\mathbb{R}^m}$ is the performance output, ${A_p} \in {\mathbb{R}^{{n_p} \times {n_p}}}$ is an unknown state matrix, ${B_p} \in {\mathbb{R}^{{n_p} \times m}}$ is an unknown input matrix of the plant, ${C_p} \in {\mathbb{R}^{{n_p} \times m}}$ is a known output matrix to form the output $z\left( t \right)$, the pair $\left( {{A_p}{\rm{,}\;}{B_p}} \right)$  is fully controllable.

The objective is to derive an optimal control law that allows the output $z\left( t \right)$ to track a given bounded piecewise constant command $r \in {L_\infty } \cap {\mathbb{R}^m}$. For this purpose, a weighted integral error state ${e_{{z_I}}}\left( t \right) = \vartheta \int\limits_{t_r^ + }^t {\left( {z\left( \tau  \right) - r} \right)d\tau } $ ($\vartheta  \ne 0$) is introduced to denote the integrated output tracking error. The system \eqref{eq3} is augmented with ${e_{{z_I}}}\left( t \right)$ to yield the extended system as follows:
\begin{equation}\label{eq4}
\forall t \ge t_r^ + {\rm{,}\;}\left\{ \begin{array}{l}
\dot x\left( t \right) = \theta _{AB}^{\rm{T}}\Psi \left( t \right) - {B_r}r = Ax\left( t \right) + Bu\left( t \right) - {B_r}r{\rm{,}}\;x\left( {t_r^ + } \right) = {x_0}{\rm{,}}\\
z\left( t \right) = {C^{\rm{T}}}x\left( t \right){\rm{,}}
\end{array} \right.
\end{equation}
where $x\left( t \right) = {{\begin{bmatrix}{{x_p}\left( t \right)}&{{e_{{z_I}}}\left( t \right)}
\end{bmatrix}}^{\rm{T}}} \in {\mathbb{R}^n}$ is an augmented state vector $\left( {n = {n_p} + m} \right)$,
\begin{equation}\label{eq5}
\begin{array}{c}
A = {\begin{bmatrix}
{{A_p}}&{{0_{{n_p} \times m}}}\\
{\vartheta C_p^{\rm{T}}}&{{0_{m \times m}}}
\end{bmatrix}} \in {\mathbb{R}^{n \times n}}{\rm{,}}\;B = {\begin{bmatrix}
{{B_p}}\\
{{0_{m \times m}}}
\end{bmatrix}} \in {\mathbb{R}^{n \times m}}{\rm{,}}\;{B_r} = {\begin{bmatrix}
{{0_{{n_p} \times m}}}\\
{\vartheta {I_{m \times m}}}
\end{bmatrix}} \in {\mathbb{R}^{n \times m}}{\rm{,}}\\
{C^{\rm{T}}} = {\begin{bmatrix}
{C_p^{\rm{T}}}&{{0_{m \times m}}}
\end{bmatrix}} \in {\mathbb{R}^{m \times n}}{\rm{,}}\;\Psi \left( t \right) = { {\begin{bmatrix}
{{x^{\rm{T}}}\left( t \right)}&{{u^{\rm{T}}}\left( t \right)}
\end{bmatrix}}^{\rm{T}}}{\rm{,}}\;\theta _{AB}^{\rm{T}} ={\begin{bmatrix}
A&B
\end{bmatrix}}.
\end{array}
\end{equation}

The vector $\Psi \left( t \right) \in {\mathbb{R}^{n + m}}$ is considered to be measurable at each time instant $t > t_r^ + $, whereas the vector ${\theta _{AB}} \in {\mathbb{R}^{\left( {n + m} \right) \times n}}$ is time-invariant and unknown. The pair $\left( {A{\rm{,}}\;B} \right)$ is controllable if $\left( {{A_p}{\rm{,}}\;{B_p}} \right)$ is controllable and $det \left\{ {{\begin{bmatrix}
{{A_p}}&{{B_p}}\\
{C_p^{\rm{T}}}&{{0_{m \times m}}}
\end{bmatrix}}} \right\} \ne 0.$

Considering the plant \eqref{eq4}, the goal is to derive a control law that minimizes the cost function:
\begin{equation}\label{eq6}
J = \frac{1}{2}\int\limits_{t_r^ + }^\infty  {{x^{\rm{T}}}\left( t \right)Qx\left( t \right) + {u^{\rm{T}}}\left( t \right)Ru\left( t \right)} {\rm{}\;}dt{\rm{,}}
\end{equation}
where $Q = {Q^{\rm{T}}} \in {L_\infty } \cap {\mathbb{R}^{n \times n}}$ and $R = {R^{\rm{T}}} \in {L_\infty } \cap {\mathbb{R}^{m \times m}}$ are the positive definite weight matrices related to the plant states and control signal respectively, the pair $\left( {A{\rm{,}}\;{Q^{{\textstyle{\frac{1}{2}}}}}} \right)$ is detectable.

The following well-known result of the optimal control theory \cite{Kalman60}, which allows one to solve the optimization problem \eqref{eq6} when the matrices $A$ and $B$ are known, will be used throughout the paper.

\begin{proposition}\label{proposition1}
Let $\left( {A{\rm{,}}\;B} \right)$ be fully controllable, and $\left( {A{\rm{,}}\;{Q^{{\textstyle{\frac{1}{2}}}}}} \right)$ be detectable, then the control law, which minimizes \eqref{eq6}, is written as:
\begin{equation}\label{eq7}
\begin{array}{c}
{u^*}\left( t \right) = {\theta ^{\rm{T}}}\omega \left( t \right) = {K_x}x\left( t \right) + {K_r}r =  - {R^{ - 1}}{B^{\rm{T}}}Px\left( t \right) - {R^{ - 1}}{B^{\rm{T}}}Vr{\rm{,}}\\
{\theta ^{\rm{T}}}{\rm{ = }}{\begin{bmatrix}
{{K_x}}&{{K_r}}
\end{bmatrix}} \in {\mathbb{R}^{m \times \left( {n + m} \right)}}{\rm{,}\;}\omega \left( t \right) = {{\begin{bmatrix}
{{x^{\rm{T}}}\left( t \right)}&{{r^{\rm{T}}}}
\end{bmatrix}}^{\rm{T}}} \in {\mathbb{R}^{n + m}}{\rm{,}}
\end{array}
\end{equation}
where ${K_x} \in {\mathbb{R}^{m \times n}}{\rm{,}\;}{K_r} \in {\mathbb{R}^{m \times m}}$, $P = {P^{\rm{T}}} \in {\mathbb{R}^{n \times n}}$ and $V \in {\mathbb{R}^{n \times m}}$ are steady-state solutions of the following differential equations:
\begin{equation}\label{eq8}
\begin{array}{l}
 - \dot P\left( t \right) = {A^{\rm{T}}}P\left( t \right) + P\left( t \right)A - P\left( t \right)B{R^{ - 1}}{B^{\rm{T}}}P\left( t \right) + Q{\rm{,}}\\
\dot V\left( t \right) =  - {A^{\rm{T}}}V\left( t \right) + P\left( t \right)B{R^{ - 1}}{B^{\rm{T}}}V\left( t \right) + P\left( t \right){B_r}{\rm{,}}
\end{array}
\end{equation}
and, if $\forall i{\rm{,}\;}j \in \overline {1,2n} $ it holds that ${\lambda _i}\left( D \right)\cancel{ \gg }{\lambda _j}\left( D \right)$, then such solutions can be obtained analytically in the following way:
\begin{equation}\label{eq9}
\begin{array}{c}
P = P\left( {{\tau _\infty }} \right){\rm{\;:}} = {\Phi _{21}}\left( {{\tau _\infty }} \right)\Phi _{11}^{ - 1}\left( {{\tau _\infty }} \right){\rm{,}}\\
\Phi \left( {{\tau _\infty }} \right){\rm{:}} = {e^{D{\tau _\infty }}} = {\begin{bmatrix}
{{\Phi _{11}}\left( {{\tau _\infty }} \right)}&{{\Phi _{12}}\left( {{\tau _\infty }} \right)}\\
{{\Phi _{21}}\left( {{\tau _\infty }} \right)}&{{\Phi _{22}}\left( {{\tau _\infty }} \right)}
\end{bmatrix}}{\rm{,}}\;D = {\begin{bmatrix}
{ - A}&{B{R^{ - 1}}{B^{\rm{T}}}}\\
Q&{{A^{\rm{T}}}}
\end{bmatrix}}{\rm{,}}\\
V = {\left( {{A^{\rm{T}}} - PB{R^{ - 1}}{B^{\rm{T}}}} \right)^{ - 1}}P{B_r}{\rm{,}}
\end{array}
\end{equation}
where $\Phi \left( {{\tau _\infty }} \right) \in {\mathbb{R}^{2n \times 2n}}$, ${\Phi _{ij}} \in {\mathbb{R}^{n \times n}}$, ${\tau _\infty } > 0,\;{\rm{}}{\tau _\infty } \to \infty $ is an auxiliary constant, which value is sufficiently high.
\end{proposition}

\begin{proof}
The correctness of \eqref{eq7} and \eqref{eq8} is shown in \cite{Nguyen18}, p. 244-246, the equations \eqref{eq9} are obtained as a combination of the results from \cite{Nguyen18}, p.244-246 and \cite{Kalman60}, the necessity of the condition $\forall i{\rm{,}}j \in \overline {1,2n} $ ${\lambda _i}\left( D \right)\cancel{ \gg }{\lambda _j}\left( D \right)$ for the existence of analytical solutions \eqref{eq9} is proved in \cite{Vaughan69} and discussed in detail in \cite{Kwakernaak72}.
\end{proof}

The optimal control law \eqref{eq7} allows one to define the required control quality for the plant \eqref{eq4} as follows:
\begin{equation}\label{eq10}
{\dot x_{ref}}\left( t \right) = {A_{ref}}{x_{ref}}\left( t \right) + {B_{ref}}r = \left( {A + B{K_x}} \right){x_{ref}}\left( t \right) +  B{K_r}r{\rm{,\;}}{x_{ref}}\left( {t_r^ + } \right) = {x_0}{\rm{,}}
\end{equation}
where ${x_{ref}}\left( t \right) \in {\mathbb{R}^n}$ is an unmeasurable state vector of the reference model with initial conditions ${x_0}$, ${A_{ref}} \in {\mathbb{R}^{n \times n}}$ is an unknown desirable closed-loop state matrix.

As the matrices $A$ and $B$ are unknown, both the optimal control law \eqref{eq7} and the reference model \eqref{eq10} could not be implemented, so the control law with adjustable parameters is introduced:
\begin{equation}\label{eq11}
u\left( t \right) = {\hat \theta ^{\rm{T}}}\left( t \right)\omega \left( t \right) = {\hat K_x}\left( t \right)x\left( t \right){\rm{ + }}{\hat K_r}\left( t \right)r{\rm{,}}
\end{equation}
where ${\hat K_x}\left( t \right) \in {\mathbb{R}^{m \times n}}{\rm{,}}\;{\hat K_r}\left( t \right) \in {\mathbb{R}^{m \times m}}{\rm{,}}\; \hat \theta \left( t \right) \in {\mathbb{R}^{\left( {n + m} \right) \times m}}$ are the adjustable parameters.

The equation \eqref{eq11} is substituted into \eqref{eq4}, then \eqref{eq10} is subtracted from the obtained result to write the error equation:
\begin{equation}\label{eq12}
\begin{array}{c}
{{\dot e}_{ref}}\left( t \right) = Ax\left( t \right) + Bu\left( t \right) - {A_{ref}}{x_{ref}}\left( t \right) - B{K_r}r \pm B{K_x}x\left( t \right) = \\
 = {A_{ref}}{e_{ref}}\left( t \right) + B\left( {u\left( t \right) - {K_x}x\left( t \right) - {K_r}r} \right) = {A_{ref}}{e_{ref}}\left( t \right) + B{{\tilde \theta }^{\rm{T}}}\left( t \right)\omega \left( t \right){\rm{,}}
\end{array}
\end{equation}
where ${e_{ref}}\left( t \right) = x\left( t \right) - {x_{ref}}\left( t \right)$ is unmeasurable tracking error calculated as the difference between the states of the plant \eqref{eq4} and the reference model \eqref{eq10}, $\tilde \theta \left( t \right) = \hat \theta \left( t \right) - \theta $ is the parameter error, i.e. the error of identification of the optimal control law \eqref{eq7} parameters.

Then we are in position to formulate the problem of the adaptive control, which is suboptimal in terms of criterion \eqref{eq6}, for the plant \eqref{eq4}.

\emph{Goal.} We consider the system \eqref{eq4} with unknown $A{\rm{,}}\;B{\rm{,}}\;x\left( {t_r^ + } \right)$ and the optimization cost function \eqref{eq6} with known matrices $Q{\rm{,}}\;R$ and parameters $\vartheta {\rm{,}}\;{\tau _\infty }$, which values together guarantee the existence of the solutions \eqref{eq9}. Let $\theta $ be defined as in \eqref{eq7}, where $\theta $ affords a minimum to \eqref{eq6}. Then the goal is to derive an update law for $\hat \theta \left( t \right)$ such that, when $\Psi \left( t \right)$ from \eqref{eq5} is FE, the control law \eqref{eq11} stabilizes the system \eqref{eq4} in a suboptimal way in terms of \eqref{eq6}, and ensures exponential convergence:
\begin{equation}\label{eq13}
\mathop {{\rm{lim}}}\limits_{t \to \infty } \left\| {\xi \left( t \right)} \right\| \le {{\rm{\varepsilon }}_\xi }{\rm{}}\;\left( {\exp } \right){\rm{,}}
\end{equation}
where $\xi \left( t \right) = {{\begin{bmatrix}
{e_{ref}^{\rm{T}}\left( t \right)}&{ve{c^{\rm{T}}}\left( {\tilde \theta \left( t \right)} \right)}
\end{bmatrix}} ^{\rm{T}}} \in {\mathbb{R}^{n + m\left( {n + m} \right)}}{\rm{,}}$ and ${{\rm{\varepsilon }}_\xi }$ value can be reduced to zero by appropriate choice of the arbitrary parameters values of the proposed method.

\begin{remark}\label{remark1}
Contrary to the conventional model reference adaptive control schemes \cite{Ioannou12,Nguyen18}, the reference model \eqref{eq10} can not be implemented, and the state vector ${x_{ref}}\left( t \right)$ is unmeasurable.
\end{remark}

\begin{remark}\label{remark2}
The control law \eqref{eq11} with the adaptive law to adjust the parameters $\hat \theta \left( t \right)$ is suboptimal in terms of criterion \eqref{eq6} as the adjustable law \eqref{eq11} coincides with the optimal one \eqref{eq7} only if ${{\rm{\varepsilon }}_\xi } = 0,{\rm{}}\;t \to \infty $. The difference between the minimum of \eqref{eq6} when the ideal control law \eqref{eq7} is applied and the value of \eqref{eq6} when the adjustable law \eqref{eq11} is used is written as:
\begin{equation}\label{eq14}
\begin{array}{c}
\tilde J = \hat J - J = \frac{1}{2}\int\limits_{t_r^ + }^\infty  {{x^{\rm{T}}}\left( t \right)Qx\left( t \right) + {u^{\rm{T}}}\left( t \right)Ru\left( t \right)} {\rm{}\;}dt
-\frac{1}{2}\int\limits_{t_r^ + }^\infty  {x_{ref}^{\rm{T}}\left( t \right)Q{x_{ref}}\left( t \right) + {{\left( {{u^*}}\left( t \right) \right)}^{\rm{T}}}R{u^*}\left( t \right)} {\rm{}\;}dt = \\
 = \frac{1}{2}\int\limits_{{t_r^+}}^\infty  {e_{ref}^{\rm{T}}\left( t \right)Q{e_{ref}}\left( t \right) + e_{ref}^{\rm{T}}\left( t \right)Q{x_{ref}}\left( t \right) + x_{ref}^{\rm{T}}\left( t \right)Q{e_{ref}}\left( t \right)} {\rm{}\;}dt + \\
 + \frac{1}{2}\int\limits_{t_r^ + }^\infty  {{\omega ^{\rm{T}}}\left( t \right)\tilde \theta \left( t \right)R{{\tilde \theta }^{\rm{T}}}\left( t \right)\omega \left( t \right) + {\omega ^{\rm{T}}}\left( t \right)\tilde \theta \left( t \right)R{\theta ^{\rm{T}}}\omega \left( t \right) + {\omega ^{\rm{T}}}\left( t \right)\theta R{{\tilde \theta }^{\rm{T}}}\left( t \right)\omega \left( t \right)} {\rm{}\;}dt
\end{array}
\end{equation}
and, in the process of the stated problem solution, it can be reduced if the rate of convergence of errors ${e_{ref}}$ and $\tilde \theta \left( t \right)$ to zero is improved. Details on how to obtain \eqref{eq14} are given in the supplementary material \cite{Glushchenko22d}.
\end{remark}

\begin{remark}\label{remark3} 
If $\exists i{\rm{,}}\;j \in \overline {1,2n}$ ${\lambda _i}\left( D \right) \gg {\lambda _j}\left( D \right)$, then, according to the results of \cite{Vaughan69}, the steady-state solutions $\left( {{\tau _\infty } \to \infty } \right)$ of the differential equations \eqref{eq8} cannot be found using \eqref{eq9}. However, in this case, when ${\tau _\infty } = {t_f} > t_r^ + $ has a sufficiently low value, the matrix $P\left( {{t_f}} \right)$ can be calculated. As a result, a stabilizing control law for the system \eqref{eq4} can be obtained in the form of \eqref{eq7}. In such case $P\left( {{t_f}} \right)$ is the boundary condition of the linear-quadratic cost function optimization problem over a time interval $\left[ {t_r^ + {\rm{;}}\;{t_f}} \right]$:
\begin{equation}\label{eq15}
J = \frac{1}{2}\left[ {{x^{\rm{T}}}\left( t \right)P\left( {{t_f}} \right)x\left( t \right) + \int\limits_{t_r^ + }^{{t_f}} {{x^{\rm{T}}}\left( t \right)Qx\left( t \right) + {u^{\rm{T}}}\left( t \right)Ru\left( t \right)} {\rm{}\;}dt} \right].
\end{equation}

Thus, in the general case, the choice of a constant ${\tau _\infty }$ value should at least guarantee the existence of equations \eqref{eq9} solutions.

The choice of matrices $Q{\rm{,}}\;R$ and the parameter $\vartheta $ makes it possible to adjust the rate of convergence of the Riccati differential equation solution $P\left( t \right)$ to the steady-state value $P\left( {t \to \infty } \right) = P$, and hence to ensure that the equality $P = P\left( {{t_f}} \right)$ holds for finite values of ${t_f}$. In this connection, considering the problem statement \eqref{eq13}, it is assumed that $Q{\rm{,}}\;R$ and $\vartheta $ have already been chosen in such a way that the analytical solution $P\left( {{\tau _\infty }} \right)$ \eqref{eq9} exists and coincides with the static value of the Riccati equation \eqref{eq8} solution. For more details on the difficulties, which arise when ${\lambda _i}\left( D \right) \gg {\lambda _j}\left( D \right)$ and ${\tau _\infty } \to \infty $, see \cite{Kwakernaak72,Vaughan69}.
\end{remark}

{\begin{remark}\label{remark3.5} 
As it is mentioned in Introduction, a lot of efforts \cite{Jha17,Jha18c,Jha19c} have been made to solve the stated problem by considering \eqref{eq8} to obtain estimate of $P$ and substitute it into equations to calculate $K_x$ and $K_r$. But they faced problems related to the regressor excitation propagation and overparametrization. In this study it is proposed to use ALQ-STR design and include equation (9) in it to solve the Riccati equation. The challenge is that $A$ and $B$ in $D$ are unknown, but the concept of STR-based adaptive control with guarantee of global exponential stability under FE condition \cite{Glushchenko22a} will give the opportunity to substitute $A$ and $B$ in $D$ with measurable signals. And then in several steps to obtain a RE with respect to $\theta$. The main difficulty is to derive the parametrization of the problem under consideration, which allow one to use advantages of the above-mentioned concept. That is one of the main scopes of the following section.
\end{remark}}

\section{Main result}\label{sec3}

The main result of the study is based on the concept of the direct self-tuning of the control law parameters. Having at hand $\Psi \left( t \right)$ and using \eqref{eq7} and \eqref{eq9}, it is proposed to obtain a dynamic regression equation with respect to the unknown parameters of the optimal control law \eqref{eq7} by a number of transformations of equation \eqref{eq4}:
\begin{equation}\label{eq16}
{y_\theta }\left( t \right) = \Delta \left( t \right)\theta  + {\varepsilon _\theta }\left( t \right){\rm{,}}
\end{equation}
where $\Delta \left( t \right) \in \mathbb{R}$ is a measurable scalar regressor such that $\Psi \left( t \right) \in {\rm{FE}} \Rightarrow \Delta \left( t \right) \in {\rm{FE}}$, ${y_\theta }\left( t \right) \in {\mathbb{R}^{\left( {n + m} \right) \times m}}$ is a measurable function, ${\varepsilon _\theta }\left( t \right) \in {\mathbb{R}^{\left( {n + m} \right) \times m}}$ is an unmeasurable disturbance.

Further, using the results from \cite{Glushchenko22a}, the corollary of the Kalman-Yakubovich-Popov lemma \eqref{eq2}, and exponential stability conditions, it is proposed to derive an adaptive law for $\hat \theta \left( t \right)$, which is based on equation \eqref{eq16} and ensures that the objective \eqref{eq13} can be achieved when $\Psi \left( t \right) \in {\rm{FE}}$.

In the following proposition it is shown how the components of \eqref{eq16} can be obtained.
\begin{proposition}\label{proposition2}
Let the following equations be at hand:
\begin{equation}\label{eq17}
\begin{array}{l}
{z_A}\left( t \right){\rm{:}} = \varphi \left( t \right)A{\rm{,}}\\
{z_B}\left( t \right){\rm{:}} = \varphi \left( t \right)B{\rm{,}}
\end{array}
\end{equation}

%\vspace{-1cm}  

\begin{equation}\label{eq18}
\begin{array}{c}
{z_{{\Phi _{11}}}}\left( t \right){\rm{:}} = {\Delta _\Phi }\left( t \right){\Phi _{11}}\left( {{\tau _\infty }} \right) + {\varepsilon _{{\Phi _{11}}}}\left( t \right){\rm{,}}\\
{z_{{\Phi _{21}}}}\left( t \right){\rm{:}} = {\Delta _\Phi }\left( t \right){\Phi _{21}}\left( {{\tau _\infty }} \right) + {\varepsilon _{{\Phi _{21}}}}\left( t \right){\rm{,}}
\end{array}
\end{equation}
where $\varphi \left( t \right) \in R{\rm{,}}\;{\Delta _\Phi }\left( t \right) \in \mathbb{R}$ are known regressors, ${z_{{\Phi _{11}}}}\left( t \right){\rm{,}}\;{z_{{\Phi _{21}}}}\left( t \right) \in {\mathbb{R}^{n \times n}}{\rm{,}}\;{z_A}\left( t \right) \in {\mathbb{R}^{n \times n}},\;{z_B}\left( t \right) \in {\mathbb{R}^{n \times m}}$ are known functions, ${\varepsilon _{{\Phi _{11}}}}\left( t \right) \in {\mathbb{R}^{n \times n}}{\rm{,}}\;{\varepsilon _{{\Phi _{21}}}}\left( t \right) \in {\mathbb{R}^{n \times n}}$ are bounded unknown disturbances.

Then the function ${y_\theta }\left( t \right)$ and regressor $\Delta \left( t \right)$ can be obtained using the following equations:
\begin{equation}\label{eq19}
\begin{array}{c}
\begin{array}{*{20}{c}}
{{y_\theta }\left( t \right){\rm{:}} = adj\left\{ {\overline \Delta \left( t \right)} \right\}{{\overline y}_\theta }\left( t \right)}\\
{{\varepsilon _\theta }\left( t \right){\rm{:}} = adj\left\{ {\overline \Delta \left( t \right)} \right\}{{\overline \varepsilon }_\theta }\left( t \right)}
\end{array}{\rm{,}}\;\\
\Delta \left( t \right){\rm{:}} = det \left\{ {\overline \Delta \left( t \right)} \right\} = \Delta _{{K_x}}^n\left( t \right)\Delta _{{K_r}}^m\left( t \right){\rm{,}}\\
\begin{array}{*{20}{c}}
{{{\overline y}_\theta }\left( t \right){\rm{:}} = {{ {\begin{bmatrix}
{{y_{{K_x}}}\left( t \right)}&{{y_{{K_r}}}\left( t \right)}
\end{bmatrix}}}^{\rm{T}}}}\\
{{{\overline \varepsilon }_\theta }\left( t \right){\rm{:}} = {{ {\begin{bmatrix}
{{\varepsilon _{{K_x}}}\left( t \right)}&{{\varepsilon _{{K_r}}}\left( t \right)}
\end{bmatrix}}}^{\rm{T}}}}
\end{array}{\rm{,}}\;\\
\overline \Delta \left( t \right){\rm{:}} = {{blockdiag}}\left\{ {{\Delta _{{K_x}}}\left( t \right){I_{n \times n}}{\rm{,}}\;{\Delta _{{K_r}}}\left( t \right){I_{m \times m}}} \right\}{\rm{,}}
\end{array}
\end{equation}
where ${y_{{K_x}}}\left( t \right) \in {{\mathop{\rm \mathbb{R}}\nolimits} ^{m \times n}}{\rm{,}}\;{\Delta _{{K_x}}}\left( t \right) \in {\mathop{\rm \mathbb{R}}\nolimits} $ and ${y_{{K_r}}}\left( t \right) \in {{\mathop{\rm \mathbb{R}}\nolimits} ^{m \times m}}{\rm{,}}\;{\Delta _{{K_r}}}\left( t \right) \in {\mathop{\rm \mathbb{R}}\nolimits} $ take the form:
\begin{displaymath}
\begin{array}{c}
{y_{{K_x}}}\left( t \right){\rm{:}} =  - {R^{ - 1}}z_B^{\rm{T}}\left( t \right){z_{{\Phi _{21}}}}\left( t \right)adj\left\{ {{z_{{\Phi _{11}}}}\left( t \right)} \right\}{\rm{,}}\\
{\Delta _{{K_x}}}\left( t \right){\rm{:}} = \varphi \left( t \right)det \left\{ {{z_{{\Phi _{11}}}}\left( t \right)} \right\}{\rm{,}}\\
{y_{{K_r}}}\left( t \right){\rm{:}} =  - \varphi \left( t \right){R^{ - 1}}z_B^{\rm{T}}\left( t \right)adj\left\{ {{\Delta _{{K_x}}}\left( t \right)z_A^{\rm{T}}\left( t \right)}\right.
\left.{+ y_{{K_x}}^{\rm{T}}\left( t \right)z_B^{\rm{T}}\left( t \right)} \right\}{z_{{\Phi _{21}}}}\left( t \right)adj\left\{ {{z_{{\Phi _{11}}}}\left( t \right)} \right\}{B_r}{\rm{,}}\\
{\Delta _{{K_r}}}\left( t \right){\rm{:}} = det \left\{ {{\Delta _{{K_x}}}\left( t \right)z_A^{\rm{T}}\left( t \right) + y_{{K_x}}^{\rm{T}}\left( t \right)z_B^{\rm{T}}\left( t \right)} \right\}{\rm{,}}
\end{array}
\end{displaymath}
and the implication ${\varepsilon _{{\Phi _{11}}}}\left( t \right) = {\varepsilon _{{\Phi _{21}}}}\left( t \right) = {0_{n \times n}} \Rightarrow {\varepsilon _\theta }\left( t \right) = {0_{\left( {n + m} \right) \times m}}$ holds for the disturbance ${\varepsilon _\theta }\left( t \right)$.
\end{proposition}

\begin{proof}
Proof of the proposition and definitions of ${\varepsilon _{{K_x}}}\left( t \right){\rm{,}}\;{\varepsilon _{{K_r}}}\left( t \right)$ are given in Appendix.
\end{proof}

According to the proved Proposition \ref{proposition2}, to obtain the regression \eqref{eq16} with respect to the unknown parameters of the optimal control law \eqref{eq7}, it is necessary and sufficient to obtain the equations \eqref{eq17} and \eqref{eq18} using a measurable vector $\Psi \left( t \right)$. This challenge can be solved on the basis of the well-known dynamic regressor extension and mixing procedure (DREM) \cite{Aranovskiy21}.

The subsections 3.1 and 3.2 of this section are to present the procedure to obtain regressions \eqref{eq17} and \eqref{eq18} respectively. In the subsection 3.3 the correctness of the implication $\Psi \left( t \right) \in {\rm{FE}} \Rightarrow \Delta \left( t \right) \in {\rm{FE}}$ is proved, the properties of $\Delta \left( t \right)$ and ${\varepsilon _\theta }\left( t \right)$ are studied, and the adaptive law is introduced, which ensures that the objective \eqref{eq13} is met when $\Psi \left( t \right) \in {\rm{FE}}$.

\subsection{Equation \texorpdfstring{\eqref{eq17}}{Lg} parametrization}

To obtain the regression equations \eqref{eq17}, the stable filters for $\dot x\left( t \right)$ and $\Psi \left( t \right)$ from \eqref{eq4} are introduced:
\begin{equation}\label{eq20}
\dot {\overline \mu} \left( t \right) =  - l\overline \mu \left( t \right) + \dot x\left( t \right){\rm{,}}\;\overline \mu \left( {t_r^ + } \right) = {0_n}{\rm{,}}
\end{equation}

%\vspace{-1cm} 

\begin{equation}\label{eq21}
\dot {\overline \Psi} \left( t \right) =  - l\overline \Psi  \left( t \right) + \Psi \left( t \right){\rm{,}}\;\overline \Psi  \left( {t_r^ + } \right) = {0_{n + m}}{\rm{,}}
\end{equation}
where $l > 0$ is the filter constant.

The regressor $\overline \Psi  \left( t \right)$ is a solution of the equation \eqref{eq21}, and, according to \cite{Glushchenko22a,Glushchenko22b,Glushchenko22c}, the function $\mu \left( t \right)$ is calculated as follows under the condition that $\dot x\left( t \right)$ is unknown:
\begin{equation}\label{eq22}
\overline \mu \left( t \right) = {e^{ - l\left( {t - t_r^ + } \right)}}\overline \mu \left( {t_r^ + } \right) + x\left( t \right) - {e^{ - l\left( {t - t_r^ + } \right)}}x\left( {t_r^ + } \right) - l\overline x\left( t \right) + l{e^{ - l\left( {t - t_r^ + } \right)}}\overline x\left( {t_r^ + } \right){\rm{,}}
\end{equation}
where $\overline x\left( t \right) = {\begin{bmatrix}{{I_{n \times n}}}&{{0_{n \times m}}}\end{bmatrix}} \overline \Psi  \left( t \right)$ are the first $n$ elements of $\overline \Psi  \left( t \right)$.

Considering \eqref{eq22}, the equation \eqref{eq4} is rewritten as:
\begin{equation}\label{eq23}
\begin{array}{c}
\overline z\left( t \right) = \overline \mu \left( t \right) + {e^{ - l\left( {t - t_r^ + } \right)}}x\left( {t_r^ + } \right) = \theta _{AB}^{\rm{T}}\overline \Psi  \left( t \right) + {e^{ - l\left( {t - t_r^ + } \right)}}x\left( {t_r^ + } \right) = \overline \theta _{AB}^{\rm{T}}\overline \varphi \left( t \right){\rm{,}}\\
\overline \varphi \left( t \right) = {{\begin{bmatrix}
{{{\overline \Psi  }^{\rm{T}}}\left( t \right)}&{{e^{ - l\left( {t - t_r^ + } \right)}}}
\end{bmatrix}}^{\rm{T}}}{\rm{,}}\;\overline \theta _{AB}^{\rm{T}} = {\begin{bmatrix}
A&B&{x\left( {t_r^ + } \right)}
\end{bmatrix}}{\rm{,}}
\end{array}
\end{equation}
where $\overline z\left( t \right)\! =\! x\left( t \right)\! -\! l\overline x\left( t \right)\! +\! {B_r}\overline r\left( t \right)$ is a measurable function with $\dot {\overline r}\left( t \right)\!\! =\!\!  - l\overline r\left( t \right)\!\! +\!\! r\left( t \right)$, $\overline r\left( {t_r^ + } \right) = {0_m}{\rm{,}}$ $\bar \varphi \left( t \right) \in {\mathbb{R}^{n + m + 1}}$ is a measurable regressor, ${\overline \theta _{AB}} \in {\mathbb{R}^{\left( {n + m + 1} \right) \times n}}$ is an extended vector of the unknown parameters.

\begin{assumption}\label{assumption1}
The parameter $l$ is chosen in such a way that $\overline \Psi  \left( t \right)\! \in\! {\rm{FE}} \Rightarrow \overline \varphi \left( t \right) \in {\rm{FE}}$ holds.
\end{assumption}

The linear minimum-phase filter $\mathfrak{H}\left[ . \right]{\rm{:}} = {1 \mathord{\left/
 {\vphantom {1 {\left( {p + {k_0}} \right)}}} \right. \kern-\nulldelimiterspace} {\left( {p + {k_0}} \right)}}\left[ . \right]$ is introduced, and the DREM procedure \cite{Aranovskiy21} is applied to \eqref{eq23}:
\begin{equation}\label{eq24}
\begin{array}{c}
z\left( t \right) = \varphi \left( t \right){{\overline \theta }_{AB}}{\rm{,}}\\
z\left( t \right){\rm{:}} = \frac{{{k_1}adj\left\{ {{{\overline \varphi }_f}\left( t \right)} \right\}{{\overline z}_f}\left( t \right)}}{{1 + {k_1}det \left\{ {{{\overline \varphi }_f}\left( t \right)} \right\}}}{\rm{,}}\;\varphi \left( t \right){\rm{:}} = \frac{{{k_1}det \left\{ {{{\overline \varphi }_f}\left( t \right)} \right\}}}{{1 + {k_1}det \left\{ {{{\overline \varphi }_f}\left( t \right)} \right\}}}{\rm{,}}\\
{{\dot {\overline z}}_f}\left( t \right) =  - {k_0}{{\overline z}_f}\left( t \right) + \overline \varphi \left( t \right){{\overline z}^{\rm{T}}}\left( t \right){\rm{,}}\;{{\overline z}_f}\left( {t_r^ + } \right) = {0_{\left( {n + m + 1} \right) \times n}}{\rm{,}}\\
{{\dot {\overline \varphi} }_f}\left( t \right) =  - {k_0}{{\overline \varphi }_f}\left( t \right) + \overline \varphi \left( t \right){{\overline \varphi }^{\rm{T}}}\left( t \right){\rm{,}}\;{{\overline \varphi }_f}\left( {t_r^ + } \right) = {0_{\left( {n + m + 1} \right) \times \left( {n + m + 1} \right)}}{\rm{,}}
\end{array}
\end{equation}
where ${k_0} > 0,{\rm{}}\;{k_1} > 0,{\rm{}}\;\varphi \left( t \right) \in \left[ {0{\rm{;\;1}}} \right) \subset \mathbb{R}{\rm{,}}\;z\left( t \right) \in {\mathbb{R}^{\left( {n + m + 1} \right) \times n}}$.

It is easy to obtain the regression equations \eqref{eq17} from \eqref{eq24}:
\begin{equation}\label{eq25}
\begin{array}{c}
{z_A}\left( t \right) = {z^{\rm{T}}}\left( t \right)L = \varphi \left( t \right)A{\rm{,}}\\
{z_B}\left( t \right) = {z^{\rm{T}}}\left( t \right){e_{n + m + 1}} = \varphi \left( t \right)B{\rm{,}}\\
L = {{\begin{bmatrix}
{{I_{n \times n}}}&{{0_{n \times \left( {m + 1} \right)}}}
\end{bmatrix}}^{\rm{T}}} \in {\mathbb{R}^{\left( {n + m + 1} \right) \times n}},\\
{e_{n + m + 1}} = {{{{\begin{bmatrix}
{{0_{m \times n}}}&{{1_{m \times m}}}&0_{m \times 1}
\end{bmatrix}}}} ^{\rm{T}}} \in {\mathbb{R}^{\left( {n + m + 1} \right) \times m}},
\end{array}
\end{equation}
where ${z_A}\left( t \right) \in {\mathbb{R}^{n \times n}},\;{z_B}\left( t \right) \in {\mathbb{R}^{n \times m}}$ are the measurable functions.

Thus, the equation \eqref{eq17}, which, following Proposition 1, is necessary to form the regression equation \eqref{eq16}, \eqref{eq19}, is obtained by filtration of the plant equations \eqref{eq4} with the help of \eqref{eq20}, \eqref{eq21} and application of the DREM procedure \eqref{eq24} to the resulting equation \eqref{eq23}.

\begin{remark}\label{remark4}
According to Lemma 6.8 \cite{Narendra05}, as the filter \eqref{eq21} is stable, if $\Psi \left( t \right) \in {\rm{FE}}$, then $\overline \Psi \left( t \right) \in {\rm{FE}}$. So, considering Assumption \ref{assumption1}, it follows that $\overline \varphi \left( t \right) \in {\rm{FE}}$. In \cite{Aranovskiy19} it is proved that the implication $\overline \varphi \left( t \right) \in {\rm{FE}} \Rightarrow \varphi \left( t \right) \in {\rm{FE}}$ holds for the DREM procedure \eqref{eq24}.
\end{remark}

The next step is to obtain the regression equations \eqref{eq18} on the basis of \eqref{eq25}.

\subsection {Equation \texorpdfstring{\eqref{eq18}}{Lg} parametrization}

The first step to obtain the equations \eqref{eq18} is to expand the matrix exponential $\Phi \left( {{\tau _\infty }} \right) = {e^{D{\tau _\infty }}}$ into the Taylor series:
\begin{equation}\label{eq26}
\Phi \left( {{\tau _\infty }} \right) = {e^{D{\tau _\infty }}} = \sum\limits_{k = 0}^p {{\textstyle{\frac{1} {k{\rm{!}}}}}{D^k}\tau _\infty ^k}  + \underbrace {\sum\limits_{k = p + 1}^\infty  {{\textstyle{\frac{1} {k{\rm{!}}}}}{D^k}\tau _\infty ^k} }_\varepsilon {\rm{,}}
\end{equation}
where $p \in \mathbb{N}$ is the degree of the Taylor polynomial, $\varepsilon  \in {\mathbb{R}^{2n \times 2n}}$ is the residual.

\begin{proposition}\label{proposition3}
The norm of the residual of the $p^{\rm th}$ Taylor polynomial is upper bounded as follows:
\begin{equation}\label{eq27}
\left\| \varepsilon  \right\| \le \underbrace {\frac{{\tau _\infty ^p{{\left\| D \right\|}^p}}}{{\left( {p + 1} \right){\rm{!}}}}\left( {{e^{\left\| D \right\|{\tau _\infty }}} - 1} \right)}_{{\varepsilon _{\max }}}{\rm{,}}
\end{equation}
moreover, when ${\tau _\infty } \in {L_\infty }$, it holds that $\mathop {\lim }\limits_{p \to \infty } {\varepsilon _{\max }} = 0$.
\end{proposition}

\begin{proof}
Proof of Proposition \ref{proposition3} is postponed to Appendix.
\end{proof}

Having the measurable functions ${z_A}\left( t \right)$ and ${z_B}\left( t \right)$, the regression with respect to the matrix $D$ is obtained as:
\begin{equation}\label{eq28}
{z_D}\left( t \right){\rm{:}} = {\begin{bmatrix}
{ - \varphi \left( t \right){z_A}\left( t \right)}&{{z_B}\left( t \right){R^{ - 1}}z_B^{\rm{T}}\left( t \right)}\\
{Q{\varphi ^2}\left( t \right)}&{\varphi \left( t \right)z_A^{\rm{T}}\left( t \right)}
\end{bmatrix}} = {\varphi ^2}\left( t \right)D.
\end{equation}

Then the equation \eqref{eq26} is multiplied by ${\varphi ^{2p}}\left( t \right)$, and the function ${z_D}\left( t \right)$ is substituted into the obtained result:
\begin{equation}\label{eq29}
\begin{array}{c}
{z_\Phi }\left( t \right){\rm{:}} = {\varphi ^{2p}}\left( t \right)\Phi \left( {{\tau _\infty }} \right) - {\varphi ^{2p}}\left( t \right)\varepsilon {\rm{,}}\\
{z_\Phi }\left( t \right){\rm{:}} = \sum\limits_{k = 0}^p {{\textstyle{\frac{1} {k{\rm{!}}}}}{\varphi ^{2p - 2k}}\left( t \right)z_D^k\left( t \right)\tau _\infty ^k} {\rm{,}}
\end{array}
\end{equation}
where ${z_D}\left( t \right) \in {\mathbb{R}^{2n \times 2n}}$ and ${z_\Phi }\left( t \right) \in {\mathbb{R}^{2n \times 2n}}$ are the measurable functions.

The regressions with respect to ${\Phi _{11}}\left( {{\tau _\infty }} \right)$  and ${\Phi _{21}}\left( {{\tau _\infty }} \right)$ are obtained from \eqref{eq29}:
\begin{equation}\label{eq30}
\begin{array}{c}
{z_{{\Phi _{11}}}}\left( t \right){\rm{:}} \! =\! L_1^{\rm{T}}{z_\Phi }\left( t \right){L_1}{\rm{:}} \! =\! {\varphi ^{2p}}\left( t \right){\Phi _{11}}\left( {{\tau _\infty }} \right) + {\varepsilon _{{\Phi _{11}}}}\left( t \right){\rm{,}}\;{\varepsilon _{{\Phi _{11}}}}\left( t \right) \! =\!  - {\varphi ^{2p}}\left( t \right)L_1^{\rm{T}}\varepsilon {L_1},\\
{z_{{\Phi _{21}}}}\left( t \right){\rm{:}} \! =\! L_2^{\rm{T}}{z_\Phi }\left( t \right){L_1}{\rm{:}} \! =\! {\varphi ^{2p}}\left( t \right){\Phi _{21}}\left( {{\tau _\infty }} \right) + {\varepsilon _{{\Phi _{21}}}}\left( t \right){\rm{,}}\;{\varepsilon _{{\Phi _{21}}}}\left( t \right)\! =\!- {\varphi ^{2p}}\left( t \right)L_2^{\rm{T}}\varepsilon {L_1},\\
{L_1} = {{\begin{bmatrix}
{{I_{n \times n}}}&{{0_{n \times n}}}
\end{bmatrix}}^{\rm{T}}} \in {\mathbb{R}^{2n \times n}}{\rm{,}}\;{L_2} = { {\begin{bmatrix}
{{0_{n \times n}}}&{{I_{n \times n}}}
\end{bmatrix}}^{\rm{T}}} \in {\mathbb{R}^{2n \times n}}{\rm{.}}
\end{array}
\end{equation}

Then the change of notation ${\Delta _\Phi }\left( t \right){\rm{:}} = {\varphi ^{2p}}\left( t \right)$ is introduced into \eqref{eq30}, and, as a result, the equation \eqref{eq18} is obtained. So, using \eqref{eq25}, the equation \eqref{eq16} could be obtained on the basis of \eqref{eq19}.

Thus, applying the above-obtained regression equations \eqref{eq25}, Taylor series expansion \eqref{eq26} of matrix exponential from \eqref{eq9} and the function ${z_D}\left( t \right)$ \eqref{eq28}, the regression equations \eqref{eq18}, which are necessary to obtain \eqref{eq16} on the basis of \eqref{eq19}, are derived.

The next step is to synthesize an adaptive law using \eqref{eq16}, which satisfies the objective \eqref{eq13}.

\subsection {Regressor properties and adaptive law}

Before introduction of the law to estimate the unknown parameters $\theta$ of the equation \eqref{eq16}, the properties of both the disturbance ${\varepsilon _\theta }\left( t \right)$ and regressor $\Delta \left( t \right)$ are analyzed. The results of such analysis are presented as propositions.

\begin{proposition}\label{proposition4}
The following two statements hold:
\begin{enumerate}
    \item[a)] If ${\tau _\infty } \in {L_\infty }$, then $\Delta \left( t \right) \in {L_\infty }$. 
    \item[b)] If Assumption 1 holds, $\Psi \left( t \right) \in {\rm{FE}}$ and $p \in {L_\infty }$, then $\Psi \left( t \right) \in {\rm{FE}} \Rightarrow \Delta \left( t \right) \in {\rm{FE}}$.
\end{enumerate}
\end{proposition}

\begin{proof}
Proof of Proposition \ref{proposition4} is given in Appendix.
\end{proof}

\begin{proposition}\label{proposition5}
When the matrix exponential $\Phi \left( {{\tau _\infty }} \right)$ is approximated by the Taylor series \eqref{eq26} and ${\tau _\infty } \in {L_\infty }$, then ${\varepsilon _\theta }\left( t \right)$ is upper bounded by $\left\| {{\varepsilon _\theta }\left( t \right)} \right\| \le \varepsilon _\theta ^{UB}$, and it holds that $\mathop {\lim }\limits_{p \to \infty } \varepsilon _\theta ^{UB} = 0$.
\end{proposition}

\begin{proof}
Proof of Proposition \ref{proposition5} is presented in Appendix.
\end{proof}

Now, having a regression equation \eqref{eq16} with a finitely exciting regressor $\Delta \left( t \right)$ and a bounded disturbance ${\varepsilon _\theta }\left( t \right)$, an adaptive law of the unknown parameters $\theta $, which satisfies the objective \eqref{eq13}, is introduced according to the results of \cite{Glushchenko22a}, Theorem.

Let the operator $\mathfrak{G}\left[ . \right]{\rm{:}} = {{{e^{ - \sigma t}}} \mathord{\left/
 {\vphantom {{{e^{ - \sigma t}}} p}} \right. \kern-\nulldelimiterspace} p}\left[ . \right]$ be introduced and applied to the regression \eqref{eq16} to obtain:
\begin{equation}\label{eq31}
\begin{array}{c}
\Upsilon \left( t \right) = \Omega \left( t \right)\theta  + \varsigma \left( t \right){\rm{,}}\\
\Upsilon \left( t \right){\rm{:}} = \int\limits_{t_r^ + }^t {{e^{ - \int\limits_{t_r^ + }^\tau  {\sigma d{\tau _1}} }}\Delta \left( \tau  \right){y_\theta }\left( \tau  \right)d\tau } {\rm{,}}\;\Omega \left( t \right){\rm{:}} = \int\limits_{t_r^ + }^t {{e^{ - \int\limits_{t_r^ + }^\tau  {\sigma d{\tau _1}} }}{\Delta ^2}\left( \tau  \right)d\tau } {\rm{,\;}}
\varsigma \left( t \right) = \int\limits_{t_r^ + }^t {{e^{ - \int\limits_{t_r^ + }^\tau  {\sigma d{\tau _1}} }}\Delta \left( \tau  \right){\varepsilon _\theta }\left( \tau  \right)d\tau } {\rm{,}}
\end{array}
\end{equation}
where $\sigma  > 0$ is the parameter of the operator $\mathfrak{G}\left[ . \right]$.

The proposition below holds for the new regressor $\Omega \left( t \right)$ and the disturbance $\varsigma \left( t \right)$.
\begin{proposition}\label{proposition6}
The following three statements hold:
\begin{enumerate}
    \item [a)] If $\Delta \left( t \right) \in {\rm{FE}}$ over $\left[ {t_r^ + {\rm{; }}{t_e}} \right]$ and $\Delta \left( t \right) \in {L_\infty }$, then
    \begin{enumerate}
        \item[1)] $\forall t \ge t_r^ + {\rm{}}\;\Omega \left( t \right) \in {L_\infty }{\rm{,}}\;\Omega \left( t \right) \ge 0,$ 
        \item[2)] $\forall t \ge {t_e}{\rm{}}\;\Omega \left( t \right) > 0,{\rm{}}\;{\Omega _{LB}} \le \Omega \left( t \right) \le {\Omega _{UB}},$
    \end{enumerate}
    \item [b)] If $\left\| {{\varepsilon _{{K_x}}}\left( t \right)} \right\| \le \varepsilon _{{K_x}}^{UB}$ and $\Delta \left( t \right) \in {L_\infty }$, then $\left\| {\varsigma \left( t \right)} \right\| \le {\varsigma _{UB}}{\rm{,}}$
    \item [c)] $\mathop {\lim }\limits_{p \to \infty } \varepsilon _\theta ^{UB} = 0 \Rightarrow \mathop {\lim }\limits_{p \to \infty } {\varsigma _{UB}} = 0.$ 
\end{enumerate}
\end{proposition}

\begin{proof}
Proof of the statement (a) of Proposition \ref{proposition6} is presented in \cite{Glushchenko22a}, proof of the statement (b) can be found in \cite{Glushchenko21}, the statement (c) holds according to Proposition \ref{proposition5} and the definition of $\varsigma \left( t \right)$.
\end{proof}

The result of the adaptive law derivation to obtain $\hat \theta \left( t \right)$ on the basis of the regression equation \eqref{eq31} is represented as Theorem.

\begin{theorem}\label{theorem1}
The following two statements hold.
\begin{enumerate}
    \item [a)] Let $\Psi \left( t \right) \in {\rm{FE}}{\rm{,}}\;p \in {L_\infty }$, Assumption \ref{assumption1} be met, and $Q{\rm{,}}\;R{\rm{,}}\;{\tau _\infty } \in {L_\infty }{\rm{,}}\;\vartheta$ be chosen so that the analytical solution \eqref{eq9} of the differential Riccati equation exists and affords a minimum to \eqref{eq6}. Then the adaptive law:
    \begin{equation}\label{eq32}
    \begin{array}{c}
    \dot {\hat \theta} \left( t \right) = \dot {\tilde \theta} \left( t \right) =  - \gamma \left( t \right)\Omega \left( t \right)\left( {\Omega \left( t \right)\hat \theta \left( t \right) - \Omega \left( t \right)\theta  - \varsigma \left( t \right)} \right) 
    =  - \gamma \left( t \right){\Omega ^2}\left( t \right)\tilde \theta \left( t \right) + \gamma \left( t \right)\Omega \left( t \right)\varsigma \left( t \right){\rm{,}\;}\\
    \gamma \left( t \right) = \left\{ \begin{array}{l}
    0,\;{\rm{ if }}\;\Omega \left( t \right) \le \rho  \in \left( {0{\rm{;}}\;{\Omega _{LB}}} \right)\\
    \frac{{{\gamma _0}{\lambda _{{\rm{max}}}}\left( {\omega \left( t \right){\omega ^{\rm{T}}}\left( t \right)} \right) + {\gamma _1}}}{{{\Omega ^2}\left( t \right)}},\;{\rm{otherwise}}
    \end{array} \right.{\rm{,}}
    \end{array}
    \end{equation}
    $\forall t \ge {t_e}$ guarantees that $\left\| {\tilde \theta \left( t \right)} \right\| \le {e^{ - {\eta _{\tilde \theta }}\left( {t - {t_e}} \right)}}\left\| {\tilde \theta \left( {{t_e}} \right)} \right\| + {\textstyle{\frac{{\varsigma _{{\rm{UB}}}}} {\rho }}}$, and, if $\exists \overline t > {t_e}$ such that $\forall t \ge \overline t{\rm{}}\;{\varsigma _{{\rm{UB}}}} < {\textstyle{ \frac{\rho {\lambda _{{\rm{min}}}}\left( {\overline Q} \right)} {2\left\| {\overline PB} \right\|}}} - \rho {e^{ - 0.5{\eta _{\tilde \theta }}\left( {\overline t - {t_e}} \right)}}\left\| {\tilde \theta \left( {{t_e}} \right)} \right\|$, where $\overline P \in {R^{n \times n}}$ is a solution of the Lyapunov equation ${A_{ref}}\overline P + A_{ref}^{\rm{T}}\overline P =  - \overline Q =  - {\overline Q^{\rm{T}}} < 0$, then \eqref{eq32} ensures additionally that:
    \begin{enumerate}
        \item[1)] $\forall t \ge t_r^ + {\rm{}}\;\xi \left( t \right) \in {L_\infty }{\rm{;}}$
        \item[2)] $\forall t \ge {t_e}{\rm{}}\;\mathop {{\rm{lim}}}\limits_{t \to \infty } \left\| {\xi \left( t \right)} \right\| \le {{\rm{\varepsilon }}_\xi }{\rm{}}\;\left( {\exp } \right).$
    \end{enumerate}
    Moreover, (i) the exponential convergence rate of the augmented error $\xi \left( t \right)$ to the set with a bound ${{\rm{\varepsilon }}_\xi }$ is directly proportional to the parameters $\gamma_0 \geq 1$ and $\gamma_1 \geq 0$; (ii) the upper bound ${{\rm{\varepsilon }}_\xi }$ can be reduced to the machine precision by the increase of the degree of the Taylor polynomial $p$.
    
    \item [b)] Let all assumptions of the statement (a) be satisfied and the degree of Taylor polynomial $p$ be such that $\varsigma \left( t \right) = o\left( {\Omega \left( t \right)\theta } \right){\rm{,}}$ then, in addition to the properties of (a), the law \eqref{eq32} guarantees that $\left| {{{\tilde \theta }_i}\left( {{t_a}} \right)} \right| \le \left| {{{\tilde \theta }_i}\left( {{t_b}} \right)} \right|{\rm{}}\;\forall {t_a} \ge {t_b}$.
\end{enumerate}
\end{theorem}

\begin{proof}
The proof of Theorem and the definition of ${{\rm{\varepsilon }}_\xi }$ are postponed to Appendix.
\end{proof}

Therefore, it follows from Theorem that the proposed adaptive law \eqref{eq32} allows one to meet the objective \eqref{eq13}. It ensures the convergence of the augmented tracking error $\xi \left( t \right)$ to a set with the bound ${{\rm{\varepsilon }}_\xi }$, which can be reduced by the improvement of the Taylor polynomial degree. As, in practice, the maximum value of such degree $p$ \eqref{eq26} depends only on the computational power of the hardware-in-use, then it can be assumed that the value ${{\rm{\varepsilon }}_\xi }$ can be reduced to the machine precision.

The method proposed in this study differs from the known approaches to the synthesis of the ALQ control laws in the following way.
\begin{enumerate}
    \item[\textbf{C1}] The regression equation without overparametrization \eqref{eq16} is obtained, which relates the measurable signals $\Delta \left( t \right){\rm{,}}\;{y_\theta }\left( t \right)$ to the unknown matrix $\theta $ of the optimal control law parameters.
    
    \item[\textbf{C2}] The regressor of the parametrization \eqref{eq16} is finitely exciting $\Delta \left( t \right) \in {\rm{FE}}$ when a relatively weak Assumption \ref{assumption1} is met. Moreover, in accordance with Proposition 4, if Assumption \ref{assumption1} holds, then the fact that $\Delta \left( t \right) \in {\rm{FE}}$ follows analytically and directly from $\Psi \left( t \right) \in {\rm{FE}}$.
    
    \item[\textbf{C3}] When $\Psi \left( t \right) \in {\rm{FE}}$, the adaptive law \eqref{eq32} ensures the convergence of the identification error $\tilde \theta \left( t \right)$ into a set, which bound depends only on the computational power of the hardware-in-use.
    
    \item[\textbf{C4}] Global exponential stability of the closed loop, which includes the plant and the developed adaptive self-tuning regulator, is guaranteed.
\end{enumerate}

{Estimating the computational complexity of the proposed method, it should be noted that its upper bound is defined by the complexity of the most computationally intensive operations. They include the calculation of the $adj\left\{ {\overline \varphi_f (t)} \right\}$ and $det \left\{ {\overline \varphi_f (t)} \right\}$ (this regressor has the highest dimension), as well as matrix multiplication (the matrices $z_D (t)$ and $z_\Phi (t)$ have the highest dimension of $2n \times 2n$). Considering the first two operations, their complexity is $O\left( {{(n+m+1)^3}} \right)$, for square matrices multiplication -- $O\left( {{(2n)^3}} \right)$. For non-square matrices such estimation is a bit different, but without losing generality, we can assume that $2n$ is the greatest value among all dimensions, then the estimation $O\left( {{(2n)^3}} \right)$ can be used for non-square matrices as a majorant. Thus, the upper bound of the computational complexity of the proposed identification algorithm depends on whether $n+m+1 > 2n$ or not. In the first case it is $O\left( {{(n+m+1)^3}} \right)$, in the second one -- $O\left( {{(2n)^3}} \right)$. Considering ALQ-STR methods, only authors of \cite{Jha19a} made such estimations and stated that the complexity of their method is $O\left( {{n^6}} \right)$}.

\newpage{The block diagram of the developed adaptive control system is shown in Fig. 1.}
\begin{figure}[h]
\centerline{\includegraphics[height=20pc]{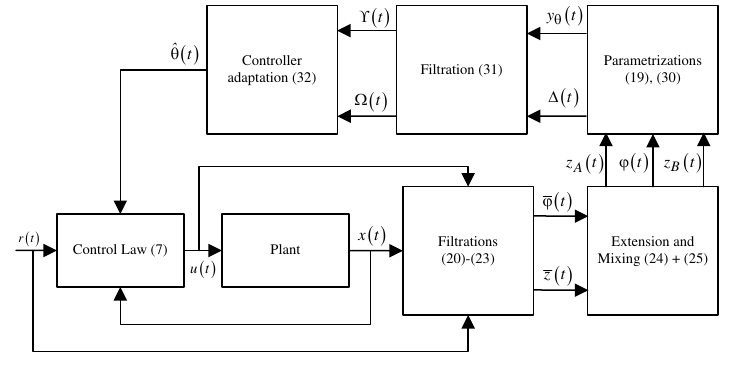}}
\caption{{{Block diagram of the proposed adaptive control system}}\label{fig1}}
\end{figure}

\subsection {Tuning guidelines}

The arbitrary parameters of the proposed ALQ control law \eqref{eq11}, \eqref{eq32} can be divided into three main groups:
\begin{enumerate}
    \item[i)] the parameters ${\tau _\infty }$ and $p$, which determine the value of the augmented error ${\varepsilon _\xi }$;
    \item[ii)] the parameters $r{\rm{,}}\;l{\rm{,}}\;{k_0}{\rm{,}}\;{k_1}$ and $\sigma$, which ensure the propagation of the regressor excitation through parametrizations \eqref{eq23}, \eqref{eq25}, \eqref{eq31};
    \item[iii)] the coefficients ${\gamma _0}{\rm{,}}\;{\gamma _1}$, which adjust the convergence rate of the augmented error $\xi \left( t \right)$, and $\rho $, which defines the initial time instant of the control law \eqref{eq11} adjustment.
\end{enumerate}

Here, the detailed tuning guidelines to choose the values of all these parameters are summarized as follows.

\subsubsection{First group}
The parameter ${\tau _\infty }$ according to \eqref{eq9} determines the value of the equation \eqref{eq8} solution at time instant ${\tau _\infty }$. Following Proposition \ref{proposition1}, only the static value $P = P\left( {{\tau _\infty } \to \infty } \right)$ allows one to obtain the parameters ${K_x}{\rm{,}}\;{K_r}$, which affords a minimum to the optimization problem \eqref{eq6}. Therefore, given the requirement ${\tau _\infty } \in {L_\infty }$ from Propositions \ref{proposition3}, \ref{proposition4}, \ref{proposition5} and Theorem \ref{theorem1}, the parameter ${\tau _\infty }$ should, first of all, satisfy the inequality ${\tau _\infty } \ge {\tau _{ss}}$, where ${\tau _{ss}}$ is some unknown value, which defines the duration of the transient of the Riccati differential equation \eqref{eq8} solution. Otherwise, when ${\tau _\infty } < {\tau _{ss}}$, the adaptive law \eqref{eq32} allows one to obtain some matrix $\mathord{\buildrel{\lower3pt\hbox{$\scriptscriptstyle\smile$}} 
\over \theta }$, which is not a solution of the optimization problem \eqref{eq6} over an infinite time interval. But it stabilizes the system \eqref{eq4} and is calculated on the basis of the matrix $P\left( {{\tau _{ss}}} \right)$, which is a boundary condition of the problem of a linear-quadratic function minimization over a finite time interval $\left[ {t_r^ + {\rm{;}}\;{\tau _{ss}}} \right]$:
\begin{equation}\label{eq33}
J = \frac{1}{2}\left[ {{x^{\rm{T}}}\left( t \right)P\left( {{\tau _{ss}}} \right)x\left( t \right) + \int\limits_{t_r^ + }^{{\tau _{ss}}} {{x^{\rm{T}}}\left( t \right)Qx\left( t \right) + {u^{\rm{T}}}\left( t \right)Ru\left( t \right)} {\rm{}\;}dt} \right].
\end{equation}

In such case, the difference $\theta  - \mathord{\buildrel{\lower3pt\hbox{$\scriptscriptstyle\smile$}} 
\over \theta } $ influences the steady-state value of the augmented error ${\varepsilon _\xi }$. Hence, if a priori information about the system \eqref{eq6} and ${\tau _{ss}}$ is unavailable, then a significantly high value of ${\tau _\infty }$ should be chosen to ensure that the steady-state solution of the Riccati equation \eqref{eq8} is achievable.

However, if $\exists i{\rm{,}}\;j \in \overline {1,2n}$ ${\lambda _i}\left( D \right) \gg {\lambda _j}\left( D \right)$, then, according to Remark \ref{remark3}, an arbitrary increase of ${\tau _\infty }$ leads to singularity of the analytical solutions \eqref{eq9}. Moreover, as it follows from \eqref{eq27}, it also results in a growth of the approximation error $\varepsilon $ of the expansion \eqref{eq26} and, consequently, the value of ${{\rm{\varepsilon }}_\xi }$. Therefore, an arbitrary increase of ${\tau _\infty }$ is unacceptable, and its value in the general case should be small enough to guarantee the existence of analytical solutions \eqref{eq9}.

In practice, as a rule, {\it a priori} information about the structure of matrices $A{\rm{,}}\;B$ and the order of magnitude of their individual elements is sufficient to select ${\tau _\infty }$ value. As noted in Remark \ref{remark3}, it is possible to influence the duration of transients ${\tau _{ss}}$ of the Riccati equation solution by appropriate choice of $Q{\rm{,}}\;R$ and $\vartheta$.

For any values of ${\tau _\infty }$, in accordance with Proposition \ref{proposition3} and Theorem \ref{theorem1}, the increase of the parameter $p$ allows one to reduce $\varepsilon$ and ${\varepsilon _\xi }$. Moreover, following statement (b) of the Theorem, if the value of $p$ is chosen such that \linebreak $\varsigma \left( t \right) = o\left( {\Omega \left( t \right)\theta } \right){\rm{,}}$ then the proposed adaptive law \eqref{eq32} ensures the monotonic transients of the parameter error elements ${\tilde \theta _i}\left( t \right)$. It is therefore advisable to make the value of $p$ be as high as the computational power of the hardware-in-use allows. On the other hand, in accordance with Proposition \ref{proposition4}, in the limit an arbitrary increase of $p$ leads to the regressor $\Delta \left( t \right)$ excitation vanishing. So, the requirement of boundedness of the $p$ value should be imposed even in case of the unlimited computational power. Therefore, in the general case, it is rational to improve $p$ up to some finite limit, at which at least \linebreak $\varsigma \left( t \right) = o\left( {\Omega \left( t \right)\theta } \right)$ is ensured.

\subsubsection{Second group}

First of all, to provide the finite excitation of the initial regressor $\Psi \left( t \right) \in {\rm{FE}}$, the reference signal $r$ for the proposed adaptive control system is to be different from zero $r \ne {0_m}$. This condition follows from the fact that, when $r = {0_m}$, then:
\begin{equation}\label{eq34}
\begin{array}{c}
\exists {\ell _i} \ne 0{\rm{}\;}\sum\limits_{i = 1}^{m + n} {{\ell _i}{\Psi _i}\left( t \right)}  = \sum\limits_{i = 1}^n {{\ell _i}{x_i}\left( t \right)}  + \sum\limits_{i = 1}^m {{\ell _i}{u_i}\left( t \right)}
 = \sum\limits_{i = 1}^n {{\ell _i}{x_i}\left( t \right)}  + \sum\limits_{i = 1}^m {{\ell _i}k_x^ix\left( t \right)}  = 0,{\rm{}\;}K_x^{\rm{T}} = {\begin{bmatrix}
{k_x^1}&{k_x^2}& \ldots &{k_x^m}
\end{bmatrix}},
\end{array}
\end{equation}
where ${\Psi _i}\left( t \right)$  is the $i^{\rm th}$ element of the regressor $\Psi \left( t \right)$, ${\ell _i} \ne 0$  is an arbitrary constant. In accordance with the results of \cite{Wang21}, Proposition 1, the linear dependence beetween regressor elements \eqref{eq34} is the sufficient condition for $\Psi \left( t \right) \notin {\rm{FE}}$.

Therefore, the choice of a nonzero reference value $r$ for the proposed adaptive control system is a condition to propagate the excitation of the regressor through the parameterizations \eqref{eq23}, \eqref{eq25}, and \eqref{eq31}. It should be noted that $r \ne {0_m}$ is a necessary but not sufficient condition for $\Psi \left( t \right) \in {\rm{FE}}$, since, even when $r \ne {0_m}$, it holds that $\Psi \left( t \right) \notin {\rm{FE}}$ if $\exists {\ell _i} > 0$ such that:
\begin{equation}\label{eq35}
\begin{array}{c}
{\rm{}\;}\sum\limits_{i = 1}^{m + n} {{\ell _i}{\Psi _i}\left( t \right)}  = \sum\limits_{i = 1}^n {{\ell _i}{x_i}\left( t \right)}  + \sum\limits_{i = 1}^m {{\ell _i}{u_i}\left( t \right)}
 = \sum\limits_{i = 1}^n {{\ell _i}{x_i}\left( t \right)}  + \sum\limits_{i = 1}^m {{\ell _i}k_x^ix\left( t \right)}  + \sum\limits_{i = 1}^m {{\ell _i}k_r^ir}  =  - \sum\limits_{i = 1}^m {{\ell _i}k_r^ir}  + \sum\limits_{i = 1}^m {{\ell _i}k_r^ir}  = {\rm{0}}{\rm{,}\;}\\
{K_r} = {\begin{bmatrix}
{k_r^1}&{k_r^2}& \ldots &{k_r^m}
\end{bmatrix}}
\end{array}
\end{equation}

In this case, the requirement $\Psi \left( t \right) \in {\rm{FE}}$ seems possible to be met by addition of a noise or harmonic component to the control signal $u\left( t \right)$ in accordance with the results of \cite{Adetola06}.

The parameter $l$ should be chosen so that the decrease rate of the exponent ${e^{ - l\left( {t - t_r^ + } \right)}}$ from \eqref{eq23} is substantially less than the assumed settling time of transients in \eqref{eq4}, i.e. $l \ll {\textstyle{\frac{{t_{ss}}} {\left( {3 \div 5} \right)}}}$, where ${t_{ss}}$ is the assumed transient duration in the system \eqref{eq4}. Otherwise, if a fast decrease rate of ${e^{ - l\left( {t - t_r^ + } \right)}}$ is chosen, the excitation time range of $\overline \varphi \left( t \right)$ may be substantially shorter than the one of the original regressor $\Psi \left( t \right)$, and the entire adaptive control scheme will quickly lose its awareness to the new data from the system. In fact, Assumption \ref{assumption1} is needed to guarantee that there can be found no function ${\overline \Psi _i}\left( t \right){\rm{,}\;}i = \overline {1,n + m} $ in $\overline \Psi \left( t \right)$ such that:
\begin{equation}\label{eq36}
{\overline \Psi _i}\left( t \right) = {e^{ - l\left( {t - t_r^ + } \right)}}\;\;{\rm{ or }}\;\;\exists {\ell _i} \ne 0{\rm{}\;}\sum\limits_{i = 1}^{m + n} {{\ell _i}{{\overline \Psi }_i}\left( t \right)}  + {\ell _{m + n + 1}}{e^{ - l\left( {t - t_r^ + } \right)}} = 0,
\end{equation}
where ${\overline \Psi _i}\left( t \right)$  is the $i^{\rm th}$ element of the regressor $\overline \Psi \left( t \right)$.

Therefore, Assumption \ref{assumption1} is not restrictive and usually met for almost all values of $l$.

The parameters $l$ and ${k_0}$ define the rate of the excitation propagation throughout the parameterizations \eqref{eq20}, \eqref{eq25} and must be chosen on the basis of the required response of the adaptive control system.

The parameter ${k_1}$ is used to make the regressor $\varphi \left( t \right)$ values be of the desired magnitude order over the excitation time range. For example, let it be known that $\forall t \in \left[ {t_r^ + {\rm{;}}\;{t_e}} \right]{\rm{}\;}{10^{ - 5}} \le det \left\{ {H\left[ {\overline \varphi \left( t \right){{\overline \varphi }^{\rm{T}}}\left( t \right)} \right]} \right\} \le {10^{ - 3}}$, then the choice ${k_1} = {10^5}$ guarantees that $\forall t \in \left[ {t_r^ + {\rm{;}}\;{t_e}} \right]{\rm{}\;}0.5 \le \varphi \left( t \right) \le 0.99$. Such normalization of the regressor $\varphi \left( t \right)$ makes it sufficiently easier to choose the parameter $\rho$ of the adaptive law \eqref{eq32} and can be very useful from the point of view of the computational stability of the proposed control system for its successful implementation on the embedded hardware.

The parameter $\sigma $ defines the time range, when the filter \eqref{eq31} is sensitive to the new input data, and ensures the boundedness of the regressor $\Omega \left( t \right)$, function $\Upsilon \left( t \right)$ and disturbance $\varsigma \left( t \right)$. The parameter $\sigma$ is advised to be calculated as $\sigma  = {\textstyle{\frac{{t_{ss}}} {\left( {3 \div 5} \right)}}}$, as in such case the sensitivity of the proposed adaptive control scheme to new data on ${y_\theta }\left( t \right)$ and $\Delta \left( t \right)$ is ensured strictly over the assumed time range, when the transients take place in the system \eqref{eq3}.

\subsubsection{Third group}
The parameter ${\gamma _0}$ is to be chosen following the condition ${\gamma _0} \ge 1$, which, according to the above proved results, guarantees the correctness of the results of Theorem. The increase of ${\gamma _0}$ allows one to improve the convergence rate of the error $\xi \left( t \right)$. The parameter ${\gamma _1}$ sets the value of the minimum rate of the error $\xi \left( t \right)$ convergence. The value of $\rho $ is chosen on the basis of a priori information about the minimum possible value of the regressor ${\Omega _{LB}}$ and determines the time instant when the parameters of the control law \eqref{eq11} is started to be adjusted. For example, if it is known that $\forall t \in \left[ {t_r^ + {\rm{;}}\;{t_e}} \right]{\rm{}}\;0.5 \le \varphi \left( t \right) \le 0.99$, then the parameter $\rho$ has to satisfy the condition: $\rho  < {\Omega _{LB}} \le {\textstyle{\frac{0.5} {\sigma}}}\left( {{t_e} - t_r^ + } \right)$.

\begin{remark}\label{remark5}
The main result of this study is the parameterization \eqref{eq16}-\eqref{eq19}, which reduces the problem of the adaptive optimal control \eqref{eq13} to the one of identification of the linear regression equation \eqref{eq16} unknown parameters. In order to solve it and achieve the objective \eqref{eq13}, various modern adaptive laws can be applied, which guarantee that $\Delta \left( t \right) \in {\rm{FE}} \Rightarrow \mathop {{\rm{lim}}}\limits_{t \to \infty } \left\| {\tilde \theta \left( t \right)} \right\| = 0$. Following the principle of certainty equivalence, it means that \eqref{eq13} is met. In this paper, of all known adaptive laws ensuring that $\Delta \left( t \right) \in {\rm{FE}} \Rightarrow \mathop {{\rm{lim}}}\limits_{t \to \infty } \left\| {\tilde \theta \left( t \right)} \right\| = 0$, the law \eqref{eq31}, \eqref{eq32}, which has been proposed earlier by the authors and described in details in \cite{Glushchenko22a,Glushchenko22b,Glushchenko22c,Glushchenko21}, is successfully applied.
\end{remark}

\begin{remark}\label{remark6}
Following the proof of Theorem, if $\exists t \ge \overline t{\rm{}}\;{\varsigma _{{\rm{UB}}}} + {w_{\max }} < {\textstyle{\frac{\rho {\lambda _{{\rm{min}}}}\left( {\overline Q} \right)} {2\left\| {\overline PB} \right\|}}} - \rho {e^{ - 0.5{\eta _{\tilde \theta }}\left( {\overline t - {t_e}} \right)}}\left\| {\tilde \theta \left( {{t_e}} \right)} \right\|$, then the developed adaptive control system is robust to the external bounded disturbance $\left\| {w\left( t \right)} \right\| \le {w_{\max }}$, which is added to the right hand side of the regression equation \eqref{eq31}. The above inequality means that the identification of the control law parameters, which stabilize the system \eqref{eq4}, is possible on the basis of the perturbed regression $\Upsilon \left( t \right) = \Omega \left( t \right)\theta  + w\left( t \right) + \varsigma \left( t \right)$.
\end{remark}

\begin{remark}\label{remark7}
The proposed adaptive law \eqref{eq32} provides stability of ${e_{ref}}\left( t \right)$ only when $\Psi \left( t \right) \in {\rm{FE}}$. Therefore, to implement \eqref{eq32} in practice, the {\it a priori} information is required that this condition holds. If $\Psi \left( t \right) \notin {\rm{FE}}$ for a particular plant \eqref{eq4} and a particular reference $r$, then it is possible to hold $\Psi \left( t \right) \in {\rm{FE}}$ artificially by addition of the dither noise to the control or the reference signal according to \cite{Adetola06,Cao07}. {It should be additionally noted that such noise should be added for a finite time interval only in order to ensure FE.} The actual problem is to obtain conditions on the reference $r$, under which the requirement $\Psi \left( t \right) \in {\rm{FE}}$ is met for the whole class of controllable systems \eqref{eq4}.
\end{remark}

\begin{remark}\label{remark8}
As follows from equations \eqref{eq23} and \eqref{eq31}, due to the decrease of the functions ${e^{ - l\left( {t - t_r^ + } \right)}}$ and ${e^{ - \sigma \left( {t - t_r^ + } \right)}}$, the developed adaptive system gradually loses the ability to receive, process, and use the vector of signals $\Psi \left( t \right)$ measured from the system for the control law parameters adjustment. It is possible to overcome this disadvantage if the time moment $t_r^ + $ is estimated dynamically by some external algorithm, e.g. $t_r^ +  \in \Im  = \left\{ {{t_1}{\rm{,}}\;{t_2}{\rm{,}}\; \ldots {\rm{,}}\;{t_{j - 1}}{\rm{,}}\;{t_j}} \right\}{\rm{,}}\forall t \in \left[ {{t_{j - 1}}{\rm{;}}\;{t_j}} \right){\rm{,}}\;t_r^ +  = {t_{j - 1}}$. Then, all filters \eqref{eq20}-\eqref{eq24}, \eqref{eq31} are reset at each time instant ${t_j}$, and the adaptive system regains its awareness to the signals $\Psi \left( t \right)$ and ensures that the objective \eqref{eq13} is achieved. From the point of view of most practical scenarios, it is very useful and convenient to assign a new value to $t_r^ + $ at each time instant when the reference $r$ value is changed.
\end{remark}
\newpage
\section{Numerical experiments}\label{sec4}

The experimental validation of the proved properties of the proposed adaptive control system has been conducted in Matlab/Simulink using the Euler numerical integration method with a time-invariant step size $\partial \tau  = {10^{ - 4}}$ seconds.

\subsection {\texorpdfstring{$\cancel{\exists }i{\rm{,}}j \in \overline {1,2n} {\rm{}\;}{\lambda _i}\left( D \right) \gg {\lambda _j}\left( D \right)$}{Lg}}

The parameters of the plant \eqref{eq3} were chosen to be the same as in \cite{Jha17,Jha14}:
\begin{equation}\label{eq37}
\forall t \ge 0{\rm{}\;}{\dot x_p}\left( t \right) = {\begin{bmatrix}
0&1\\
{ - 1}&{ - 1}
\end{bmatrix}} {x_p}\left( t \right) + {\begin{bmatrix}
0\\
1
\end{bmatrix}} u\left( t \right){\rm{,}}
\end{equation}
therefore, the interested reader, if necessary, can compare the experimental results, which are presented below, with the ones in \cite{Jha17,Jha14}.

The first experiment was to demonstrate what effect the parameter ${\tau _\infty }$ value had on the control quality when the ideal law \eqref{eq7} was implemented. For this purpose, we chose the following values of the reference $r$, initial conditions $x\left( 0 \right)$, matrices $Q{\rm{,}}\;R$ and the parameter $\vartheta$:
\begin{equation}\label{eq38}
r = 1,{\rm{}}\;x\left( 0 \right) = {0_3}{\rm{,}}\;Q = {I_{3 \times 3}}{\rm{,}}\;R = 1,{\rm{}}\;\vartheta  = 1.
\end{equation}

Applying the eigen-decomposition to $D$, it is easy to show that $\cancel{\exists }i{\rm{,}}\;j \in \overline {1,2}$ ${\lambda _i}\left( D \right) \gg {\lambda _j}\left( D \right)$:
\begin{equation}\label{eq39}
{V^{ - 1}}DV = diag\left\{ \begin{array}{c}
{\rm{0}}{\rm{.79}} + {\rm{0}}{\rm{.92}}{\mathop{\rm i}\nolimits} \\
{\rm{0}}{\rm{.79}}  - {\rm{0}}{\rm{.92i}}\\
{\rm{0}}{\rm{.67}}\\
 - {\rm{0}}{\rm{.79  +  0}}{\rm{.92}}{\mathop{\rm i}\nolimits} \\
 - {\rm{0}}{\rm{.79}} - {\rm{0}}{\rm{.92}}{\mathop{\rm i}\nolimits} \\
 - {\rm{0}}{\rm{.67}}
\end{array} \right\}{\rm{,}}
\end{equation}
where $V \in {\mathbb{R}^{6 \times 6}}$ is the matrix, which constitutes of the eigenvectors of $D$, $i$ is the imaginary unit.

Figure 2 depicts the transients of the Riccati equation solution $P\left( t \right)$, in which the values of the matrix elements $P\left( {{\tau _\infty }} \right)$ are marked with crosses. They were analytically calculated with the help of \eqref{eq9} for different values of ${\tau _\infty }$.

\begin{figure}[h]
\begin{center}
\includegraphics[scale=0.3]{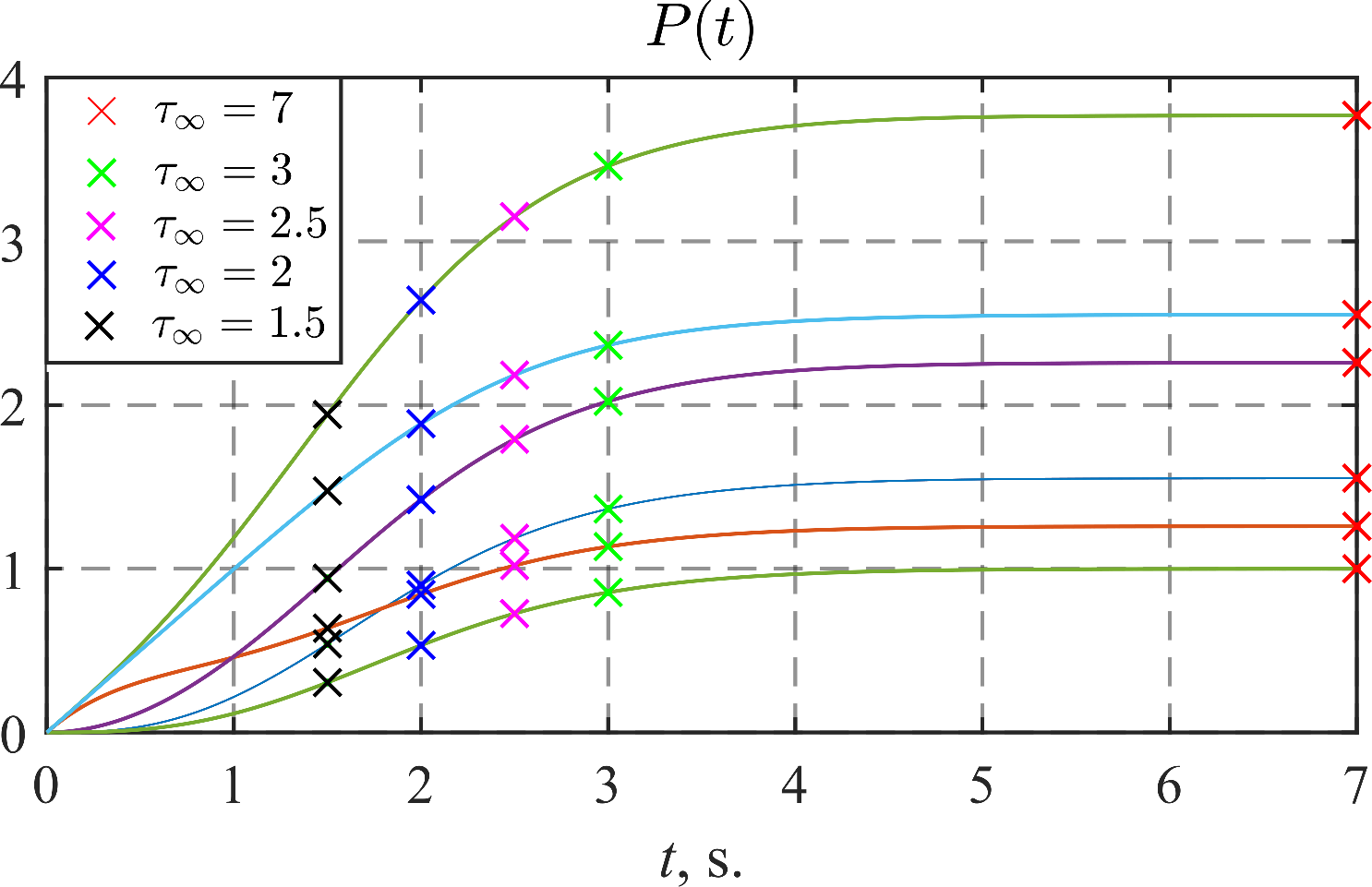}    % The printed column width is 8.4 cm.
\caption{Solution of Riccati equation $P\left( t \right)$ and matrices $P\left( {{\tau _\infty }} \right)$} 
\end{center}
\end{figure}

 Figure 3(a) presents the transients of the output $z\left( t \right)$ of the closed-loop system \eqref{eq10}, which controller parameters ${K_x}$ and ${K_r}$ were calculated using matrices $P\left( {{\tau _\infty }} \right)$ related to the different values of ${\tau _\infty }$. The transients of the cost function \eqref{eq6}, which were obtained using such controller parameters, are shown in Figure 3(b).
 
\begin{figure}
\begin{center}
\includegraphics[scale=0.25]{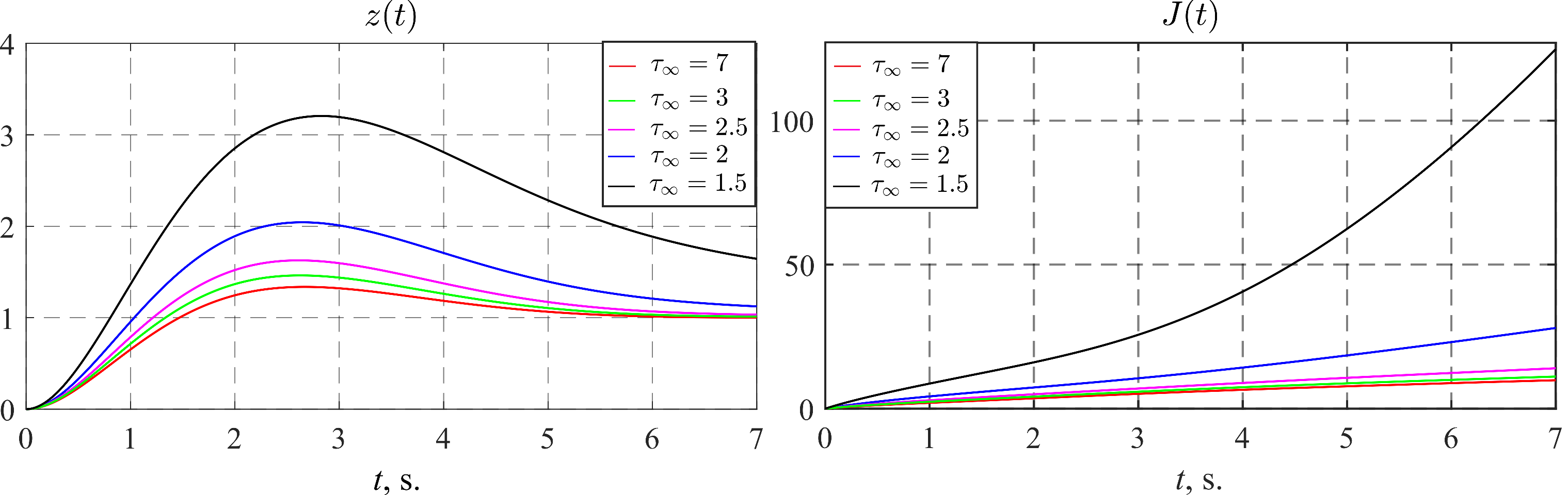}    % The printed column width is 8.4 cm.
\caption{{Dependencies on ${\tau _\infty }$values of: (a) output $z\left( t \right)$ and (b) cost function $J\left( t \right)$}} 
\end{center}
\end{figure}

The experimental results confirmed the recommendations on the choice of ${\tau _\infty }$ given in Remark \ref{remark2} and at the end of Section 3. Indeed, as it was noted in Proposition \ref{proposition1}, when ${\tau _\infty } \to \infty $, the analytical solution \eqref{eq9} of the Riccati differential equation \eqref{eq8} coincided with the steady-state one and allowed us to obtain the control law, which minimized the linear-quadratic cost function \eqref{eq6}. The transients of the system output $z\left( t \right)$ clarified what effect the parameter value of ${\tau _\infty }$ had on the required control quality of the system \eqref{eq4} with parameter uncertainty.

Table 1 presents the experimentally obtained relationship between $\left\| \varepsilon  \right\|$ and the values of the parameters $p{\rm{,}}\;{\tau _\infty }$.

\begin{table}[th]
\centering
\caption{Relationship between $\left\| \varepsilon  \right\|$ and  $p{\rm{,}}\;{\tau _\infty }$.}
{\begin{tabular}{ccccc} \toprule
 & \multicolumn{4}{c}{$p$} \\ \cmidrule{2-5}
 ${\tau _\infty }$ & 20 & 25 & 30 & 35 \\ \midrule
 1.5 & $2.21\cdot 10^{-14}$ & $2.922 \cdot 10^{-15}$ & $2.922 \cdot 10^{-15}$ & $2.922 \cdot 10^{-15}$\\
 2 &	$9.033 \cdot 10^{-12}$	& $4.28 \cdot 10^{-15}$ & $4.278 \cdot 10^{-15}$	& $4.278 \cdot 10^{-15}$\\
 2.5 & $9.798 \cdot 10^{-10}$ & $2.508 \cdot 10^{-14}$ & $7.158 \cdot 10^{-15}$ & $7.158 \cdot 10^{-15}$\\ 
 3 & $4.511 \cdot 10^{-8}$ & $2.698 \cdot 10^{-12}$ & $1.597 \cdot 10^{-14}$	& $1.597 \cdot 10^{-14}$\\
 7 & 2.416 & 0.0105 & $2.769 \cdot 10^{-5}$ & $2.234 \cdot 10^{-8}$\\\bottomrule
\end{tabular}}
\label{table1}
\end{table}

The ideal values of the matrix exponential were calculated with the help of the Matlab/Simulink function $expm(.)$, the value of the error was calculated as follows: $\varepsilon  = {\rm{expm}}\left( {D{\tau _\infty }} \right) - \sum\limits_{k = 0}^p {{\textstyle{\frac{1} {k{\rm{!}}}}}{D^k}\tau _\infty ^k} $.

The dependence presented in Table 1 confirmed the conclusions from Proposition \ref{proposition3}. Indeed, if ${\tau _\infty } \in {L_\infty }$, then the higher the value of $p$, the lower the error of the Taylor-series-based \eqref{eq26} approximation of the matrix exponential $\Phi \left( {{\tau _\infty }} \right)$.

Now, having an experimental confirmation of the important theoretical conclusions concerning the parameters ${\tau _\infty }$ and $p$, the next step was to evaluate the performance of the proposed adaptive control system. For this purpose, we chose the following values of the initial states of the system \eqref{eq4}, the control law parameters \eqref{eq11}, the reference $r$, the matrices $Q{\rm{,}\;}R$, the arbitrary parameters of the parameterizations \eqref{eq23}, \eqref{eq24}, \eqref{eq31} and the adaptive law \eqref{eq32} as follows:
\begin{equation}\label{eq40}
\begin{array}{c}
x\left( 0 \right) = {\begin{bmatrix}
{ - 1}&1&0
\end{bmatrix}}^{\rm T} {\rm{,}}\;r = {e^{ - 7t}}{\rm{,}}\;\hat \theta \left( 0 \right) = {{\begin{bmatrix}
{0.1}&{0.1}&{0.1}&{0.1}
\end{bmatrix}}^{\rm{T}}}{\rm{,}}\;Q = {I_{3 \times 3}}{\rm{,}}\;\\R = 1,\;
l = 2.5,{\rm{}}\;{k_0} = 10,{\rm{}}\;{k_1} = 1.8 \cdot {10^{35}}{\rm{,}}\;\sigma  = {\textstyle{\frac{5} {7}}}{\rm{,}}\;\rho  = {10^{35}}{\rm{,}}\;{\gamma _0} = 1,{\rm{}}\;{\gamma _1} = 10.
\end{array}
\end{equation}

The parameters $r{\rm{,}}\;l{\rm{,}}\;{k_0}{\rm{,}}\;{k_1}{\rm{,}}\;\sigma {\rm{,}}\;\rho {\rm{,}}\;{\gamma _0}{\rm{,}}\;{\gamma _1}$ values were set in accordance with the recommendations, which were made in Section 3. Using the results, which are presented in Figire 2, Figure 3 and Table 1, the values of the parameters ${\tau _\infty }$ and $p$ were chosen as: ${\tau _\infty } = 7,{\rm{}}\;p = 35$.

The transients of $\Omega \left( t \right){\rm{,}}\;{\textstyle{\frac{\varsigma \left( t \right)} {\Omega \left( t \right)\theta }}}$ are depicted in Figure 4.
\begin{figure}
\begin{center}
\includegraphics[scale=0.25]{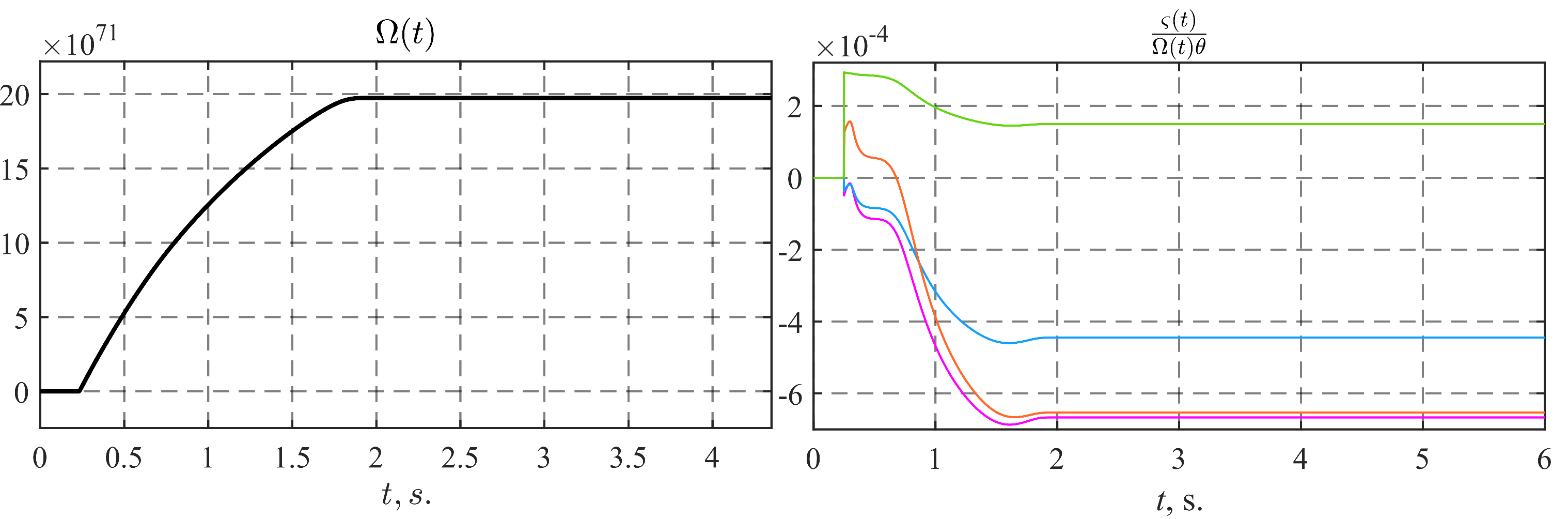}    % The printed column width is 8.4 cm.
\caption{Transients of $\Omega \left( t \right)$ and ${\textstyle{\frac{\varsigma \left( t \right)} {\Omega \left( t \right)\theta }}}$} 
\end{center}
\end{figure}  

The transients of ${\textstyle{\frac{\varsigma \left( t \right)} {\Omega \left( t \right)\theta }}}$ demonstrated that the chosen values ${\tau _\infty } = 7,{\rm{}}\;p = 35$ satisfied the equality $\varsigma \left( t \right) = o\left( {\Omega \left( t \right)\theta } \right){\rm{,}}$ while the curve of $\Omega \left( t \right)$ proved the conclusions made in Proposition \ref{proposition6}.

Figure 5 depicts the transients of the plant states $x\left( t \right)$ and adjustable parameters $\hat \theta \left( t \right)$ of the control law \eqref{eq11}.
\begin{figure}[th]
\begin{center}
\includegraphics[scale=0.25]{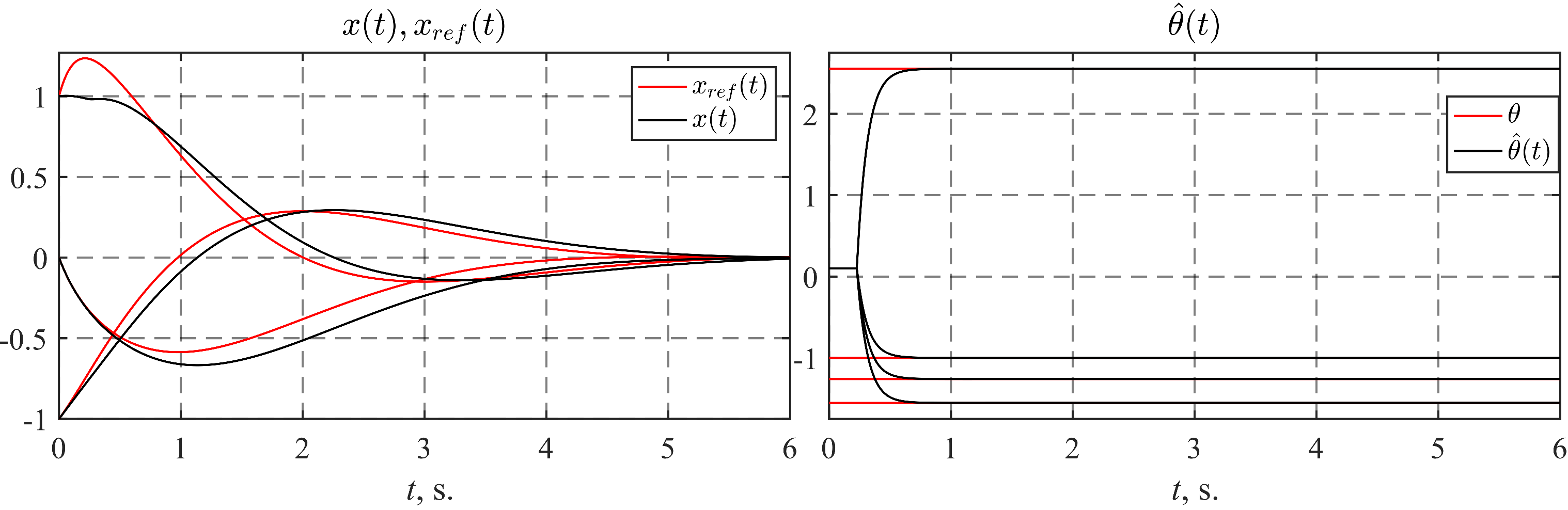}    % The printed column width is 8.4 cm.
\caption{Transients of (a) plant states $x\left( t \right)$ and (b) adjustable parameters $\hat \theta \left( t \right)$} 
\end{center}
\end{figure}  

The experimental results validated the theoretical conclusions and the remarks given in the paper. The proposed adaptive control system guaranteed that the suboptimal control objective \eqref{eq13} could be achieved for a plant with parameter uncertainty and provided monotonicity of transients of the control law adjustable parameters if the equality $\varsigma \left( t \right) = o\left( {\Omega \left( t \right)\theta } \right)$ was met. Experiments fully confirmed the recommendations on how to choose the arbitrary parameters of the adaptive system.

\subsection {\texorpdfstring{$\exists i{\rm{,}}j \in \overline {1,2n} {\rm{}\;}{\lambda _i}\left( D \right) \gg {\lambda _j}\left( D \right)$}{Lg}}

The plant \eqref{eq3} parameters were chosen in the same way as in \cite{Jha18a,Jha18b,Jha19b,Jha18c}:
\begin{equation}\label{eq41}
\forall t \ge 0{\rm{}}\;{\dot x_p}\left( t \right) = {\begin{bmatrix}
0&1&0\\
0&0&{4.438}\\
0&{ - 12}&{ - 24}
\end{bmatrix}} {x_p}\left( t \right) + {\begin{bmatrix}
0\\
0\\
{20}
\end{bmatrix}} u\left( t \right){\rm{,}}
\end{equation}
therefore, the interested reader, if necessary, can compare the experimental results, which are presented below, with the ones in \cite{Jha18a,Jha18b,Jha19b,Jha18c}.

The first experiment was to investigate the influence of parameter $\vartheta$ value on the existence of the analytical solution \eqref{eq9} of the Riccati equation \eqref{eq8} and optimization problem \eqref{eq6}. For this purpose, we chose the following values of the reference $r$, initial conditions $x\left( 0 \right)$ and matrices $Q{\rm{,}}\;R$:
\begin{equation}\label{eq42}
r = 1,{\rm{}}\;x\left( 0 \right) = {0_4}{\rm{,}}\;Q = {I_{4 \times 4}}{\rm{,}}\;R = 1.
\end{equation}

The eigenvalues of the matrix $D$ are calculated for $\vartheta  = 1$ and $\vartheta  = 100$:
\begin{equation}\label{eq43}
\begin{array}{c}
\vartheta  = 1,{\rm{}}\;{V^{ - 1}}DV = diag\left\{ \begin{array}{c}
 - {\rm{29}}{\rm{.27}}\\
 - {\rm{0}}{\rm{.78}} + {\rm{0}}{\rm{.51}}{\mathop{\rm i}\nolimits} \\
 - {\rm{0}}{\rm{.78}} - {\rm{0}}{\rm{.51}}{\mathop{\rm i}\nolimits} \\
 - {\rm{3}}{\rm{.43}}\\
{\rm{29}}{\rm{.27}}\\
{\rm{3}}{\rm{.43}}\\
{\rm{0}}{\rm{.78}} + {\rm{0}}{\rm{.51}}{\mathop{\rm i}\nolimits} \\
{\rm{0}}{\rm{.78}} - {\rm{0}}{\rm{.51}}{\mathop{\rm i}\nolimits} 
\end{array} \right\}{\rm{,\;}}
\vartheta  = 100,{\rm{}\;}{V^{ - 1}}DV = diag\left\{ \begin{array}{c}
 - {\rm{29}}{\rm{.27}}\\
 - {\rm{7}}{\rm{.11}}\\
 - {\rm{3}}{\rm{.44}} + {\rm{5}}{\rm{.55}}{\mathop{\rm i}\nolimits} \\
 - {\rm{3}}{\rm{.44}} - {\rm{5}}{\rm{.55}}{\mathop{\rm i}\nolimits} \\
{\rm{29}}{\rm{.27}}\\
{\rm{7}}{\rm{.11}}\\
{\rm{3}}{\rm{.44}} + {\rm{5}}{\rm{.55}}{\mathop{\rm i}\nolimits} \\
{\rm{3}}{\rm{.44}} - {\rm{5}}{\rm{.55}}{\mathop{\rm i}\nolimits} 
\end{array} \right\}{\rm{,}}
\end{array}
\end{equation}
where $V \in {\mathbb{R}^{8 \times 8}}$ is the matrix, which constitutes of $D$ eigenvectors, $i$ is the imaginary unit.

The equation \eqref{eq43} shows that if the parameter $\vartheta $ value is improved, then the difference between the eigenvalues of the matrix $D$ is reduced. This results in improvement of the transients rate of the Riccati equation \eqref{eq8} solution and, thereby, reduction of ${\tau _{ss}}$ value, which defines when the steady-state mode is achieved.

Figure 6 presents transients of the norms of the Riccati equation solution $P\left( t \right)$ and its dynamic analytical \eqref{eq9} estimate \linebreak $\hat P\left( t \right) = {\Phi _{21}}\left( t \right)\Phi _{11}^{ - 1}\left( t \right)$ for various values of the parameter $\vartheta$.
\begin{figure}[th]
\begin{center}
\includegraphics[scale=0.25]{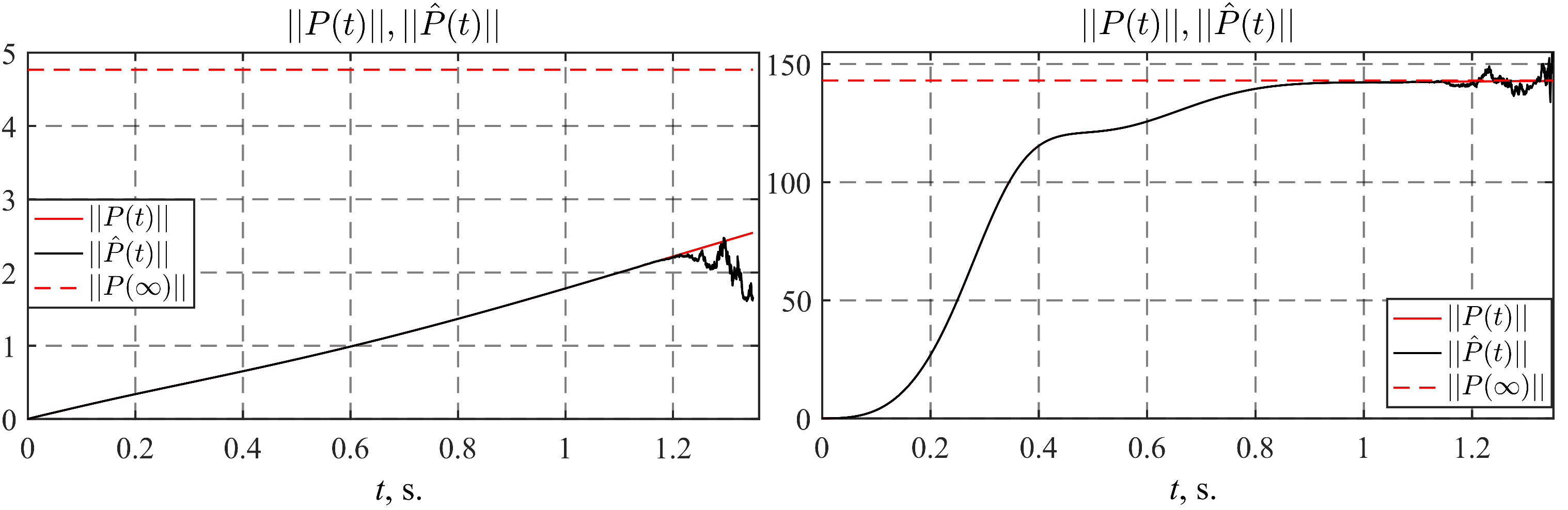}    % The printed column width is 8.4 cm.
\caption{Transients of $P\left( t \right)$ and $\hat P\left( t \right)$ when (a) $\vartheta  = 1$; (b) $\vartheta  = 100$} 
\end{center}
\end{figure}

It follows from Figure 6 that, when $\vartheta  = 1$, it was impossible to calculate analytically \eqref{eq9} the steady-state solution of the Riccati equation whatever ${\tau _\infty }$ was chosen, because the singularity of solutions \eqref{eq9} had occurred significantly earlier than the transient processes in equation \eqref{eq8} were completed {(see vibrations of black curves on the right-hand side of left figure)}. In such case, if ${\tau _\infty }$ was chosen according to Remark \ref{remark3}, the proposed adaptive control system allowed one to identify only the parameters of the stabilizing control law. At the same time, when $\vartheta  = 100$, there existed the value of ${\tau _\infty }$, which made it possible to calculate the steady-state solution of the Riccati equation analytically \eqref{eq9}, because the singularity of the analytical solutions \eqref{eq9} occurred later than the transients of the Riccati equation \eqref{eq8} solution had reached the steady state. When the value of $\vartheta $ was further improved, this effect became even stronger.

Figure 7 depicts the transients of the output $z\left( t \right)$, which were obtained using various values of the parameters $\vartheta $ and ${\tau _\infty }$.
\begin{figure}[th]
\begin{center}
\includegraphics[scale=0.25]{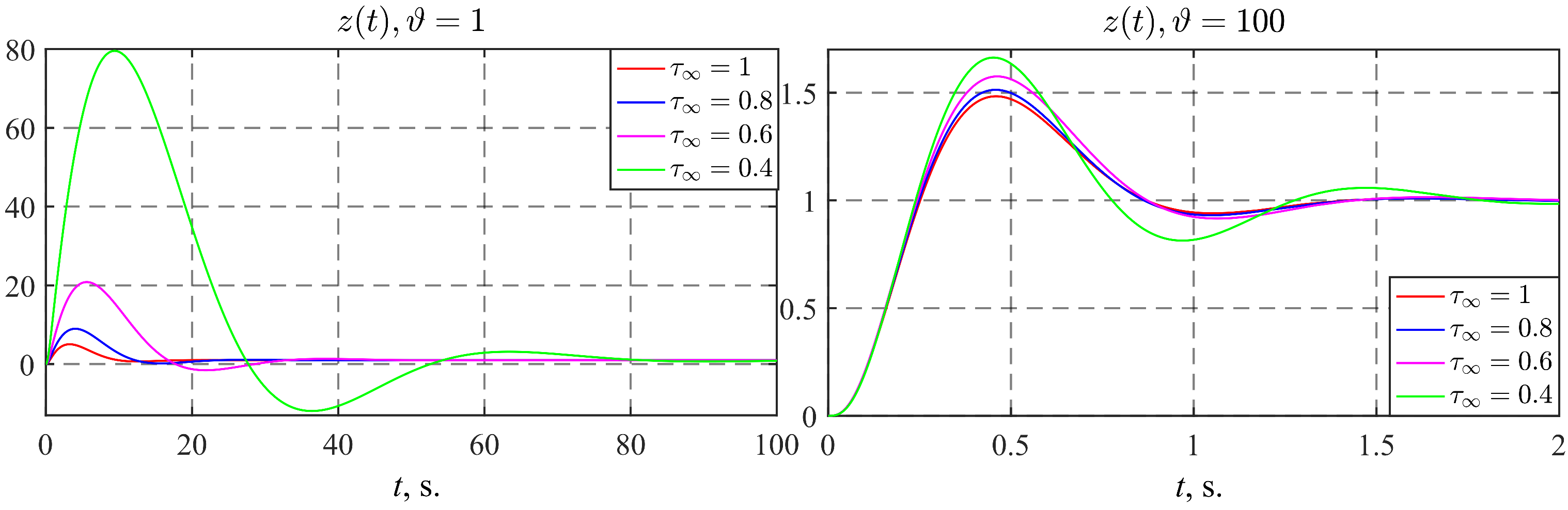}    % The printed column width is 8.4 cm.
\caption{Transients of output $z\left( t \right)$ for various values of $\vartheta $ and ${\tau _\infty }$} 
\end{center}
\end{figure}

\newpage
Figure 7 illustrates that in terms of the required control quality for the plant with parameter uncertainty \eqref{eq4}, it was advantageous to choose high values of $\vartheta$.

The next step was to evaluate the performance of the developed adaptive control system. A part of the arbitrary parameters was set according to \eqref{eq42}, and the remaining ones were chosen as follows:
\begin{equation}\label{eq44}
\begin{array}{c}
\hat \theta \left( 0 \right) = {{\begin{bmatrix}
{{0_{1 \times 4}}}&{10}
\end{bmatrix}}^{\rm{T}}}{\rm{,}\;}\vartheta {\rm{ = 100}}{\rm{,}\;}l = 10,{\rm{}\;}{k_0} = 10,{\rm{}\;}{k_1} = 1.8 \cdot {10^{35}}{\rm{,}\;}\\
p = 85,{\rm{}\;}{\tau _\infty }{\rm{ = 1}}{\rm{,}}\;\sigma  = {\textstyle{\frac{5} {7}}}{\rm{,}}\;\rho  = {10^{35}}{\rm{,}}\;{\gamma _0} = 1,{\rm{}\;}{\gamma _1} = 1,
\end{array}
\end{equation}

The transients of $\left\| {e_{ref} \left( t \right)} \right\|$ presented in Figure 8(a), the curve of $\left\| {\tilde \theta \left( t \right)} \right\|$ – in Figure 8(b).
\begin{figure}[th]
\begin{center}
\includegraphics[scale=0.25]{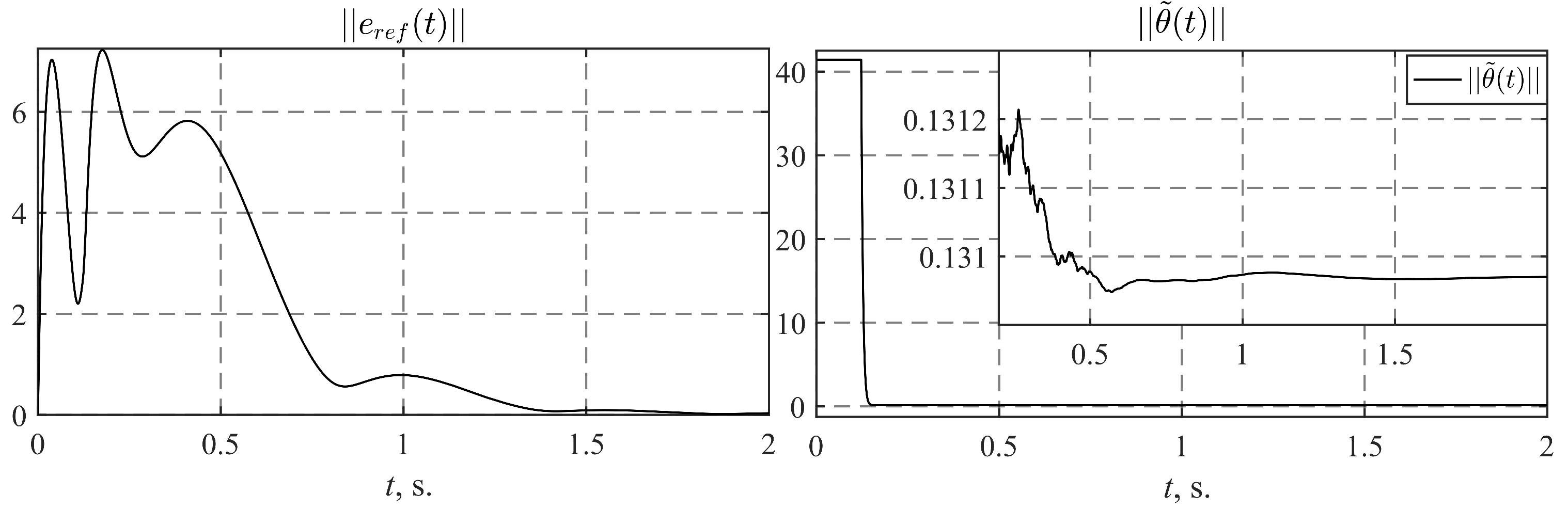}    % The printed column width is 8.4 cm.
\caption{Transients of (a) $\left\| {e_{ref} \left( t \right)} \right\|$; (b) $\left\| {\tilde \theta \left( t \right)} \right\|$} 
\end{center}
\end{figure}
  
Figure 8 confirms the statements of Theorem. The developed adaptive control system provided identification of the optimal control law \eqref{eq7} unknown parameters $\theta $ with bounded error and guaranteed the achievement of the objective \eqref{eq13}.

Thus, the main disadvantage of the proposed system is the possible singularity of the analytical solution \eqref{eq9} if \linebreak $\exists i{\rm{,}}\;j \in \overline {1,2n} {\rm{}}\;{\lambda _i}\left( D \right) \gg {\lambda _j}\left( D \right)$ and the value of ${\tau _\infty }$ is high. However, the simulation results presented in Figure 6 demonstrate that if a sufficiently small value of ${\tau _\infty }$ is chosen and appropriate value of $\vartheta $ is found, this disadvantage can be effectively overcome. On the whole, the developed system makes it possible to effectively solve the problem of synthesis of the adaptive suboptimal control system to meet \eqref{eq13}.

\section{Conclusion}\label{sec5}

The adaptive control system is proposed to estimate the unknown parameters of the control law, which are optimal in terms of a linear-quadratic cost function. The global stability of the closed-loop adaptive control system and the convergence of the tracking error to the set with adjustable bound are proved. Recommendations on the choice of the adaptive system arbitrary parameters are given. The conducted numerical experiments confirmed all theoretical results and conclusions.

Scope of the further research and improvements of the proposed adaptive suboptimal control system can be structured as follows (in descending order of importance): 1) to relax the requirement $\Psi \left( t \right) \in {\rm{FE}}$ to a weaker semi-persistent excitation \linebreak (${\Psi \left( t \right) \in}$ s-PE) or semi-finite excitation (${\Psi \left( t \right) \in}$ s-FE) {, which are simpler to be satisfied by real world signals in comparison with FE}; 2) to expand these results to systems with piecewise constant parameters {with guarantee of exponential stability and convergence of both tracking and parameter errors to a compact set with adjustable bound}; 3) not to use the analytical solution \eqref{eq9} of the Riccati equation \eqref{eq8} in \eqref{eq16} and derive a more reliable parameterization of \eqref{eq16} type in order to prevent possible singularity of solutions in \eqref{eq9} in case of ${\lambda _i}\left( D \right) \gg {\lambda _j}\left( D \right)$ and wrong choice of ${\tau _\infty }$ value {(the solution of the first problem in this list will be of great importance to achieve the goal under consideration)}; 4) to approximate the fundamental matrix $\Phi \left( {{\tau _\infty }} \right)$ by more computationally stable method in comparison with the applied Taylor series expansion (for example, with the help of the Padé approximation \cite{Moler78}).

%\backmatter

\section*{Acknowledgments}
This research was financially supported in part by Grants Council of the President of the Russian Federation (project MD-1787.2022.4). 

\subsection*{Conflict of interest}

The authors declare no potential conflict of interests.

\appendix

\section{Proof of Proposition 2.} 

In order to prove the proposition, it is assumed that the vector ${\theta _{AB}}$ and matrices ${\Phi _{11}}\left( {{\tau _\infty }} \right){\rm{,}}\;{\Phi _{21}}\left( {{\tau _\infty }} \right)$ are known. It is also taken into consideration that the following equalities hold and $\forall A \in {\mathbb{R}^{n \times n}}$:
\begin{equation}\label{eqA1}
\begin{array}{c}
{c^n}det \left\{ A \right\} = det \left\{ {cA} \right\}{\rm{,}}\;\;{c^{n - 1}}adj\left\{ A \right\} = adj\left\{ {cA} \right\}{\rm{,}}\\
adj\left\{ A \right\}A = Aadj\left\{ A \right\} = det \left\{ A \right\}{I_{n \times n}}.
\end{array}
\end{equation}

Then:

a) the definition of the matrix $P$ (see \eqref{eq9}) is multiplied by ${\Phi _{11}}\left( {{\tau _\infty }} \right)adj\left\{ {{\Phi _{11}}\left( {{\tau _\infty }} \right)} \right\}$ to obtain the static regression equation with respect to $P$:
\begin{equation}\label{eqA2}
\begin{array}{c}
{y_P} = {\Delta _P}P{\rm{,}}\\
{y_P}{\rm{:}} = {\Phi _{21}}\left( {{\tau _\infty }} \right)adj\left\{ {{\Phi _{11}}\left( {{\tau _\infty }} \right)} \right\}{\rm{,}\;}{\Delta _P}{\rm{:}} = det \left\{ {{\Phi _{11}}\left( {{\tau _\infty }} \right)} \right\}{\rm{,}}
\end{array}
\end{equation}
where ${\Delta _P} \ne 0$ as the matrix $\Phi _{11}^{ - 1}\left( {{\tau _\infty }} \right)$ exists in accordance with Proposition \ref{proposition1}.

b) the definition of the matrix $V$ (see \eqref{eq9}) is multiplied by $adj\left\{ {{A^{\rm{T}}} - PB{R^{ - 1}}{B^{\rm{T}}}} \right\}\left( {{A^{\rm{T}}} - PB{R^{ - 1}}{B^{\rm{T}}}} \right)$ to obtain the static regression equation with respect to $V$:
\begin{equation}\label{eqA3}
\begin{array}{c}
{y_V} = {\Delta _V}V\\
{y_V}{\rm{:}} = adj\left\{ {{A^{\rm{T}}} - PB{R^{ - 1}}{B^{\rm{T}}}} \right\}P{B_r}{\rm{,}}\;\;{\Delta _V}{\rm{:}} = det \left\{ {{A^{\rm{T}}} - PB{R^{ - 1}}{B^{\rm{T}}}} \right\},
\end{array}
\end{equation}
where ${\Delta _V} \ne 0$ as the matrix ${\left( {{A^{\rm{T}}} - PB{R^{ - 1}}{B^{\rm{T}}}} \right)^{ - 1}}$ exists in accordance with Proposition \ref{proposition1}.

c) the equation to calculate ${K_x}$ (see \eqref{eq7}) is multiplied by ${\Delta _P}$ to obtain the static regression equation with respect to ${K_x}$:
\begin{equation}\label{eqA4}
\begin{array}{c}
{y_{{K_x}}} = {\Delta _P}{K_x}{\rm{,}}\\
{y_{{K_x}}}{\rm{:}} =  - {R^{ - 1}}{B^{\rm{T}}}{\Phi _{21}}\left( {{\tau _\infty }} \right)adj\left\{ {{\Phi _{11}}\left( {{\tau _\infty }} \right)} \right\}{\rm{,}}
\end{array}
\end{equation}

d) considering \eqref{eqA2} and \eqref{eqA3}, the equation to calculate ${K_r}$ (see \eqref{eq7}) is multiplied by ${\Delta _V}\Delta _P^n$ to obtain the static regression equation with respect to ${K_r}$:
\begin{equation}\label{eqA5}
\begin{array}{c}
{y_{{K_r}}} = {\Delta _{{K_r}}}{K_r}{\rm{,}}\\
{y_{{K_r}}}{\rm{:}} =  - {R^{ - 1}}{B^{\rm{T}}}V{\Delta _V}\Delta _P^n =  - \Delta _P^n{R^{ - 1}}{B^{\rm{T}}}adj\left\{ {{A^{\rm{T}}} - PB{R^{ - 1}}{B^{\rm{T}}}} \right\}P{B_r} = \\
 =  - {\Delta _P}{R^{ - 1}}{B^{\rm{T}}}adj\left\{ {{A^{\rm{T}}}{\Delta _P} - {\Delta _P}PB{R^{ - 1}}{B^{\rm{T}}}} \right\}P{B_r} = \\
 =  - {R^{ - 1}}{B^{\rm{T}}}adj\left\{ \begin{array}{c}
{A^{\rm{T}}}det \left\{ {{\Phi _{11}}\left( {{\tau _\infty }} \right)} \right\} - \\
 - {\Phi _{21}}\left( {{\tau _\infty }} \right)adj\left\{ {{\Phi _{11}}\left( {{\tau _\infty }} \right)} \right\}B{R^{ - 1}}{B^{\rm{T}}}
\end{array} \right\}{\Phi _{21}}\left( {{\tau _\infty }} \right)adj\left\{ {{\Phi _{11}}\left( {{\tau _\infty }} \right)} \right\}{B_r}{\rm{,}}\\
{\Delta _{{K_r}}}{\rm{:}} = {\Delta _V}\Delta _P^n{\rm{.}}
\end{array}
\end{equation}

Then, to obtain the dynamic regression equation with respect to ${K_x}$, the equation \eqref{eqA4} is multiplied by $\varphi \left( t \right)\Delta _\Phi ^n\left( t \right)$ and, using \eqref{eqA1}, the equations \eqref{eq17} and \eqref{eq18} are substituted into the obtained result:
\begin{equation}\label{eqA6}
\begin{array}{c}
{y_{{K_x}}}\left( t \right) = {\Delta _{{K_x}}}\left( t \right){K_x} + {\varepsilon _{{K_x}}}\left( t \right){\rm{,}}\\
{y_{{K_x}}}\left( t \right){\rm{:}} = \Delta _\Phi ^n\left( t \right)\varphi \left( t \right){y_{{K_x}}} =  - \Delta _\Phi ^n\left( t \right)\varphi \left( t \right){R^{ - 1}}{B^{\rm{T}}}{\Phi _{21}}\left( {{\tau _\infty }} \right)adj\left\{ {{\Phi _{11}}\left( {{\tau _\infty }} \right)} \right\} = \\
 =  - \Delta _\Phi ^{n - 1}\left( t \right){R^{ - 1}}z_B^{\rm{T}}\left( t \right){\Delta _\Phi }\left( t \right){\Phi _{21}}\left( {{\tau _\infty }} \right)adj\left\{ {{\Phi _{11}}\left( {{\tau _\infty }} \right)} \right\}
 =  - {R^{ - 1}}z_B^{\rm{T}}\left( t \right){\Delta _\Phi }\left( t \right){\Phi _{21}}\left( {{\tau _\infty }} \right)adj\left\{ {{\Delta _\Phi }\left( t \right){\Phi _{11}}\left( {{\tau _\infty }} \right)} \right\} = \\
 =  - {R^{ - 1}}z_B^{\rm{T}}\left( t \right)\left( {{z_{{\Phi _{21}}}}\left( t \right) - {\varepsilon _{{\Phi _{21}}}}\left( t \right)} \right)adj\left\{ {{z_{{\Phi _{11}}}}\left( t \right) - {\varepsilon _{{\Phi _{11}}}}\left( t \right)} \right\}
 \buildrel \Delta \over =  - {R^{ - 1}}z_B^{\rm{T}}\left( t \right){z_{{\Phi _{21}}}}\left( t \right)adj\left\{ {{z_{{\Phi _{11}}}}\left( t \right)} \right\}{\rm{,}}\\
{\Delta _{{K_x}}}\left( t \right){\rm{:}} = \Delta _\Phi ^n\left( t \right)\varphi \left( t \right){\Delta _P} = \Delta _\Phi ^n\left( t \right)\varphi \left( t \right)det \left\{ {{\Phi _{11}}\left( {{\tau _\infty }} \right)} \right\}
 = \varphi \left( t \right)det \left\{ {{\Delta _\Phi }\left( t \right){\Phi _{11}}\left( {{\tau _\infty }} \right)} \right\} \buildrel \Delta \over = \varphi \left( t \right)det \left\{ {{z_{{\Phi _{11}}}}\left( t \right)} \right\}{\rm{,}}\\
{\varepsilon _{{K_x}}}\left( t \right){\rm{:}} = {R^{ - 1}}z_B^{\rm{T}}\left( t \right){z_{{\Phi _{21}}}}\left( t \right)adj\left\{ { - {\varepsilon _{{\Phi _{11}}}}\left( t \right)} \right\}
-{R^{ - 1}}z_B^{\rm{T}}\left( t \right){\varepsilon _{{\Phi _{21}}}}\left( t \right)adj\left\{ {{z_{{\Phi _{11}}}}\left( t \right) - {\varepsilon _{{\Phi _{11}}}}\left( t \right)} \right\} + \\
 + \varphi \left( t \right)\left( {det \left\{ {{\Delta _\Phi }\left( t \right){\Phi _{11}}\left( {{\tau _\infty }} \right)} \right\} - det \left\{ {{z_{{\Phi _{11}}}}\left( t \right)} \right\}} \right){K_x}.
\end{array}
\end{equation}

From now on the sign $ \buildrel \Delta \over = $ omits some intermediate routine derivations, which can be found in \cite{Glushchenko22d}, and presents the final result.

Following the same strategy, the equation \eqref{eqA5} is multiplied by ${\varphi ^{2n}}\left( t \right)\Delta _\Phi ^{{n^2}}\left( t \right)$, and the equations \eqref{eq17} and \eqref{eq18} are substituted into the obtained result to obtain the dynamic regression equation with respect to ${K_r}$:
\begin{equation}\label{eqA7}
\begin{array}{c}
{y_{{K_r}}}\left( t \right) = {\Delta _{{K_r}}}{K_r}\left( t \right) + {\varepsilon _{{K_r}}}\left( t \right){\rm{,}}\\
{y_{{K_r}}}\left( t \right){\rm{:}} = {\varphi ^{2n}}\left( t \right)\Delta _\Phi ^{{n^2}}\left( t \right){y_{{K_r}}}
 =  - {\varphi ^2}\left( t \right){\varphi ^{2\left( {n - 1} \right)}}\left( t \right)\Delta _\Phi ^{n\left( {n - 1} \right)}\left( t \right)\Delta _\Phi ^n\left( t \right){R^{ - 1}}{B^{\rm{T}}} \times \\
 \times adj\left\{ \begin{array}{c}
{A^{\rm{T}}}det \left\{ {{\Phi _{11}}\left( {{\tau _\infty }} \right)} \right\} - \\
 - {\Phi _{21}}\left( {{\tau _\infty }} \right)adj\left\{ {{\Phi _{11}}\left( {{\tau _\infty }} \right)} \right\}B{R^{ - 1}}{B^{\rm{T}}}
\end{array} \right\}{\Phi _{21}}\left( {{\tau _\infty }} \right)adj\left\{ {{\Phi _{11}}\left( {{\tau _\infty }} \right)} \right\}{B_r} = \\
 =  - {\varphi ^2}\left( t \right){\varphi ^{2\left( {n - 1} \right)}}\left( t \right)\Delta _\Phi ^{n\left( {n - 1} \right)}\left( t \right){R^{ - 1}}{B^{\rm{T}}} adj\left\{ \begin{array}{c}
{A^{\rm{T}}}det \left\{ {{\Phi _{11}}\left( {{\tau _\infty }} \right)} \right\} - \\
 - {\Phi _{21}}\left( {{\tau _\infty }} \right)adj\left\{ {{\Phi _{11}}\left( {{\tau _\infty }} \right)} \right\}B{R^{ - 1}}{B^{\rm{T}}}
\end{array} \right\} \times \\
 \times {\Delta _\Phi }\left( t \right){\Phi _{21}}\left( {{\tau _\infty }} \right)adj\left\{ {{\Delta _\Phi }\left( t \right){\Phi _{11}}\left( {{\tau _\infty }} \right)} \right\}{B_r} = \\
 =  - {\varphi ^2}\left( t \right){\varphi ^{2\left( {n - 1} \right)}}\left( t \right){R^{ - 1}}{B^{\rm{T}}}adj\left\{ \begin{array}{c}
{A^{\rm{T}}}\Delta _\Phi ^n\left( t \right)det \left\{ {{\Phi _{11}}\left( {{\tau _\infty }} \right)} \right\} - \\
 - \Delta _\Phi ^n\left( t \right){\Phi _{21}}\left( {{\tau _\infty }} \right)adj\left\{ {{\Phi _{11}}\left( {{\tau _\infty }} \right)} \right\}B{R^{ - 1}}{B^{\rm{T}}}
\end{array} \right\}\times\\
 \times {\Delta _\Phi }\left( t \right){\Phi _{21}}\left( {{\tau _\infty }} \right)adj\left\{ {{\Delta _\Phi }\left( t \right){\Phi _{11}}\left( {{\tau _\infty }} \right)} \right\}{B_r} = \\
 =  - \varphi \left( t \right){\varphi ^{2\left( {n - 1} \right)}}\left( t \right){R^{ - 1}}z_B^{\rm{T}}\left( t \right)  adj\left\{ \begin{array}{c}
{A^{\rm{T}}}det \left\{ {{\Delta _\Phi }\left( t \right){\Phi _{11}}\left( {{\tau _\infty }} \right)} \right\} - \\
 - {\Delta _\Phi }\left( t \right){\Phi _{21}}\left( {{\tau _\infty }} \right)adj\left\{ {{\Delta _\Phi }\left( t \right){\Phi _{11}}\left( {{\tau _\infty }} \right)} \right\}B{R^{ - 1}}{B^{\rm{T}}}
\end{array} \right\} \times \\
 \times {\Delta _\Phi }\left( t \right){\Phi _{21}}\left( {{\tau _\infty }} \right)adj\left\{ {{\Delta _\Phi }\left( t \right){\Phi _{11}}\left( {{\tau _\infty }} \right)} \right\}{B_r} = \\
 =  - \varphi \left( t \right){R^{ - 1}}z_B^{\rm{T}}\left( t \right) adj\left\{ \begin{array}{c}
\varphi \left( t \right)z_A^{\rm{T}}\left( t \right)det \left\{ {{\Delta _\Phi }\left( t \right){\Phi _{11}}\left( {{\tau _\infty }} \right)} \right\} - \\
 - {\Delta _\Phi }\left( t \right){\Phi _{21}}\left( {{\tau _\infty }} \right)adj\left\{ {{\Delta _\Phi }\left( t \right){\Phi _{11}}\left( {{\tau _\infty }} \right)} \right\}{z_B}\left( t \right){R^{ - 1}}z_B^{\rm{T}}\left( t \right)
\end{array} \right\} \times \\
 \times {\Delta _\Phi }\left( t \right){\Phi _{21}}\left( {{\tau _\infty }} \right)adj\left\{ {{\Delta _\Phi }\left( t \right){\Phi _{11}}\left( {{\tau _\infty }} \right)} \right\}{B_r} \buildrel \Delta \over = \\
 \buildrel \Delta \over =  - \varphi \left( t \right){R^{ - 1}}z_B^{\rm{T}}\left( t \right)adj\left\{ \begin{array}{c}
\varphi \left( t \right)z_A^{\rm{T}}\left( t \right)det \left\{ {{z_{{\Phi _{11}}}}\left( t \right)} \right\} - \\
 - {z_{{\Phi _{21}}}}\left( t \right)adj\left\{ {{z_{{\Phi _{11}}}}\left( t \right)} \right\}{z_B}\left( t \right){R^{ - 1}}z_B^{\rm{T}}\left( t \right)
\end{array} \right\} {z_{{\Phi _{21}}}}\left( t \right)adj\left\{ {{z_{{\Phi _{11}}}}\left( t \right)} \right\}{B_r}{\rm{ = }}\\
{\rm{ = }} - \varphi \left( t \right){R^{ - 1}}z_B^{\rm{T}}\left( t \right)adj\left\{ {z_A^{\rm{T}}\left( t \right){\Delta _{{K_x}}}\left( t \right) + y_{{K_x}}^{\rm{T}}\left( t \right)z_B^{\rm{T}}\left( t \right)} \right\}  {z_{{\Phi _{21}}}}\left( t \right)adj\left\{ {{z_{{\Phi _{11}}}}\left( t \right)} \right\}{B_r}{\rm{,}}
\end{array}
\end{equation}

\begin{displaymath}
\begin{array}{c}
{\Delta _{{K_r}}}\left( t \right) = {\varphi ^{2n}}\left( t \right)\Delta _\Phi ^{{n^2}}\left( t \right){\Delta _V}\Delta _P^n = {\varphi ^{2n}}\left( t \right)\Delta _\Phi ^{{n^2}}\left( t \right)det \left\{ {{A^{\rm{T}}} - PB{R^{ - 1}}{B^{\rm{T}}}} \right\}\Delta _P^n = \\
 = \Delta _\Phi ^{{n^2}}\left( t \right)det \left\{ {{\varphi ^2}\left( t \right){A^{\rm{T}}} - {\varphi ^2}\left( t \right)PB{R^{ - 1}}{B^{\rm{T}}}} \right\}\Delta _P^n = \\
 = \Delta _\Phi ^{{n^2}}\left( t \right)det \left\{ {{\Delta _P}\varphi \left( t \right)z_A^{\rm{T}}\left( t \right) - {\Delta _P}P{z_B}\left( t \right){R^{ - 1}}z_B^{\rm{T}}\left( t \right)} \right\} = \\
 = \Delta _\Phi ^{{n^2}}\left( t \right)det \left\{ \begin{array}{c}
det \left\{ {{\Phi _{11}}\left( {{\tau _\infty }} \right)} \right\}\varphi \left( t \right)z_A^{\rm{T}}\left( t \right) - \\
 - {\Phi _{21}}\left( {{\tau _\infty }} \right)adj\left\{ {{\Phi _{11}}\left( {{\tau _\infty }} \right)} \right\}{z_B}\left( t \right){R^{ - 1}}z_B^{\rm{T}}\left( t \right)
\end{array} \right\} = \\
 = det \left\{ \begin{array}{c}
\Delta _\Phi ^n\left( t \right)det \left\{ {{\Phi _{11}}\left( {{\tau _\infty }} \right)} \right\}\varphi \left( t \right)z_A^{\rm{T}}\left( t \right) - \\
 - \Delta _\Phi ^n\left( t \right){\Phi _{21}}\left( {{\tau _\infty }} \right)adj\left\{ {{\Phi _{11}}\left( {{\tau _\infty }} \right)} \right\}{z_B}\left( t \right){R^{ - 1}}z_B^{\rm{T}}\left( t \right)
\end{array} \right\} = \\
 = det \left\{ \begin{array}{c}
det \left\{ {{\Delta _\Phi }\left( t \right){\Phi _{11}}\left( {{\tau _\infty }} \right)} \right\}\varphi \left( t \right)z_A^{\rm{T}}\left( t \right) - \\
 - {\Delta _\Phi }\left( t \right){\Phi _{21}}\left( {{\tau _\infty }} \right)adj\left\{ {{\Delta _\Phi }\left( t \right){\Phi _{11}}\left( {{\tau _\infty }} \right)} \right\}{z_B}\left( t \right){R^{ - 1}}z_B^{\rm{T}}\left( t \right)
\end{array} \right\} \buildrel \Delta \over = \\
 \buildrel \Delta \over = det \left\{ {det \left\{ {{z_{{\Phi _{11}}}}\left( t \right)} \right\}\varphi \left( t \right)z_A^{\rm{T}}\left( t \right) - {z_{{\Phi _{21}}}}\left( t \right)adj\left\{ {{z_{{\Phi _{11}}}}\left( t \right)} \right\}{z_B}\left( t \right){R^{ - 1}}z_B^{\rm{T}}\left( t \right)} \right\} = \\
 = det \left\{ {{\Delta _{{K_x}}}\left( t \right)z_A^{\rm{T}}\left( t \right) + y_{{K_x}}^{\rm{T}}\left( t \right)z_B^{\rm{T}}\left( t \right)} \right\},
\end{array}
\end{displaymath}
where ${\varepsilon _{{K_r}}}\left( t \right)$ is obtained via simple, but tedious speculations:
\begin{displaymath}
\begin{array}{c}
{\varepsilon _{{K_r}}}\left( t \right) = det \left\{ {det \left\{ {{z_{{\Phi _{11}}}}\left( t \right)} \right\}\varphi \left( t \right)z_A^{\rm{T}}\left( t \right) - } \right.
{z_{{\Phi _{21}}}}\left( t \right)adj\left\{ {{z_{{\Phi _{11}}}}\left( t \right)} \right\}{z_B}\left( t \right){R^{ - 1}}z_B^{\rm{T}}\left( t \right) + \\
 + \left( {det \left\{ {{z_{{\Phi _{11}}}}\left( t \right) - {\varepsilon _{{\Phi _{11}}}}\left( t \right)} \right\} - det \left\{ {{z_{{\Phi _{11}}}}\left( t \right)} \right\}} \right)\varphi \left( t \right)z_A^{\rm{T}}\left( t \right)
 - {z_{{\Phi _{21}}}}\left( t \right)adj\left\{ { - {\varepsilon _{{\Phi _{11}}}}\left( t \right)} \right\}{z_B}\left( t \right){R^{ - 1}}z_B^{\rm{T}}\left( t \right) + \\
 + \left. {{\varepsilon _{{\Phi _{21}}}}\left( t \right)adj\left\{ {{\Delta _\Phi }\left( t \right){\Phi _{11}}\left( {{\tau _\infty }} \right)} \right\}{z_B}\left( t \right){R^{ - 1}}z_B^{\rm{T}}\left( t \right)} \right\}{K_r}
 - det \left\{ \begin{array}{c}
det \left\{ {{z_{{\Phi _{11}}}}\left( t \right)} \right\}\varphi \left( t \right)z_A^{\rm{T}}\left( t \right) - \\
 - {z_{{\Phi _{21}}}}\left( t \right)adj\left\{ {{z_{{\Phi _{11}}}}\left( t \right)} \right\}{z_B}\left( t \right){R^{ - 1}}z_B^{\rm{T}}\left( t \right)
\end{array} \right\}{K_r} - \\
 - \varphi \left( t \right){R^{ - 1}}z_B^{\rm{T}}\left( t \right) adj\left\{ \begin{array}{c}
\varphi \left( t \right)z_A^{\rm{T}}\left( t \right)det \left\{ {{\Delta _\Phi }\left( t \right){\Phi _{11}}\left( {{\tau _\infty }} \right)} \right\} - \\
 - {\Delta _\Phi }\left( t \right){\Phi _{21}}\left( {{\tau _\infty }} \right)adj\left\{ {{\Delta _\Phi }\left( t \right){\Phi _{11}}\left( {{\tau _\infty }} \right)} \right\}{z_B}\left( t \right){R^{ - 1}}z_B^{\rm{T}}\left( t \right)
\end{array} \right\} \times \\
 \times {\varepsilon _{{\Phi _{21}}}}\left( t \right)adj\left\{ {{\Delta _\Phi }\left( t \right){\Phi _{11}}\left( {{\tau _\infty }} \right)} \right\}{B_r} + \varphi \left( t \right){R^{ - 1}}z_B^{\rm{T}}\left( t \right) adj\left\{ \begin{array}{c}
\varphi \left( t \right)z_A^{\rm{T}}\left( t \right)det \left\{ {{\Delta _\Phi }\left( t \right){\Phi _{11}}\left( {{\tau _\infty }} \right)} \right\} - \\
 - {\Delta _\Phi }\left( t \right){\Phi _{21}}\left( {{\tau _\infty }} \right)adj\left\{ {{\Delta _\Phi }\left( t \right){\Phi _{11}}\left( {{\tau _\infty }} \right)} \right\}{z_B}\left( t \right){R^{ - 1}}z_B^{\rm{T}}\left( t \right)
\end{array} \right\} \times \\
 \times {z_{{\Phi _{21}}}}\left( t \right)adj\left\{ { - {\varepsilon _{{\Phi _{11}}}}\left( t \right)} \right\}{B_r} + \varphi \left( t \right){R^{ - 1}}z_B^{\rm{T}}\left( t \right) adj\left\{ {{\varepsilon _{{\Phi _{21}}}}\left( t \right)adj\left\{ {{\Delta _\Phi }\left( t \right){\Phi _{11}}\left( {{\tau _\infty }} \right)} \right\}{z_B}\left( t \right){R^{ - 1}}z_B^{\rm{T}}\left( t \right)} \right\} {z_{{\Phi _{21}}}}\left( t \right)adj\left\{ {{z_{{\Phi _{11}}}}\left( t \right)} \right\}{B_r} + \\
 + \varphi \left( t \right){R^{ - 1}}z_B^{\rm{T}}\left( t \right)adj\left\{ { - {z_{{\Phi _{21}}}}\left( t \right)adj\left\{ { - {\varepsilon _{{\Phi _{11}}}}\left( t \right)} \right\}{z_B}\left( t \right){R^{ - 1}}z_B^{\rm{T}}\left( t \right)} \right\} {z_{{\Phi _{21}}}}\left( t \right)adj\left\{ {{z_{{\Phi _{11}}}}\left( t \right)} \right\}{B_r} + \varphi \left( t \right){R^{ - 1}}z_B^{\rm{T}}\left( t \right) \times \\
 \times adj\left\{ {\varphi \left( t \right)z_A^{\rm{T}}\left( t \right)\left( {det \left\{ {{z_{{\Phi _{11}}}}\left( t \right) - {\varepsilon _{{\Phi _{11}}}}\left( t \right)} \right\} - det \left\{ {{z_{{\Phi _{11}}}}\left( t \right)} \right\}} \right)} \right\} {z_{{\Phi _{21}}}}\left( t \right)adj\left\{ {{z_{{\Phi _{11}}}}\left( t \right)} \right\}{B_r}.
\end{array}
\end{displaymath}

More details on derivation of the equations \eqref{eqA6} and \eqref{eqA7} and ${\varepsilon _{{K_r}}}\left( t \right)$ can be found in the supplementary material \cite{Glushchenko22d}.

Having combined the equations \eqref{eqA6} and \eqref{eqA7}, the dynamic regression equation  \eqref{eq16},  \eqref{eq19} with respect to the vector of the unknown parameters $\theta $ is obtained.

It follows from the definitions of the disturbances ${\varepsilon _{{K_x}}}\left( t \right)$ and ${\varepsilon _{{K_r}}}\left( t \right)$ that
\begin{displaymath}
{\varepsilon _{{\Phi _{11}}}}\left( t \right) = {\varepsilon _{{\Phi _{21}}}}\left( t \right) = {0_{n \times n}} \Rightarrow \left\{ \begin{array}{l}
{\varepsilon _{{K_r}}}\left( t \right) = {0_{m \times m}}\\
{\varepsilon _{{K_x}}}\left( t \right) = {0_{m \times n}}
\end{array} \right.{\rm{,}}
\end{displaymath}
so the implication ${\varepsilon _{{\Phi _{11}}}}\left( t \right) = {\varepsilon _{{\Phi _{21}}}}\left( t \right) = {0_{n \times n}} \Rightarrow {\varepsilon _\theta }\left( t \right) = {0_{\left( {n + m} \right) \times m}}$ holds, which completes the proof of the Proposition.

\section{Proof of Proposition 3.}

The upper bound of $\varepsilon $ is written as:
\begin{equation}\label{eqA8}
\begin{array}{c}
\left\| \varepsilon  \right\| = \left\| {\sum\limits_{k = p + 1}^\infty  {{\textstyle{1 \over {k{\rm{!}}}}}{D^k}\tau _\infty ^k} } \right\| = \left\| {\frac{{\tau _\infty ^{p + 1}}}{{\left( {p + 1} \right){\rm{!}}}}{D^{p + 1}} + \frac{{\tau _\infty ^{p + 2}}}{{\left( {p + 2} \right){\rm{!}}}}{D^{p + 2}} + \frac{{\tau _\infty ^{p + 3}}}{{\left( {p + 3} \right){\rm{!}}}}{D^{p + 3}} +  \ldots  + } \right\| \le \\
 \le \frac{{\tau _\infty ^p}}{{\left( {p + 1} \right){\rm{!}}}}{\left\| D \right\|^p}\left\| {{\tau _\infty }D + \frac{{\tau _\infty ^2}}{{\left( {p + 2} \right)}}{D^2} + \frac{{\tau _\infty ^3}}{{\left( {p + 2} \right)\left( {p + 3} \right)}}{D^3} +  \ldots  + } \right\| \le \\
 \le \frac{{\tau _\infty ^p}}{{\left( {p + 1} \right){\rm{!}}}}{\left\| D \right\|^p}\left\| {{\tau _\infty }D + \frac{{\tau _\infty ^2}}{{2{\rm{!}}}}{D^2} + \frac{{\tau _\infty ^3}}{{3{\rm{!}}}}{D^3} +  \ldots  + } \right\| \le \frac{{\tau _\infty ^p}}{{\left( {p + 1} \right){\rm{!}}}}{\left\| D \right\|^p}\left( {{e^{\left\| D \right\|{\tau _\infty }}} - 1} \right){\rm{,}}
\end{array}
\end{equation}
whence it follows that the equation \eqref{eq27} holds.

When ${\tau _\infty } \in {L_\infty }$, it follows from \eqref{eqA8} that:
\begin{equation}\label{eqA9}
\begin{array}{c}
{\varepsilon _{{\rm{max}}}} = \frac{{{c_2}c_1^p}}{{\left( {p + 1} \right){\rm{!}}}} = \frac{{{c_2}c_1^p}}{{\left( {p + 1} \right)p{\rm{!}}}} = \frac{{{c_2}}}{{p + 1}}\left( {\frac{{{c_1}}}{1}\frac{{{c_1}}}{2}\frac{{{c_1}}}{3} \ldots \frac{{{c_1}}}{p}} \right) \le \frac{{c_1^k{c_2}}}{{p + 1}}{\left( {\frac{{{c_1}}}{k}} \right)^{p - k}}{\rm{,}}\\
{c_1} = {\tau _\infty }\left\| D \right\|{\rm{,}\;}{c_2} = {e^{\left\| D \right\|{\tau _\infty }}} - 1,{\rm{}}\;k > {c_1}{\rm{,}}
\end{array}
\end{equation}

Therefore, it holds that $\mathop {\lim }\limits_{p \to \infty } {\varepsilon _{\max }} = 0$, as was to be proved in the proposition.

\section{Proof of Proposition 4.} 

To prove the statement (a), the definitions of ${\Delta _\Phi }\left( t \right){\rm{,}}\;{\Delta _{{K_x}}}\left( t \right){\rm{,}}\;{\Delta _{{K_r}}}\left( t \right){\rm{,}}\;{\Delta _V}{\rm{,}}\;{\Delta _P}$ are substituted into the equation of $\Delta \left( t \right)$:
\begin{equation}\label{eqA10}
\begin{array}{c}
\Delta \left( t \right){\rm{:}} = \Delta _{{K_x}}^n\left( t \right)\Delta _{{K_r}}^m\left( t \right) = \\
 = \Delta _\Phi ^{{n^2}}\left( t \right){\varphi ^n}\left( t \right){det ^n}\left\{ {{\Phi _{11}}\left( {{\tau _\infty }} \right)} \right\}{\varphi ^{2mn}}\left( t \right)\Delta _\Phi ^{m{n^2}}\left( t \right){det ^m}\left\{ {{A^{\rm{T}}} - PB{R^{ - 1}}{B^{\rm{T}}}} \right\}
 {det ^{mn}}\left\{ {{\Phi _{11}}\left( {{\tau _\infty }} \right)} \right\} = \\
 = {\varphi ^{2mn + n}}\left( t \right)\Delta _\Phi ^{m{n^2} + {n^2}}\left( t \right){det ^m}\left\{ {{A^{\rm{T}}} - PB{R^{ - 1}}{B^{\rm{T}}}} \right\}{det ^{mn + n}}\left\{ {{\Phi _{11}}\left( {{\tau _\infty }} \right)} \right\} = \\
 = {\varphi ^{2mn + n}}\left( t \right){\varphi ^{2p\left( {m + 1} \right){n^2}}}\left( t \right){det ^m}\left\{ {{A^{\rm{T}}} - PB{R^{ - 1}}{B^{\rm{T}}}} \right\}{det ^{mn + n}}\left\{ {{\Phi _{11}}\left( {{\tau _\infty }} \right)} \right\} = \\
 = {\varphi ^q}\left( t \right){det ^m}\left\{ {{A^{\rm{T}}} - L_2^{\rm{T}}\Phi \left( {{\tau _\infty }} \right){L_1}{{\left( {L_1^{\rm{T}}\Phi \left( {{\tau _\infty }} \right){L_1}} \right)}^{ - 1}}B{R^{ - 1}}{B^{\rm{T}}}} \right\}{det ^{mn + n}}\left\{ {L_1^{\rm{T}}\Phi \left( {{\tau _\infty }} \right){L_1}} \right\}.
\end{array}
\end{equation}
where $q = \left( {m + 1} \right)2p{n^2} + 2mn + n{\rm{.}}$

As $\varphi \left( t \right) \in \left[ {0{\rm{;\;1}}} \right)$ according to \eqref{eq24} and, following \eqref{eq30}, the matrices ${L_1}{\rm{,}}\;{L_2}$ are bounded, then in order to prove that $\Delta \left( t \right) \in {L_\infty }$ it is necessary and sufficient to show that $R{\rm{,}}\;A{\rm{,}}\;B{\rm{,}}\;\Phi \left( {{\tau _\infty }} \right)$ are bounded. In their turn, the matrices $A$ and $B$ are bounded as the pair $\left( {A{\rm{,}}\;B} \right)$ is controllable, and $R \in {L_\infty }$ is bounded because of its definition. The matrix exponential $\Phi \left( {{\tau _\infty }} \right) = {e^{D{\tau _\infty }}}$ is bounded if $D$ and ${\tau _\infty }$ are bounded. The matrix $D$ is bounded as $A$ and $B$ are bounded, $Q$ and $R$ are bounded because of their definition. Then $\Delta \left( t \right) \in {L_\infty }$ if ${\tau _\infty } \in {L_\infty }$, as was to be proved in the statement (a).

The next step is to prove the statement (b). 

According to Assumption \ref{assumption1} and Remark \ref{remark4}, the implication $\Psi \left( t \right) \in {\rm{FE}} \Rightarrow \varphi \left( t \right) \in {\rm{FE}}$ holds. Then, to prove Proposition \ref{proposition4}, it is enough to show that $\varphi \left( t \right) \in {\rm{FE}} \Rightarrow \Delta \left( t \right) \in {\rm{FE}}$ holds. Considering \eqref{eqA10}, the finite excitation definition \eqref{eq1} is written for $\Delta \left( t \right)$:
\begin{equation}\label{eqA11}
\int\limits_{t_r^ + }^{{t_e}} {{\Delta ^2}\left( \tau  \right)} {\rm{}\;}d\tau  = {det} ^{2m}\left\{ {{A^{\rm{T}}} - PB{R^{ - 1}}{B^{\rm{T}}}} \right\}{det }^{2mn + 2n}\left\{ {{\Phi _{11}}\left( {{\tau _\infty }} \right)} \right\}\int\limits_{t_r^ + }^{{t_e}} {{\varphi ^{2q}}\left( \tau  \right)} {\rm{}\;}d\tau.
\end{equation}

As, following Proposition \ref{proposition1}, $\exists \Phi _{11}^{ - 1}\left( {{\tau _\infty }} \right){\rm{,}}\;\exists {\left( {{A^{\rm{T}}} - PB{R^{ - 1}}{B^{\rm{T}}}} \right)^{ - 1}}$, then the inequality ${det ^{2m}}\left\{ {{A^{\rm{T}}} - PB{R^{ - 1}}{B^{\rm{T}}}} \right\}{det ^{2mn + 2n}}\left\{ {{\Phi _{11}}\left( {{\tau _\infty }} \right)} \right\} > 0$ holds. Therefore, to prove the proposition, it remains to show that the following inequality is true:
\begin{equation}\label{eqA12}
\int\limits_{t_r^ + }^{{t_e}} {{\varphi ^{2q}}\left( \tau  \right)} {\rm{}}\;d\tau  \ge \overline \alpha  > 0
\end{equation}
where $\overline \alpha $  is the excitation level of the regressor ${\varphi ^q}\left( t \right)$.

To that end, the finite excitation definition \eqref{eq1} is written for $\varphi \left( t \right)$:
\begin{equation}\label{eqA13}
\int\limits_{t_r^ + }^{{t_e}} {{\varphi ^2}\left( \tau  \right)d\tau }  \ge \alpha.
\end{equation}

It is easy to check that the inequality \eqref{eqA13} holds when $\exists {t_0} \in \left[ {t_r^ + {\rm{;}}\;{t_e}} \right]{\rm{}}\;{\varphi ^2}\left( {{t_0}} \right) > 0$. Indeed, let ${\varphi ^2}\left( {{t_0}} \right) = \beta  > 0$, then, as the function ${\varphi ^2}\left( t \right)$ is continuous, there exists such a neighborhood of ${t_0}$ that for each time segment $\left[ {{t_a}{\rm{;}}\;{t_b}} \right]$, which entirely belongs to this neighborhood, the inequality ${\varphi ^2}\left( t \right) \ge {\textstyle{\beta  \over 2}}$ holds. Then, an equation, which is equivalent to \eqref{eqA13}, also holds:
\begin{equation}\label{eqA14}
\begin{array}{c}
\int\limits_{t_r^ + }^{{t_e}} {{\varphi ^2}\left( \tau  \right)d\tau }  = \int\limits_{t_r^ + }^{{t_a}} {{\varphi ^2}\left( \tau  \right)d\tau }  + \int\limits_{{t_a}}^{{t_b}} {{\varphi ^2}\left( \tau  \right)d\tau }  + \int\limits_{{t_b}}^{{t_e}} {{\varphi ^2}\left( \tau  \right)d\tau } 
 \ge \int\limits_{{t_a}}^{{t_b}} {{\varphi ^2}\left( \tau  \right)d\tau }  \ge {\textstyle{\beta  \over 2}}\int\limits_{{t_a}}^{{t_b}} {1d\tau }  = {\textstyle{\beta  \over 2}}\left( {{t_b} - {t_a}} \right) = \alpha {\rm{,}}
\end{array}
\end{equation}

Therefore, when $\varphi \left( t \right) \in {\rm{FE}}$, the inequality ${\varphi ^{2q}}\left( t \right) \ge {\textstyle{{{\beta ^{2q}}} \over {{2^{2q}}}}}$ is true over $\left[ {{t_a}{\rm{;}}\;{t_b}} \right]$. As a result, the equations \eqref{eqA11} and \eqref{eqA12} are rewritten as:
\begin{equation}\label{eqA15}
\begin{array}{c}
\int\limits_{t_r^ + }^{{t_e}} {{\varphi ^{2q}}\left( \tau  \right)} {\rm{}\;}d\tau  = \int\limits_{t_r^ + }^{{t_a}} {{\varphi ^{2q}}\left( \tau  \right)d\tau }  + \int\limits_{{t_a}}^{{t_b}} {{\varphi ^{2q}}\left( \tau  \right)d\tau }  + \int\limits_{{t_b}}^{{t_e}} {{\varphi ^{2q}}\left( \tau  \right)d\tau }  \ge \\
 \ge \int\limits_{{t_a}}^{{t_b}} {{\varphi ^{2q}}\left( \tau  \right)d\tau }  \ge {\left( {{\textstyle{\beta  \over 2}}} \right)^{2q}}\int\limits_{{t_a}}^{{t_b}} {1d\tau }  = {\left( {{\textstyle{\beta  \over 2}}} \right)^{2q}}\left( {{t_b} - {t_a}} \right) = \overline \alpha  > 0,\\
\int\limits_{t_r^ + }^{{t_e}} {{\Delta ^2}\left( \tau  \right)} {\rm{}\;}d\tau  \ge {det ^{2m}}\left\{ {{A^{\rm{T}}} - PB{R^{ - 1}}{B^{\rm{T}}}} \right\}{det ^{2mn + 2n}}\left\{ {{\Phi _{11}}\left( {{\tau _\infty }} \right)} \right\}\overline \alpha  > 0,
\end{array}
\end{equation}
whence the validity of the implication $\Psi \left( t \right) \in {\rm{FE}} \Rightarrow \varphi \left( t \right) \in {\rm{FE}} \Rightarrow \Delta \left( t \right) \in {\rm{FE}}$ follows.

The necessity to have a finite degree of the Taylor polynomial to implement the series expansion is shown by the contrary:
\begin{equation}\label{eqA16}
\begin{array}{c}
\mathop {{\rm{lim}}}\limits_{p \to \infty } \int\limits_{t_r^ + }^{{t_e}} {{\Delta ^2}\left( \tau  \right)} {\rm{}\;}d\tau  
 = {det ^{2m}}\left\{ {{A^{\rm{T}}} - PB{R^{ - 1}}{B^{\rm{T}}}} \right\}{det ^{2mn + 2n}}\left\{ {{\Phi _{11}}\left( {{\tau _\infty }} \right)} \right\}\mathop {{\rm{lim}}}\limits_{p \to \infty } \int\limits_{t_r^ + }^{{t_e}} {{\varphi ^{\left( {m + 1} \right)4p{n^2} + 4mn + 2n{\rm{.}}}}\left( \tau  \right)} {\rm{}}\;d\tau .
\end{array}
\end{equation}

As $\varphi \left( t \right) \in \left[ {0{\rm{;\;1}}} \right)$ by virtue of the definition \eqref{eq24}, then, when $p \to \infty $, it holds that $\mathop {{\rm{lim}}}\limits_{p \to \infty } \int\limits_{t_r^ + }^{{t_e}} {{\Delta ^2}\left( \tau  \right)} {\rm{}}\;d\tau  = 0$. So, it is necessary to have $p \in {L_\infty }$ to prove that $\Psi \left( t \right) \in {\rm{FE}} \Rightarrow \Delta \left( t \right) \in {\rm{FE}}$. This completes the proof of the Proposition.

\section{Proof of Proposition 5.}

It follows from the definition of ${\varepsilon _\theta }\left( t \right)$ \eqref{eq19} that, if ${\varepsilon _{{K_r}}}\left( t \right){\rm{,}}\;{\varepsilon _{{K_x}}}\left( t \right)$ and $\overline \Delta \left( t \right)$ are bounded, then the disturbance ${\varepsilon _\theta }\left( t \right)$ is bounded $\left\| {{\varepsilon _\theta }\left( t \right)} \right\| \le \varepsilon _\theta ^{UB}$. When ${\tau _\infty } \in {L_\infty }$, the regressor $\overline \Delta \left( t \right){\rm{:}} = {{blockdiag}}\left\{ {{\Delta _{{K_x}}}\left( t \right){I_{n \times n}}{\rm{,}\;}{\Delta _{{K_r}}}\left( t \right){I_{m \times m}}} \right\}$ is bounded because, following the proof of Proposition \ref{proposition4}, when ${\tau _\infty } \in {L_\infty }$,  then ${\Delta _{{K_x}}}\left( t \right)$ and ${\Delta _{{K_r}}}\left( t \right)$ are bounded. In its turn, it follows from the definitions \eqref{eqA6} and \eqref{eqA7} that ${\varepsilon _{{K_r}}}\left( t \right){\rm{,}}\;{\varepsilon _{{K_x}}}\left( t \right)$ are bounded if $A{\rm{,}}\;B{\rm{,}}\;R{\rm{,}}\;{B_r}{\rm{,}}\;{L_1}{\rm{,}}\;{L_2}$ and $\Phi \left( {{\tau _\infty }} \right){\rm{,}\;}\varepsilon $ are bounded.

The matrices $A$ and $B$ are bounded because of the fact that the pair $\left( {A{\rm{,}}\;B} \right)$ is controllable, while ${L_1}{\rm{,}}\;{L_2}$  and $R{\rm{,}}\;{B_r}$ are bounded according to their definitions. The matrix exponential $\Phi \left( {{\tau _\infty }} \right) = {e^{D{\tau _\infty }}}$ is bounded if $D$ and ${\tau _\infty }$ are bounded. The matrix $D$ is bounded as $A$, $B$ and $Q$, $R$ are bounded. The Taylor-series-based approximation error $\varepsilon $ is bounded following the proof of Proposition \ref{proposition3}. Then $\left\| {{\varepsilon _\theta }\left( t \right)} \right\| \le \varepsilon _\theta ^{UB}$ when ${\tau _\infty } \in {L_\infty }$, as was to be proved in the statement (a).

As, according to \eqref{eq30} and because of $\mathop {\lim }\limits_{p \to \infty } {\varepsilon _{\max }} = 0$ and $\varphi \left( t \right) \in \left[ {0{\rm{;\;1}}} \right)$, the following is true for the disturbances ${\varepsilon _{{\Phi _{11}}}}\left( t \right)$ and ${\varepsilon _{{\Phi _{21}}}}\left( t \right)$:
\begin{equation}\label{eqA17}
\begin{array}{c}
\mathop {\lim }\limits_{p \to \infty } {\varepsilon _{{\Phi _{11}}}}\left( t \right) = \mathop {\lim }\limits_{p \to \infty } \left( { - {\varphi ^{2p}}\left( t \right)L_1^{\rm{T}}\varepsilon {L_1}} \right) = 0,\\
\mathop {\lim }\limits_{p \to \infty } {\varepsilon _{{\Phi _{21}}}}\left( t \right) = \mathop {\lim }\limits_{p \to \infty } \left( { - {\varphi ^{2p}}\left( t \right)L_2^{\rm{T}}\varepsilon {L_2}} \right) = 0,
\end{array}
\end{equation}
then, following the proof of Proposition \ref{proposition2}, it is easy to check the correctness of the equality $\mathop {\lim }\limits_{p \to \infty } {\varepsilon _\theta } = \mathop {\lim }\limits_{p \to \infty } \varepsilon _\theta ^{UB} = 0$. This completes the proof of the proposition.

\section {Proof of Theorem.} 

(a) The following function is considered to prove Theorem:
\begin{equation}\label{eqA18}
{V_{\tilde \theta }} = tr\left( {{{\tilde \theta }^{\rm{T}}}\tilde \theta } \right)
\end{equation}

Following Proposition \ref{proposition6}, the derivative of \eqref{eqA18} with respect to \eqref{eq32} $\forall t \ge {t_e}$ is written as:
\begin{equation}\label{eqA19}
\begin{array}{c}
{{\dot V}_{\tilde \theta }} = 2tr\left( { - {{\tilde \theta }^{\rm{T}}}\gamma {\Omega ^2}\tilde \theta  + {{\tilde \theta }^{\rm{T}}}\gamma \Omega \varsigma } \right) = 2tr\left\{ { - {{\tilde \theta }^{\rm{T}}}{\textstyle{{\left( {{\gamma _0}{\lambda _{{\rm{max}}}}\left( {\omega {\omega ^{\rm{T}}}} \right) + {\gamma _1}} \right){\Omega ^2}} \over {{\Omega ^2}}}}\tilde \theta}\right.
\left.{+ {{\tilde \theta }^{\rm{T}}}{\textstyle{{\left( {{\gamma _0}{\lambda _{{\rm{max}}}}\left( {\omega {\omega ^{\rm{T}}}} \right) + {\gamma _1}} \right)\Omega } \over {{\Omega ^2}}}}\varsigma } \right\} \le \\
 \le  - 2tr\left\{ {{{\tilde \theta }^{\rm{T}}}\left( {{\gamma _0}{\lambda _{{\rm{max}}}}\left( {\omega {\omega ^{\rm{T}}}} \right) + {\gamma _1}} \right)\tilde \theta  - {{\tilde \theta }^{\rm{T}}}{\textstyle{{{\gamma _0}{\lambda _{{\rm{max}}}}\left( {\omega {\omega ^{\rm{T}}}} \right) + {\gamma _1}} \over \rho }}\varsigma } \right\}
 \le  - 2\left( {{\gamma _0}{\lambda _{{\rm{max}}}}\left( {\omega {\omega ^{\rm{T}}}} \right) + {\gamma _1}} \right){\left\| {\tilde \theta } \right\|^2} + 2{\textstyle{{{\gamma _0}{\lambda _{{\rm{max}}}}\left( {\omega {\omega ^{\rm{T}}}} \right) + {\gamma _1}} \over \rho }}{\varsigma _{{\rm{UB}}}}\left\| {\tilde \theta } \right\|.
\end{array}
\end{equation}

Assuming that $a = \sqrt {{\gamma _0}{\lambda _{{\rm{max}}}}\left( {\omega {\omega ^{\rm{T}}}} \right) + {\gamma _1}} \left\| {\tilde \theta } \right\|{\rm{,}}$ $b = {\textstyle{{{\varsigma _{{\rm{UB}}}}\sqrt {{\gamma _0}{\lambda _{{\rm{max}}}}\left( {\omega {\omega ^{\rm{T}}}} \right) + {\gamma _1}} } \over \rho }}{\rm{,}}$ and applying the inequality $ - {a^2} + ab \le  - {\textstyle{1 \over 2}}{a^2} + {\textstyle{1 \over 2}}{b^2}$, it is obtained from \eqref{eqA19} that:
\begin{equation}\label{eqA20}
{\dot V_{\tilde \theta }}\left( t \right) \le  - \underbrace {\left( {{\gamma _0}{\lambda _{{\rm{max}}}}\left( {\omega \left( t \right){\omega ^{\rm{T}}}\left( t \right)} \right) + {\gamma _1}} \right)}_{{\kappa _1}\left( t \right)}{\left\| {\tilde \theta } \right\|^2} + {\textstyle{{{\kappa _1}\left( t \right)} \over {{\rho ^2}}}}\varsigma _{{\rm{UB}}}^2
\end{equation}

The general solution of the differential equation with time-varying parameters \eqref{eqA20} $\forall t \ge {t_e}$ is written as:
\begin{equation}\label{eqA21}
{V_{\tilde \theta }}\left( t \right) \le {e^{ - \int\limits_{{t_e}}^t {{\kappa _1}\left( \tau  \right)d\tau } }}V\left( {{t_e}} \right) + {e^{ - \int\limits_{{t_e}}^t {{\kappa _1}\left( \tau  \right)d\tau } }}\int\limits_{{t_e}}^t {{e^{\int\limits_{t_r^+}^\tau  {{\kappa _1}\left( s \right)ds} }}{\textstyle{{{\kappa _1}\left( {{\tau _1}} \right)\varsigma _{{\rm{UB}}}^2} \over {{\rho ^2}}}}d{\tau _1}}.
\end{equation}

Since that point two different situations are possible: ${\kappa _1}\left( t \right) \in {L_\infty }$ and ${\kappa _1}\left( t \right) \notin {L_\infty }$. Let both of them be considered.

As $x\left( t \right)$ and $r$ are continuous, then the functions ${\lambda _{{\rm{max}}}}\left( {\omega \left( t \right){\omega ^{\rm{T}}}\left( t \right)} \right)$ and ${\kappa _1}\left( t \right)$ are continuous as well. Then, when ${\kappa _1}\left( t \right) \notin {L_\infty }$, the following exponential upper bound of ${\kappa _1}\left( t \right)$ holds:
\begin{displaymath}
{\kappa _1}\left( t \right) \le {\overline a_0}{e^{{{\overline a}_1}t}}{\rm{,}\;}{\overline a_0} > 0,{\rm{}\;}{\overline a_1} > 0.
\end{displaymath}

Considering ${\kappa _1}\left( t \right) \le {\overline a_0}{e^{{{\overline a}_1}t}}$, the inequality \eqref{eqA21} is solved to obtain:
\begin{equation}\label{eqA22}
{V_{\tilde \theta }}\left( t \right) \le {e^{ - \frac{{{{\overline a}_0}}}{{{{\overline a}_1}}}{e^{{a_1}\left( {t - {t_e}} \right)}}}}V\left( {{t_e}} \right) + {\textstyle{{\varsigma _{{\rm{UB}}}^2} \over {{\rho ^2}}}} \le {e^{ - {{\overline a}_0}\overline a_1^{ - 1}\left( {t - {t_e}} \right)}}V\left( {{t_e}} \right) + {\textstyle{{\varsigma _{{\rm{UB}}}^2} \over {{\rho ^2}}}}.
\end{equation}

If ${\kappa _1}\left( t \right) \in {L_\infty }$, then it holds that ${\lambda _{{\rm{max}}}}\left( {\omega \left( t \right){\omega ^{\rm{T}}}\left( t \right)} \right) \le \lambda _{{\rm{max}}}^{{\rm{UB}}}$, and the solution of \eqref{eqA21} is written as:
\begin{equation}\label{eqA23}
{V_{\tilde \theta }}\left( t \right) \le {e^{ - \left( {{\gamma _0}\lambda _{{\rm{max}}}^{{\rm{UB}}} + {\gamma _1}} \right)\left( {t - {t_e}} \right)}}V\left( {{t_e}} \right) + {\textstyle{{\varsigma _{{\rm{UB}}}^2} \over {{\rho ^2}}}}.
\end{equation}

Thus, regardless of the boundedness of ${\lambda _{{\rm{max}}}}\left( {\omega \left( t \right){\omega ^{\rm{T}}}\left( t \right)} \right)$ and ${\kappa _1}\left( t \right)$, the error $\tilde \theta \left( t \right)$ exponentially converges to a set with a bound:
\begin{equation}\label{eqA24}
\forall t \ge {t_e}{\rm{:\;}}\left\| {\tilde \theta \left( t \right)} \right\| \le {e^{ - 0.5{\eta _{\tilde \theta }}\left( {t - {t_e}} \right)}}\left\| {\tilde \theta \left( {{t_e}} \right)} \right\| + {\textstyle{{{\varsigma _{{\rm{UB}}}}} \over \rho }}{\rm{,}}
\end{equation}
where ${\eta _{\tilde \theta }} \ge {\rm{min}}\left\{ {{\textstyle{{{{\overline a}_0}} \over {{{\overline a}_1}}}}{\rm{,}\;}{\gamma _0}\lambda _{{\rm{max}}}^{{\rm{UB}}} + {\gamma _1}} \right\}$ is a lower bound of the error $\tilde \theta \left( t \right)$ convergence rate.

The next step is to analyze the properties of the augmented tracking error $\xi \left( t \right)$.

As ${A_{ref}}$ is a Hurwitz matrix and $B$ is of full column rank, then according to the corollary of the KYP lemma \eqref{eq2}, there exist $M{\rm{,}}\;N{\rm{,}}\;D{\rm{,}}\;K$, which satisfy \eqref{eq2}. Then the quadratic function to analyze the stability of \eqref{eq12} is chosen:
\begin{equation}\label{eqA25}
\begin{array}{c}
{V_\xi } = e_{ref}^{\rm{T}}M{e_{ref}} + {\textstyle{1 \over 2}}tr\left( {{{\tilde \theta }^{\rm{T}}}\tilde \theta } \right){\rm{,}\;}H = {{blockdiag}}\left\{ {M{\rm{,}\;}{\textstyle{1 \over 2}}I} \right\}\\
\underbrace {{\lambda _{{\rm{min}}}}\left( H \right)}_{{\lambda _{\mathop{\rm m}\nolimits} }}{\left\| \xi  \right\|^2} \le V\left( {\left\| \xi  \right\|} \right) \le \underbrace {{\lambda _{{\rm{max}}}}\left( H \right)}_{{\lambda _M}}{\left\| \xi  \right\|^2}{\rm{,}}
\end{array}
\end{equation}

Applying $tr\left( {AB} \right) = BA$, the derivative of \eqref{eqA25} with respect to \eqref{eq12}, \eqref{eq32} is written as:
\begin{equation}\label{eqA26}
\begin{array}{c}
{{\dot V}_\xi } = e_{ref}^{\rm{T}}\left( {A_{ref}^{\rm{T}}M + M{A_{ref}}} \right){e_{ref}} + 2e_{ref}^{\rm{T}}MB{{\tilde \theta }^{\rm{T}}}\omega
-tr\left( {{{\tilde \theta }^{\rm{T}}}\gamma {\Omega ^2}\tilde \theta  + {{\tilde \theta }^{\rm{T}}}\gamma \Omega \varsigma } \right) = \\
 =  - \mu e_{ref}^{\rm{T}}M{e_{ref}} + tr\left( { - {K^{\rm{T}}}{N^{\rm{T}}}{e_{ref}}e_{ref}^{\rm{T}}NK + 2{{\tilde \theta }^{\rm{T}}}\omega e_{ref}^{\rm{T}}NK}\right.
 \left.{-{{\tilde \theta }^{\rm{T}}}\left[ {\gamma {\Omega ^2}\tilde \theta  + \gamma \Omega \varsigma } \right]} \right).
\end{array}
\end{equation}

Without the loss of generality, $D$ is chosen so as $K{K^{\rm{T}}} = {K^{\rm{T}}}K = {I_{m \times m}}$. Then, completing the square in \eqref{eqA26}, it is obtained:
\begin{equation}\label{eqA27}
\begin{array}{c}
{{\dot V}_\xi } =  - \mu e_{ref}^{\rm{T}}M{e_{ref}} + tr\left( { - {K^{\rm{T}}}{N^{\rm{T}}}{e_{ref}}e_{ref}^{\rm{T}}NK + 2{{\tilde \theta }^{\rm{T}}}\omega e_{ref}^{\rm{T}}NK} \right.
\left. { \pm {{\tilde \theta }^{\rm{T}}}\omega {\omega ^{\rm{T}}}\tilde \theta  - {{\tilde \theta }^{\rm{T}}}\left[ {\gamma {\Omega ^2}\tilde \theta  + \gamma \Omega \varsigma } \right]} \right) =  - \mu e_{ref}^{\rm{T}}M{e_{ref}} +\\
+tr\left\{ { - {{\left( {e_{ref}^{\rm{T}}NK - {{\tilde \theta }^{\rm{T}}}\omega } \right)}^2} + {{\tilde \theta }^{\rm{T}}}\omega {\omega ^{\rm{T}}}\tilde \theta  - {{\tilde \theta }^{\rm{T}}}\left[ {\gamma {\Omega ^2}\tilde \theta  + \gamma \Omega \varsigma } \right]} \right\}
 \le  - \mu e_{ref}^{\rm{T}}M{e_{ref}} + tr\left\{ {{{\tilde \theta }^{\rm{T}}}\omega {\omega ^{\rm{T}}}\tilde \theta  - {{\tilde \theta }^{\rm{T}}}\left[ {\gamma {\Omega ^2}\tilde \theta  + \gamma \Omega \varsigma } \right]} \right\}.
\end{array}
\end{equation}

Two different cases are to be considered: $t < {t_e}{\rm{;}}\;t \ge {t_e}$.

Firstly, let $t < {t_e}$. In accordance with Proposition \ref{proposition6}, in the worst case scenario $\forall t < {t_e}$ it holds that $\Omega \left( t \right) < \rho  \Leftrightarrow \gamma \Omega \left( t \right) = 0$ and $\left\| {\tilde \theta \left( t \right)} \right\| = \left\| {\tilde \theta \left( {t_r^ + } \right)} \right\|$. Then $\forall t < {t_e}$ the equation \eqref{eqA27} is rewritten as:
\begin{equation}\label{eqA28}
\begin{array}{c}
{{\dot V}_\xi } \le  - \mu e_{ref}^{\rm{T}}M{e_{ref}} + tr\left\{ {{{\tilde \theta }^{\rm{T}}}\left( {t_r^ + } \right)\omega {\omega ^{\rm{T}}}\tilde \theta \left( {t_r^ + } \right) \pm {{\tilde \theta }^{\rm{T}}}\tilde \theta } \right\} \le \\
 \le  - \mu e_{ref}^{\rm{T}}M{e_{ref}} - tr\left\{ {{{\tilde \theta }^{\rm{T}}}\tilde \theta } \right\} + tr\left\{ {{{\tilde \theta }^{\rm{T}}}\left( {t_r^ + } \right)\omega {\omega ^{\rm{T}}}\tilde \theta \left( {t_r^ + } \right) + {{\tilde \theta }^{\rm{T}}}\left( {t_r^ + } \right)\tilde \theta \left( {t_r^ + } \right)} \right\}.
\end{array}
\end{equation}

Let the notion of the maximum eigenvalue of $\omega \left( t \right){\omega ^{\rm{T}}}\left( t \right)$ over $\left[ {t_r^+{\rm{;}}\;{t_e}} \right)$ be introduced:
\begin{equation}\label{eqA29}
\delta  = {\rm{sup}}\mathop {{\rm{max}}}\limits_{\forall t < {t_e}} {\lambda _{{\rm{max}}}}\left( {\omega \left( t \right){\omega ^{\rm{T}}}\left( t \right)} \right) \le {\rm{max}}\left\{ {{{\overline a}_0}{e^{{{\overline a}_1}{t_e}}}{\rm{,}\;}\lambda _{{\rm{max}}}^{{\rm{UB}}}} \right\}.
\end{equation}

Taking into consideration \eqref{eqA29}, the equation \eqref{eqA28} for $t < {t_e}$ is rewritten as:
\begin{equation}\label{eqA30}
{\dot V_\xi } \le  - \mu {\lambda _{{\rm{min}}}}\left( M \right){\left\| {{e_{ref}}} \right\|^2} - {\left\| {\tilde \theta } \right\|^{\rm{2}}} + \left( {\delta  + 1} \right){\left\| {\tilde \theta \left( {t_r^ + } \right)} \right\|^2} \le  - {\eta _{{\xi _1}}}{V_\xi } + {r_B}{\rm{,}}
\end{equation}
where ${\eta _{{\xi _1}}} = {\rm{min}}\left\{ {{\textstyle{{\mu {\lambda _{{\rm{min}}}}\left( M \right)} \over {{\lambda _{{\rm{max}}}}\left( M \right)}}}{\rm{;\;2}}} \right\}{\rm{;}}\;{r_B} = \left( {\delta  + 1} \right){\left\| {\tilde \theta \left( {t_r^ + } \right)} \right\|^2}$.

Having solved \eqref{eqA30}, it is obtained:
\begin{equation}\label{eqA31}
\forall t < {t_e}{\rm{:}}\;{V_\xi }\left( t \right) \le {e^{ - {\eta _{{\xi _1}}}\left( {t - t_r^ + } \right)}}{V_\xi }\left( {t_r^ + } \right) + {\textstyle{{{r_B}} \over {{\eta _{{\xi _1}}}}}}.
\end{equation}

As ${\lambda _{\mathop{\rm m}\nolimits} }{\left\| {\xi \left( t \right)} \right\|^2} \le {V_\xi }\left( t \right)$ and ${V_\xi }\left( {t_r^ + } \right) \le {\lambda _M}{\left\| {\xi \left( {t_r^ + } \right)} \right\|^2}$, then $\forall t < {t_e}$ the estimate of $\xi $ is obtained from \eqref{eqA31}:
\begin{equation}\label{eqA32}
\left\| {\xi \left( t \right)} \right\| \le \sqrt {{\textstyle{{{\lambda _M}} \over {{\lambda _m}}}}{e^{ - {\eta _{\rm{1}}}\left( {t - t_r^ + } \right)}}{{\left\| {\xi \left( {t_r^ + } \right)} \right\|}^2} + {\textstyle{{{r_B}} \over {{\lambda _m}{\eta _{\rm{1}}}}}}}  \le \sqrt {{\textstyle{{{\lambda _M}} \over {{\lambda _m}}}}{{\left\| {\xi \left( {t_r^ + } \right)} \right\|}^2} + {\textstyle{{{r_B}} \over {{\lambda _m}{\eta _{\rm{1}}}}}}}.
\end{equation}

So $\xi \left( t \right)$ is bounded $\forall t < {t_e}$.

Secondly, let $t \ge {t_e}$. Taking into account \eqref{eq32} and the facts that, if Assumption \ref{assumption1} holds, $\Psi \left( t \right) \in {\rm{FE}}$ and $p \in {L_\infty }$ $\forall t \ge {t_e}{\rm{}}$ \linebreak $0 < {\Omega _{{\rm{LB}}}} \le \Omega\left(t\right)  \le {\Omega _{{\rm{UB}}}}$ are true, it is obtained from \eqref{eqA27} that $\forall t \ge {t_e}$:
\begin{equation}\label{eqA33}
\begin{array}{c}
{{\dot V}_\xi } \le  - \mu e_{ref}^{\rm{T}}P{e_{ref}}
+tr\left\{ {{{\tilde \theta }^{\rm{T}}}\omega {\omega ^{\rm{T}}}\tilde \theta  - {{\tilde \theta }^{\rm{T}}}\left[ {{\textstyle{{\left( {{\gamma _0}{\lambda _{{\rm{max}}}}\left( {\omega {\omega ^{\rm{T}}}} \right) + {\gamma _1}} \right){\Omega ^2}} \over {{\Omega ^2}}}}\tilde \theta  + {\textstyle{{\left( {{\gamma _0}{\lambda _{{\rm{max}}}}\left( {\omega {\omega ^{\rm{T}}}} \right) + {\gamma _1}} \right)\Omega } \over {{\Omega ^2}}}}\varsigma } \right]} \right\} \le \\
 \le  - \mu e_{ref}^{\rm{T}}P{e_{ref}}
 +tr\left\{ {{{\tilde \theta }^{\rm{T}}}\omega {\omega ^{\rm{T}}}\tilde \theta  - {{\tilde \theta }^{\rm{T}}}\left( {{\gamma _0}{\lambda _{{\rm{max}}}}\left( {\omega {\omega ^{\rm{T}}}} \right) + {\gamma _1}} \right)\tilde \theta  - {{\tilde \theta }^{\rm{T}}}\left( {{\textstyle{{\left( {{\gamma _0}{\lambda _{{\rm{max}}}}\left( {\omega {\omega ^{\rm{T}}}} \right) + {\gamma _1}} \right)} \over \rho }}\varsigma } \right)} \right\}.
\end{array}
\end{equation}

It holds for any possible unbounded $\omega \left( t \right)$ that:
\begin{equation}\label{eqA34}
\begin{array}{c}
tr\left\{ {{{\tilde \theta }^{\rm{T}}}\left( t \right)\left( {\omega \left( t \right)\omega {{\left( t \right)}^{\rm{T}}} - \left[ {{\gamma _0}{\lambda _{{\rm{max}}}}\left( {\omega \left( t \right){\omega ^{\rm{T}}}\left( t \right)} \right) + {\gamma _1}} \right]{I_{\left( {m + n} \right) \times \left( {m + n} \right)}}} \right)\tilde \theta \left( t \right)} \right\}
\le - {\gamma _1}tr\left\{ {{{\tilde \theta }^{\rm{T}}}\left( t \right)\tilde \theta \left( t \right)} \right\}.
\end{array}
\end{equation}

So \eqref{eqA33} is written as:
\begin{equation}\label{eqA35}
\begin{array}{c}
{{\dot V}_\xi } \le  - \mu e_{ref}^{\rm{T}}P{e_{ref}} - {\gamma _1}tr\left\{ {{{\tilde \theta }^{\rm{T}}}\tilde \theta  + {{\tilde \theta }^{\rm{T}}}{\textstyle{{{\gamma _0}{\lambda _{{\rm{max}}}}\left( {\omega {\omega ^{\rm{T}}}} \right) + {\gamma _1}} \over \rho }}\varsigma } \right\}
 \le  - \mu {\lambda _{{\rm{min}}}}\left( P \right){\left\| {{e_{ref}}} \right\|^2} - {\gamma _1}{\left\| {\tilde \theta } \right\|^{\rm{2}}} + {\textstyle{{{\gamma _0}{\lambda _{{\rm{max}}}}\left( {\omega {\omega ^{\rm{T}}}} \right) + {\gamma _1}} \over \rho }}{\varsigma _{{\rm{UB}}}}\left\| {\tilde \theta } \right\|.
\end{array}
\end{equation}

Assuming $a = \sqrt {{\gamma _1}} \left\| {\tilde \theta } \right\|{\rm{,}}$ $b = {\textstyle{{{\gamma _0}{\lambda _{{\rm{max}}}}\left( {\omega {\omega ^{\rm{T}}}} \right) + {\gamma _1}} \over {\rho \sqrt {{\gamma _1}} }}}{\varsigma _{{\rm{UB}}}}{\rm{,}}$ and applying the inequality $ - {a^2} + ab \le  - {\textstyle{1 \over 2}}{a^2} + {\textstyle{1 \over 2}}{b^2}$, it follows from \eqref{eqA35} that:
\begin{equation}\label{eqA36}
\begin{array}{c}
{{\dot V}_\xi } \le  - \mu {\lambda _{{\rm{min}}}}\left( P \right){\left\| {{e_{ref}}} \right\|^2} - {\textstyle{{{\gamma _1}} \over 2}}{\left\| {\tilde \theta } \right\|^{\rm{2}}} + {\textstyle{1 \over {2{\gamma _1}}}}{\left( {{\textstyle{{{\gamma _0}{\lambda _{{\rm{max}}}}\left( {\omega {\omega ^{\rm{T}}}} \right) + {\gamma _1}} \over \rho }}{\varsigma _{{\rm{UB}}}}} \right)^2}
 \le  - {\eta _{{\xi _2}}}V + {\textstyle{1 \over {2{\gamma _1}}}}{\left( {{\textstyle{{{\gamma _0}{\lambda _{{\rm{max}}}}\left( {\omega {\omega ^{\rm{T}}}} \right) + {\gamma _1}} \over \rho }}{\varsigma _{{\rm{UB}}}}} \right)^2}
\end{array}
\end{equation}
where ${\eta _{\xi {\rm{2}}}} = {\rm{min}}\left\{ {{\textstyle{{\mu {\lambda _{{\rm{min}}}}\left( P \right)} \over {{\lambda _{{\rm{max}}}}\left( P \right)}}}{\rm{;}}\;{\gamma _1}} \right\}.$

The inequality \eqref{eqA36} is solved, and it is obtained for $t \ge {t_e}$ that:
\begin{equation}\label{eqA37}
\begin{array}{c}
{V_\xi }\left( t \right) \le {e^{ - {\eta _{{\xi _2}}}\left( {t - {t_e}} \right)}}V\left( {{t_e}} \right) + \int\limits_{{t_e}}^t {{e^{ - {\eta _2}\left( {t - \tau } \right)}}{\textstyle{1 \over {2{\eta _{{\xi _2}}}{\gamma _1}}}}{{\left( {{\textstyle{{{\gamma _0}{\lambda _{{\rm{max}}}}\left( {\omega \left( \tau  \right){\omega ^{\rm{T}}}\left( \tau  \right)} \right) + {\gamma _1}} \over \rho }}{\varsigma _{{\rm{UB}}}}} \right)}^2}d\tau }  \le \\
 \le {e^{ - {\eta _{{\xi _2}}}\left( {t - {t_e}} \right)}}V\left( {{t_e}} \right) + \underbrace {{\textstyle{{\varsigma _{{\rm{UB}}}^2} \over {2{\eta _{{\xi _2}}}{\rho ^2}{\gamma _1}}}}\int\limits_{{t_e}}^t {{e^{ - {\eta _{{\xi _2}}}\left( {t - \tau } \right)}}{{\left( {{\gamma _0}{\lambda _{{\rm{max}}}}\left( {\omega \left( \tau  \right){\omega ^{\rm{T}}}\left( \tau  \right)} \right) + {\gamma _1}} \right)}^2}d\tau } }_{\varepsilon _\xi ^2}.
\end{array}
\end{equation}

For further proof of Theorem, it is to be shown that ${\varepsilon _\xi } \in {L_\infty }$. So, we need to prove that $\forall t \ge {t_e}{\rm{}\;}{\lambda _{{\rm{max}}}}\left( {\omega \left( t \right){\omega ^{\rm{T}}}\left( t \right)} \right) \le \lambda _{{\rm{max}}}^{{\rm{UB}}}$.

Following the definition $\omega \left( t \right) = {{\begin{bmatrix}
{{x^{\rm{T}}}\left( t \right)}&{{r^{\rm{T}}}}\end{bmatrix}}^{\rm{T}}}$, we have that ${\lambda _{{\rm{max}}}}\left( {\omega \left( t \right){\omega ^{\rm{T}}}\left( t \right)} \right) \le \lambda _{{\rm{max}}}^{{\rm{UB}}}$ if $x\left( t \right)$ and $r$ are bounded. The reference signal $r$ is bounded according to the problem statement. So, the boundedness of the state vector $x\left( t \right)$ is to be shown.

To do that, the following quadratic function is considered:
\begin{equation}\label{eqA38}
{V_{{e_{ref}}}} = e_{ref}^{\rm{T}}\overline P{e_{ref}}
\end{equation}

The derivative of \eqref{eqA38} with respect to \eqref{eq12}, \eqref{eq32} is written as:
\begin{equation}\label{eqA39}
\begin{array}{c}
{{\dot V}_{{e_{ref}}}} = e_{ref}^{\rm{T}}\left( {A_{ref}^{\rm{T}}\overline P + \overline P{A_{ref}}} \right){e_{ref}} + 2e_{ref}^{\rm{T}}\overline PB{{\tilde \theta }^{\rm{T}}}\omega 
 =  - e_{ref}^{\rm{T}}Q{e_{ref}} + 2e_{ref}^{\rm{T}}\overline PB{{\tilde K}_x}x + 2e_{ref}^{\rm{T}}\overline PB{{\tilde K}_r}r \le \\
 \le  - {\lambda _{{\rm{min}}}}\left( {\overline Q} \right){\left\| {{e_{ref}}} \right\|^2} + 2\left\| {{e_{ref}}} \right\|\left\| {\overline PB} \right\|\left\| {\tilde \theta } \right\|\left\| x \right\|
 +2\left\| {{e_{ref}}} \right\|\left\| {\overline PB} \right\|\left\| {\tilde \theta } \right\|{r_{\max }} \le \\
 \le  - {\lambda _{{\rm{min}}}}\left( {\overline Q} \right){\left\| {{e_{ref}}} \right\|^2} + 2{\left\| {{e_{ref}}} \right\|^2}\left\| {\overline PB} \right\|\left\| {\tilde \theta } \right\| 
 +2\left( {\left\| {{x_{ref}}} \right\| + {r_{\max }}} \right)\left\| {{e_{ref}}} \right\|\left\| {\overline PB} \right\|\left\| {\tilde \theta } \right\| \le \\
 \le \left( { - {\lambda _{{\rm{min}}}}\left( {\overline Q} \right) + 2\left\| {\overline PB} \right\|\left\| {\tilde \theta } \right\|} \right){\left\| {{e_{ref}}} \right\|^2}
 +2\left( {\left\| {{x_{ref}}} \right\| + {r_{\max }}} \right)\left\| {{e_{ref}}} \right\|\left\| {\overline PB} \right\|\left\| {\tilde \theta } \right\|
\end{array}
\end{equation}

Following Theorem statement, the disturbance $\varsigma \left( t \right)$ satisfies an inequality:
\begin{equation}\label{eqA40}
\forall t \ge \overline t > {t_e}{\rm{}\;}{\varsigma _{{\rm{UB}}}} < {\textstyle{{\rho {\lambda _{{\rm{min}}}}\left( {\overline Q} \right)} \over {2\left\| {\overline PB} \right\|}}} - \rho {e^{ - 0.5{\eta _{\tilde \theta }}\left( {\overline t - {t_e}} \right)}}\left\| {\tilde \theta \left( {{t_e}} \right)} \right\|.
\end{equation}

Then $\forall t \ge \overline t > {t_e}$ the definition \eqref{eqA40} is substituted into \eqref{eqA24} to obtain:
\begin{equation}\label{eqA41}
\left\| {\tilde \theta \left( t \right)} \right\| \le {e^{ - 0.5{\eta _{\tilde \theta }}\left( {\overline t - {t_e}} \right)}}\left\| {\tilde \theta \left( {{t_e}} \right)} \right\| + {\textstyle{{{\varsigma _{{\rm{UB}}}}} \over \rho }} < {\textstyle{{{\lambda _{{\rm{min}}}}\left( {\overline Q} \right)} \over {2\left\| {\overline PB} \right\|}}}.
\end{equation}

Considering the inequality \eqref{eqA41}, the upper bound of the derivative \eqref{eqA39} is written as:
\begin{equation}\label{eqA42}
\forall t \ge \overline t{\rm{:}}\;{\dot V_{{e_{ref}}}} <  - {\Delta _{\overline Q}}{\left\| {{e_{ref}}} \right\|^2} + {\lambda _{{\rm{min}}}}\left( {\overline Q} \right)\left\| {{e_{ref}}} \right\|\left( {\left\| {{x_{ref}}} \right\| + {r_{\max }}} \right){\rm{,}}
\end{equation}
where ${\Delta _{\overline Q}} = {\lambda _{{\rm{min}}}}\left( {\overline Q} \right) - 2\left\| {\overline PB} \right\|\left\| {\tilde \theta \left( {\overline t} \right)} \right\| > 0$ because of the inequality \eqref{eqA41}.

Assuming $a = \sqrt {{\Delta _{\overline Q}}} \left\| {{e_{ref}}} \right\|{\rm{,}}$ $b = {\textstyle{{{\lambda _{{\rm{min}}}}\left( {\overline Q} \right)\left( {\left\| {{x_{ref}}} \right\| + {r_{\max }}} \right)} \over {\sqrt {{\Delta _{\overline Q}}} }}}{\rm{,}}$ and applying the inequality $ - {a^2} + ab \le  - {\textstyle{1 \over 2}}{a^2} + {\textstyle{1 \over 2}}{b^2}$, it is obtained from \eqref{eqA42} that:
\begin{equation}\label{eqA43}
\begin{array}{c}
{\dot V_{{e_{ref}}}} <  - {\textstyle{1 \over 2}}{\Delta _{\overline Q}}{\left\| {{e_{ref}}} \right\|^2} + {\textstyle{1 \over 2}}{\textstyle{{\lambda _{{\rm{min}}}^2\left( {\overline Q} \right){{\left( {\left\| {{x_{ref}}} \right\| + {r_{\max }}} \right)}^2}} \over {{\Delta _{\overline Q}}}}} =  - {\textstyle{{{\Delta _{\overline Q}}} \over {2{\lambda _{{\rm{max}}}}\left( {\overline P} \right)}}}{V_{{e_{ref}}}}
+{\textstyle{1 \over 2}}{\textstyle{{\lambda _{{\rm{min}}}^2\left( {\overline Q} \right){{\left( {\left\| {{x_{ref}}} \right\| + {r_{\max }}} \right)}^2}} \over {{\Delta _{\overline Q}}}}}.
\end{array}
\end{equation}

Taking into account ${\lambda _{\min }}\left( {\overline P} \right){\left\| {{e_{ref}}\left( t \right)} \right\|^2} \le {V_{{e_{ref}}}}\left( t \right)$, $V\left( {\overline t} \right) \le {\lambda _{\max }}\left( {\overline P} \right){\left\| {{e_{ref}}\left( {\overline t} \right)} \right\|^2}$, the estimate of ${e_{ref}}\left( t \right)$ for $t \ge \overline t$ is obtained from \eqref{eqA43}:
\begin{equation}\label{eqA44}
\left\| {{e_{ref}}\left( t \right)} \right\| < {e^{ - {\textstyle{{{\Delta _{\overline Q}}} \over {4{\lambda _{{\rm{max}}}}\left( {\overline P} \right)}}}\left( {t - \overline t} \right)}}\left\| e_{ref}\left(\overline{t}\right)\right\|  + {\textstyle{{\sqrt {{\lambda _{{\rm{max}}}}\left( {\overline P} \right)} {\lambda _{{\rm{min}}}}\left( {\overline Q} \right)\left( {\left\| {{x_{ref}}} \right\| + {r_{\max }}} \right)} \over {{\Delta _{\overline Q}}}}}{\rm{,}}
\end{equation}

Whence, as ${x_{ref}}\left( t \right)$ is bounded, it follows that $\forall t \ge \overline t{\rm{}}\;{e_{ref}}\left( t \right) \in {L_\infty }$. Since the system \eqref{eq12} is linear and has no finite escape time, then $\forall t \in \left[ {t_r^ + {\rm{;}}\;\overline t} \right]{\rm{}}\;{e_{ref}}\left( t \right) \in {L_\infty }$. Hence, in addition to \eqref{eqA44}, it is true that $\forall t \ge t_r^ + {\rm{}}\;{e_{ref}}\left( t \right) \in {L_\infty }$.

As ${e_{ref}}\left( t \right)$ is bounded, as well as ${x_{ref}}\left( t \right)$, then $x\left( t \right)$ is also bounded. As a result, $\forall t \ge t_r^ + $ ${\lambda _{{\rm{max}}}}\left( {\omega \left( t \right){\omega ^{\rm{T}}}\left( t \right)} \right) \le \lambda _{{\rm{max}}}^{{\rm{UB}}}$ holds.

Then ${\varepsilon _\xi } \in {L_\infty }$, and it satisfies the inequality:
\begin{equation}\label{eqA45}
{\varepsilon _\xi } = {\textstyle{{{\varsigma _{{\rm{UB}}}}\left( {{\gamma _0}\lambda _{{\rm{max}}}^{{\rm{UB}}} + {\gamma _1}} \right)} \over {\rho \sqrt {2{\eta _{{\xi _2}}}{\gamma _1}} }}}\sqrt {\int\limits_{{t_e}}^t {{e^{ - {\eta _{{\xi _2}}}\left( {t - \tau } \right)}}d\tau } }  \le {\textstyle{{{\varsigma _{{\rm{UB}}}}\left( {{\gamma _0}\lambda _{{\rm{max}}}^{{\rm{UB}}} + {\gamma _1}} \right)} \over {\rho {\eta _{{\xi _2}}}\sqrt {2{\gamma _1}} }}}
\end{equation}

Taking into account ${\lambda _{\mathop{\rm m}\nolimits} }{\left\| {\xi \left( t \right)} \right\|^2} \le V\left( t \right)$, $V\left( {{t_e}} \right) \le {\lambda _M}{\left\| {\xi \left( {{t_e}} \right)} \right\|^2}$ and \eqref{eqA45}, the estimate of $\xi \left( t \right)$ for $t \ge {t_e}$ is obtained from \eqref{eqA37}:
\begin{equation}\label{eqA46}
\left\| {\xi \left( t \right)} \right\| \le \sqrt {{\textstyle{{{\lambda _M}} \over {{\lambda _m}}}}{e^{ - {\eta _2}t}}{{\left\| {\xi \left( {{t_e}} \right)} \right\|}^2} + \varepsilon _\xi ^2}  \le \sqrt {{\textstyle{{{\lambda _M}} \over {{\lambda _m}}}}\left( {{\textstyle{{{\lambda _M}} \over {{\lambda _m}}}}{{\left\| {\xi \left( {t_r^ + } \right)} \right\|}^2} + {\textstyle{{{r_B}} \over {{\lambda _m}{\eta _{\rm{1}}}}}}} \right) + \varepsilon _\xi ^2}.
\end{equation}

Hence, it is concluded from \eqref{eqA46} and \eqref{eqA32} that $\forall t \ge t_r^ + {\rm{}}\;\xi \left( t \right) \in {L_\infty }$ and $\left\| {\xi \left( t \right)} \right\|$ converges to ${\varepsilon _\xi }$ exponentially $\forall t \ge {t_e}$ with adjustable by $\gamma_0$, $\gamma_1$ rate. Therefore, following Proposition \ref{proposition5}, when ${\tau _\infty } \in {L_\infty }$, the implication $\mathop {\lim }\limits_{p \to \infty } \varepsilon _\theta ^{UB} = 0 \Rightarrow \mathop {\lim }\limits_{p \to \infty } {\varsigma _{UB}} = 0$ holds. Then, taking into consideration the requirement $p \in {L_\infty }$, the upper bound of ${\varepsilon _\xi }$ can be reduced to its finite limit with the help of choice of the finitely high value of $p$. This completes the proof of the statement (a) of the theorem.

According to the statement (b) of Theorem, it is true that $\varsigma \left( t \right) = o\left( {\Omega \left( t \right)\theta } \right){\rm{,}}$ so it is admissible to consider that \linebreak $\Upsilon \left( t \right) = \Omega \left( t \right)\theta  + o\left( {\Omega \left( t \right)\theta } \right) \buildrel \Delta \over = \Omega \left( t \right)\theta $. Then, as $\Omega \left( t \right) \in \mathbb{R}$, the elementwise solution of the equation \eqref{eq32} is written as:
\begin{equation}\label{eqA47}
{\tilde \theta _i}\left( t \right) = {e^{ - \int\limits_{t_r^ + }^t {\gamma \left( \tau  \right){\Omega ^2}\left( \tau  \right)d\tau } }}{\tilde \theta _i}\left( {t_r^ + } \right).
\end{equation}

As ${\rm{sign}}\left( {\gamma {\Omega ^2}} \right) = {\mathop{\rm const}\nolimits}  > 0$, then it follows from \eqref{eqA47} that $\left| {{{\tilde \theta }_i}\left( {{t_a}} \right)} \right| \le \left| {{{\tilde \theta }_i}\left( {{t_b}} \right)} \right|$ $\forall {t_a} \ge {t_b}$. This completes the proof of the part (b) of the theorem.

\nocite{*}% Show all bib entries - both cited and uncited; comment this line to view only cited bib entries;
\bibliography{Glushchenko}%

\clearpage

\section*{Author Biography}

\begin{biography}
{\includegraphics[width=60pt,height=78pt]{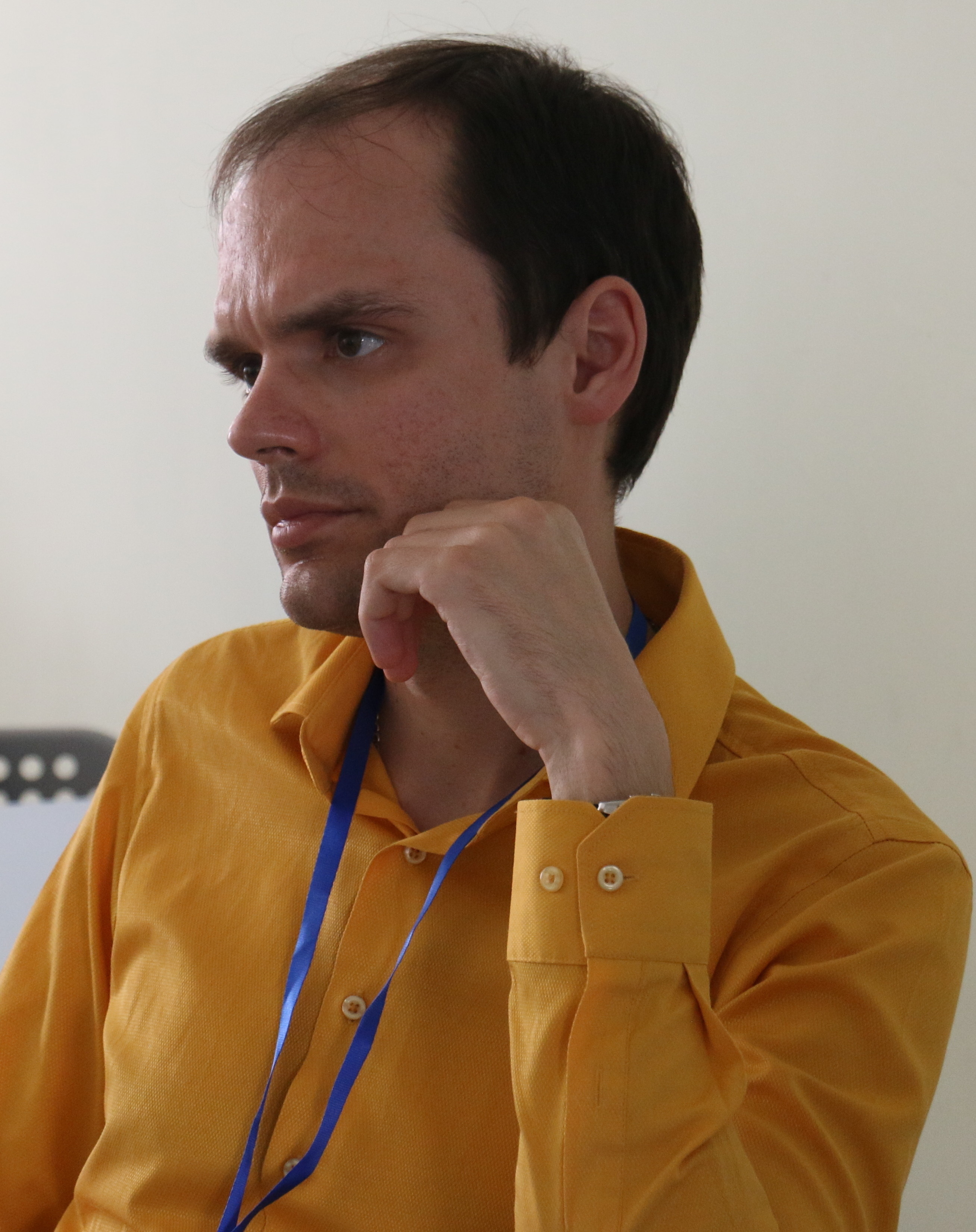}}{\textbf{Anton Glushchenko.} Anton Glushchenko received Software Engineering Degree from National University of Science and Technology "MISIS" (NUST "MISIS", Moscow, Russia) in 2008. In 2009 he gained Candidate of Sciences (Eng.) Degree from NUST "MISIS", in 2021 - Doctor of Sciences (Eng.) Degree from Voronezh State Technical University (Voronezh, Russia). Anton Glushchenko is currently the leading research scientist of laboratory 7 of V.A. Trapeznikov Institute of Control Sciences of Russian Academy of Sciences, Moscow, Russia. His research interests are mainly concentrated on exponentially stable adaptive control of linear time-invariant and time-varying plants under relaxed excitation conditions. The list of his works published includes more than 70 titles.}
\end{biography}

\begin{biography}
{\includegraphics[width=60pt,height=78pt]{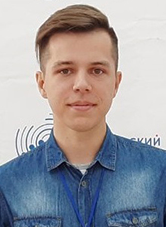}}{\textbf{Konstantin Lastochkin.} Konstantin Lastochkin received his bachelor’s degree in Electrical Engineering from National University of Science and Technology ”MISIS” (NUST ”MISIS”, Moscow, Russia) in 2020, master’s degree in Automation and control of technological processes - from NUST ”MISIS” in 2022. Konstantin
Lastochkin is currently a post-graduate student and a junior research scientist of laboratory 7 of V.A. Trapeznikov Institute of Control Sciences of Russian Academy of Sciences, Moscow, Russia. His present interests include adaptive and robust control, identification theory, nonlinear systems. The list of his works published includes 20 titles.}
\end{biography}

\end{document}